\documentclass[reprint,amsmath,amssymb,aps]{revtex4-2}

\usepackage{graphicx}
\usepackage{dcolumn}
\usepackage{bm}
\usepackage{array}
\usepackage{tabularx}
\usepackage{multirow}

\usepackage{hyperref}
\hypersetup{
  colorlinks=true,
  linkcolor=blue,
  citecolor=blue,
  urlcolor=blue,
  breaklinks=true
}

\usepackage[normalem]{ulem}
\usepackage{comment}
\usepackage{placeins}
\usepackage{enumitem}
\usepackage{relsize}
\usepackage{overpic}
\usepackage{ragged2e}

\usepackage{tikz}
\usetikzlibrary{positioning, arrows.meta}

\usepackage{subcaption}
\usepackage{amsbsy}

\newtheorem{lemma}{Lemma}
\newtheorem{theorem}{Theorem}
\providecommand{\qedsymbol}{\rule{0.7em}{0.7em}}
\newenvironment{proof}{\par\noindent\textit{Proof.}\ }{\hfill\qedsymbol\par}

\usepackage{ragged2e}
\usepackage{algorithm}
\usepackage{algpseudocode}

\begin{document}

\preprint{APS/123-QED}

\title{Mind the Gap: \\Where Analog Ising Machines Cease to Minimize the Ising Hamiltonian}

\author{E.M.~Hasantha~Ekanayake\textsuperscript{1*},
Arvind~R.~Venkatakrishnan\textsuperscript{2*},
Francesco~Bullo\textsuperscript{2},
Nikhil~Shukla\textsuperscript{1\#}}

\affiliation{%
\textnormal{\textsuperscript{}}University of Virginia, Charlottesville, VA, USA\\
\textnormal{\textsuperscript{2}}University of California Santa Barbara, Santa Barbara, CA, USA\\
\textnormal{\textsuperscript{*}Equal Contribution};\;
\textnormal{\textsuperscript{\#}Email: ns6pf@virginia.edu}
}

\begin{abstract}
The design of nonlinear dynamical systems whose gradient flows minimize the Ising Hamiltonian has emerged as a compelling paradigm for realizing Ising machines, forming the foundation of architectures including coherent Ising machines, simulated bifurcation machines, oscillator-based Ising machines, and dynamical Ising machines. Here, we identify a fundamental structural feature shared by these systems—a functional parameter gap defined by the separation between the destabilization of the trivial state and the stabilization of Ising-encoded states. We demonstrate that this separation creates a finite parameter interval in which convergence to an Ising-encoded solution is no longer functionally guaranteed, and the resulting evolution is dictated by the spectral structure of the Jacobian at bifurcation. Subsequently, by introducing a hybrid dynamical framework that reshapes the bifurcation topology, we establish a principled pathway for modulating this parameter gap. The parameter gap thus emerges as a unifying structural principle for the analysis, design and optimization of analog Ising machines.
\end{abstract}

\maketitle

\textit{Introduction}---The pursuit of efficient solutions to combinatorial optimization problems (COPs) remains a central challenge in computing. Ising machines—computational platforms that minimize the Ising Hamiltonian—have emerged as a promising paradigm. The expressive universality of the Ising model enables diverse combinatorial optimization problems to be encoded into a common energy-minimization framework, with solutions identified as ground-state configurations \cite{FBarahona_1982,PhysRevE.86.011116, 10.3389/fphy.2014.00005,10011425}.

A diverse range of Ising machines have been realized across physical domains and computational strategies. Of particular interest are those that harness the analog dynamics of physical systems, offering intrinsic parallelism and energy efficiency. Representative examples include Coherent Ising Machines (CIMs) \cite{PhysRevA.88.063853, articleMarandi,PhysRevLett.126.143901, 10.1063/5.0016140, 2023Bifurcation}, Simulated Bifurcation Machines (SBMs)~\cite{doi:10.1126/sciadv.aav2372}, Oscillator-based Ising Machines (OIMs)~\cite{Wang2021,Chou2019, Vaidya2022}, and the recently proposed Dynamical Ising Machines (DIMs)~\cite{cdc9-y234}. These platforms are distinguished by the nonlinear mechanisms governing their dynamics, ranging from Kuramoto-type phase coupling in OIMs and conjugate coupling in DIMs to effective $\chi^{(2)}$ and Kerr $\chi^{(3)}$ nonlinearities in CIMs and SBMs. The unifying principle across these architectures is the construction of a dynamical manifold whose ground state encodes the optimal Ising configuration.

Although widely referred to as ``Ising machines'', we show in this Letter that these systems do not strictly operate as such across all dynamical regimes. Strikingly, we uncover a previously overlooked regime—where the dynamics deviate from exact minimization of the Ising Hamiltonian—that emerges pervasively across a broad class of regularized Ising gradient flows encompassing DIMs, OIMs, CIMs, and SBMs. We refer to this regime as the \emph{parameter gap}. Most importantly, the system must traverse this parameter gap during operation, which can fundamentally undermine its intended role as a faithful Ising solver. To address this challenge, we introduce a hybrid framework designed to mitigate the effects of the non-Ising dynamical regime.


\textit{Mind the Gap}---To understand the notion of the parameter gap, we consider a general class of regularized phase-encoded Ising gradient flows of the form 
\begin{equation}
    \dot{\phi}_i=
    \sum_jW_{ij}f_{\mathrm{coup}}(\phi_i\pm\phi_j)
    +
    f_{\mathrm{reg}}(\phi_i,\eta)=-\nabla_{\phi_{i}} E(\phi,\eta;W), \label{eq:unified_form}
\end{equation}
where $f_{\mathrm{coup}}$ encodes a continuous relaxation of the spin-spin interaction. The sign in \(\phi_i\pm\phi_j\) is fixed by the model, with \(+\) denoting phase-sum coupling and \(-\) denoting phase-difference coupling. The term $f_{\mathrm{reg}}$ encodes the regularization. In this work, we restrict the phase-flow analysis to the class of models in which the regularization is implemented by a local second-harmonic term, \(f_{\mathrm{reg}}(\phi_i;\eta)=-\eta\sin(2\phi_i)\), which favors the binary phase manifold \(\phi_i\in\{0,\pi\}\). $W \in \mathbb{R}^{N\times N}$ is the weight matrix and \(\eta\) is the regularization parameter. Although the analysis accommodates nonuniform and mixed-sign couplings, as discussed in Supplementary Notes~1, 2, the simulations presented in the main text focus on antiferromagnetic spin interactions. For these simulations, $W_{ij} = -G_{ij}$, where $G$ is the signed adjacency matrix of the underlying graph, defined such that $G_{ij} = -1$ if an edge exists between nodes $i$ and $j$, and $G_{ij} = 0$ otherwise. This setting is relevant to a broad class of problems including MaxCut, graph coloring, and Maximum Independent Set~\cite{10.3389/fphy.2014.00005,Glover2022}. By construction, the minimizers of $E$ coincide with the minimizers of the Ising Hamiltonian. Noise is typically included in practice but is omitted here for clarity.  
\begin{figure*}[t]
  \centering
  {\renewcommand{\arraystretch}{1.5}%
  \setlength\arrayrulewidth{1.5pt}
  \begin{tabular}{|c|w{c}{4.0in}|c|}\hline
 \textbf{{\large Ising}} & 
    \multirow{2}{*}{\textbf{{\large Data}}}  &
     \multirow{2}{*}{\textbf{{\large Simulation}}} \\
    \textbf{{\large machine}} & & \\  \hline
    \multirow{1}{1em}{ \huge  \rotatebox[origin=c]{90}{{DIM}}}& 
     \hspace{-0.35in}\multirow{1}{*}{$\begin{aligned}
        &  \boldsymbol{\eta}: & &\;\; K_{s} \\
        &\boldsymbol{E(\phi, \eta, W)}:& &\;\; \frac{K}{2}\sum\limits_{i=1}^{N}\sum\limits_{j=1}^{N}W_{ij}\cos(\phi_{i} +\! \phi_{j}) -\!\frac{K_{s}}{2}\sum\limits_{i=1}^{N}\cos(2\phi_{i}) \\
        &\boldsymbol{J_f}(\phi^{\ast}, K_s):
        &&\;\;
        K J_{\text{DIM}}(\phi^{\ast})
        -\!
        2K_{s}\operatorname{diag}
         \big(
        \cos(2\phi^{\ast}_{1}),
        \ldots,
        \cos(2\phi_{N}^{\ast})
        \big)
        \\
        &\boldsymbol{\Delta_{\textbf{DIM}}(\mathcal{G})}:& &\;\; \dfrac{K}{2}\lambda_{\max}\big(J_{\text{DIM}}^{\{0,\pi\}}(\phi^{\ast})\big)+\! \frac{K}{2}\lambda_{\max}(-\!{D} -\! {W})
        \end{aligned}$} & 
         \raisebox{-1in}{\includegraphics[height=1.1in]{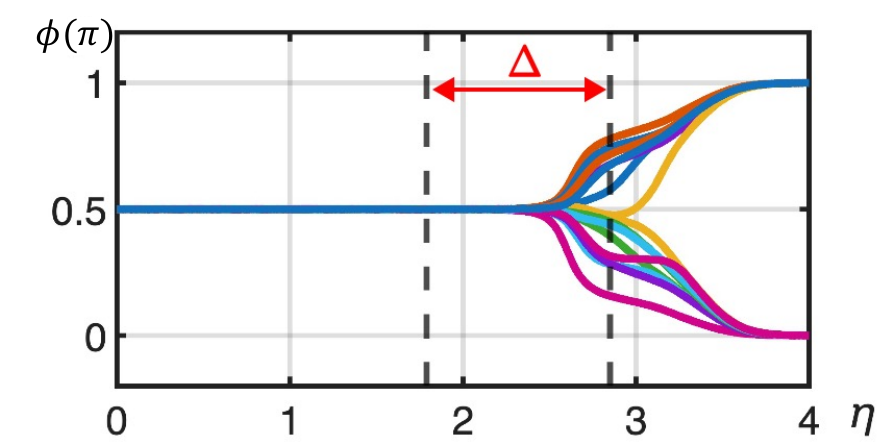}}\\ \hline
        \multirow{1}{1em}{ \huge \rotatebox[origin=c]{90}{{OIM}}}& 
        \hspace{-0.38in} \multirow{1}{*}{$\begin{aligned}
        &\hspace{-0.6mm}\boldsymbol{\eta}: & &\;\; K_{s} \\
        &\hspace{-0.6mm}\boldsymbol{E(\phi, \eta, W)}:& &\;\;  \frac{K}{2}\sum\limits_{i=1}^{N}\sum\limits_{j=1}^{N}W_{ij}\cos(\phi_{i} -\! \phi_{j}) -\! \frac{K_{s}}{2}\sum\limits_{i=1}^{N}\cos(2\phi_{i}) \\ 
        &\hspace{-0.6mm}\boldsymbol{J_f}(\phi^{\ast}, K_s):
         &&\;\;
         K J_{\text{OIM}}(\phi^{\ast})
          -\!
        2K_{s}\operatorname{diag}
        \big(
        \cos(2\phi^{\ast}_{1}),
        \ldots,
        \cos(2\phi_{N}^{\ast})
        \big)
        \\
        &\hspace{-0.6mm}\boldsymbol{\Delta_{\textbf{OIM}}(\mathcal{G})}:& &\;\; \dfrac{K}{2}\lambda_{\max}\big(J_{\text{OIM}}^{\{0,\pi\}}(\phi^{\ast})\big)+\! \frac{K}{2}\lambda_{\max}({D} -\! {W})
        \end{aligned}$}& 
        \raisebox{-1in}{\includegraphics[height=1.1in]{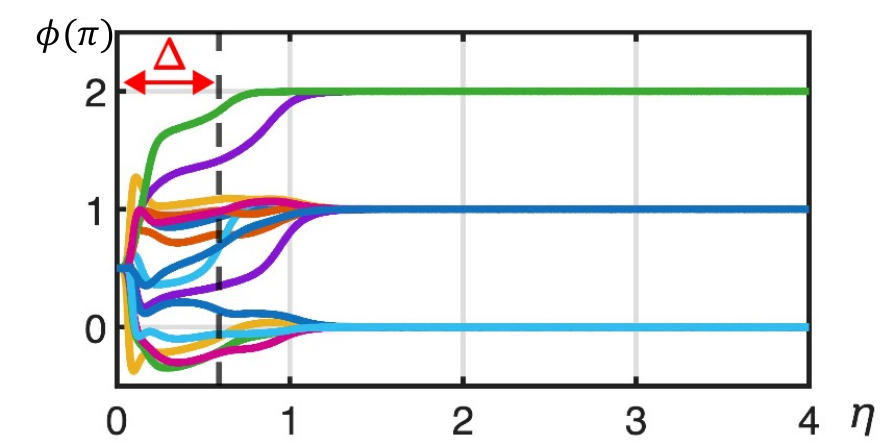}}\\ \hline
        \multirow{1}{1em}{ \huge \rotatebox[origin=c]{90}{{CIM}}}& 
        \hspace{-0.25in}\multirow{1}{*}{$\begin{aligned}
        &\:\boldsymbol{\eta}: & &\;\; p \\
        &\:\boldsymbol{E(x, \eta, W)}:& &\;\; \sum\limits_{i=1}^{N}\left[\left(1 -\! p\right)\dfrac{x_i^{2}}{2} +\! \dfrac{x_{i}^{4}}{4}\right] 
         +\! \dfrac{\xi}{2} \sum\limits_{j=1}^N W_{ij}\,x_j \\ 
         &\:\boldsymbol{J_f}(x^\ast,p):
         &&\;\;
         \big[
        (-\!1+p)I
        -\! 3\,\mathrm{diag}(x_1^{*2}, \ldots, x_N^{*2})
        - \xi W
        \big]
        \\
        &\:\boldsymbol{\Delta_{\textbf{CIM}}(\mathcal{G})}:& &\;\; -\! \!\lambda_{\max}\big(-\xi W -\! 3\mathrm{diag}(x_1^{*2}, \ldots, x_N^{*2}) \big) +\! \lambda_{\max}(-\!\xi W)
        \end{aligned}$}& 
        \raisebox{-1in}{\includegraphics[height=1.1in]{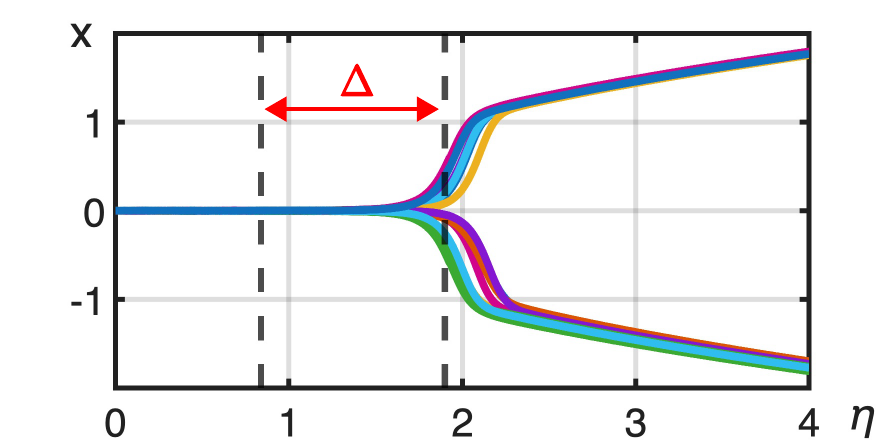}}\\ \hline
        
        \multirow{1}{1em}{ \huge \rotatebox[origin=c]{90}{{SBM}}}& \hspace{-0.005in}
        \multirow{1}{*}{$\begin{aligned}
        & \boldsymbol{\eta}: & &\;\; p \\
        &\boldsymbol{E(x, \eta, W)}:& &\;\; \sum\limits_{i=1}^{N}\left[\left(\Delta_i^{\textnormal{SBM}} -\! p\right)\dfrac{x_i^{2}}{2} +\! K_e\,\dfrac{x_i^3}{4}\right] 
        +\! \dfrac{\xi_0}{2} \sum\limits_{j=1}^N W_{ij}\,x_j \\ 
    
        &\boldsymbol{J_f}(x^\ast,p):
        &&\;\;
        \big[
        (-\!\Delta_i^{\textnormal{SBM}}+\!p)I
        -\! 3 K_e\,\mathrm{diag}(x_1^{*2},\ldots,x_N^{*2})
        -\! \xi_0 W
        \big]
        \\

        &\boldsymbol{\Delta_{\textbf{SBM}}(\mathcal{G})}:& &\;\; -\! \lambda_{\max}\big(-\!\xi_0 W -\! 3 K_e\mathrm{diag}(x_1^{*2}, \ldots, x_N^{*2}) \big) +\! \lambda_{\max}(-\!\xi_0 W)
        \end{aligned}$}& 
      \raisebox{-1 in}{\includegraphics[height=1.1in]{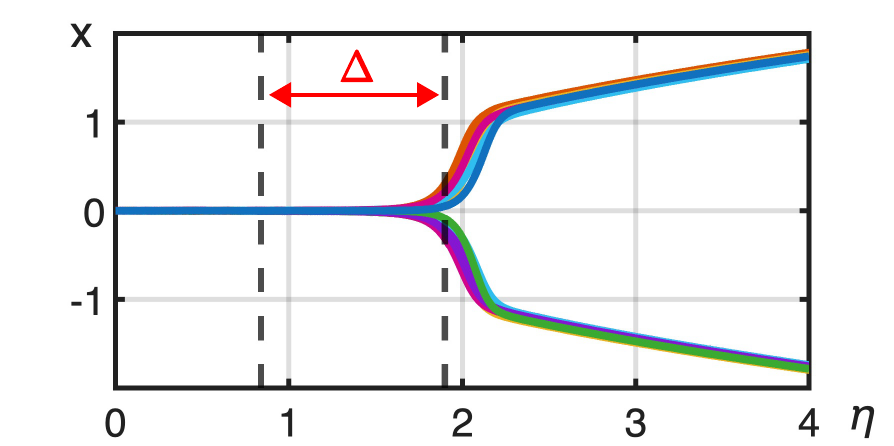}}\\ \hline
\end{tabular}}
  \caption{\justifying Comparison of dynamical system formulations for minimizing the Ising Hamiltonian and their key characteristics. Across all models, a finite parameter gap ($\Delta$) emerges in which the system evolution does not correspond to minimization of the Ising ground state in the general case. DIM: Dynamical Ising Machine; OIM: Oscillator Ising Machine; CIM: Coherent Ising Machine; SBM: Simulated Bifurcation Machine. $\phi$ represents the state variable in DIM and OIM, and $x$ denotes the state variable in SBM and CIM. The parameter $K$ denotes the coupling strength in DIM and OIM, whereas $\xi$ and $\xi_0 $ represent the coupling strength in CIM and SBM, respectively. In SBM, $\Delta_i^{\textnormal{SBM}}$ represents the positive detuning frequency, and $K_e$ denotes the positive Kerr coefficient. \(J_f\) denotes the full flow Jacobian evaluated at the indicated fixed point.} 
  \label{fig:ising_systems}
\end{figure*}

A striking observation emerges within this class: \emph{In the general case, there exists a finite interval of the regularization parameter $\eta$, with an intractable upper bound, over which the system dynamics fail to faithfully minimize the Ising Hamiltonian}. We refer to this interval as the \emph{parameter gap},
\begin{equation}
    \Delta=\eta_{\max}-\eta_{\min},
    \label{equation2}
\end{equation}
where \(\eta_{\min}\) is the threshold at which the trivial state loses
stability, and \(\eta_{\max}\) is the threshold beyond which at least one
Ising ground-state configuration becomes stable. These thresholds are
determined by the largest eigenvalues of the Jacobians evaluated at the
trivial state and at Ising ground-state fixed points, respectively.
Determining \(\eta_{\max}\) requires knowledge of the Ising ground-state
set and is therefore generally intractable
(see Supplementary Notes~2 and~3).
For the class of regularized phase-encoded Ising gradient flows considered
here, Supplementary Note~2 establishes
the ordering
\(
    \eta_{\max}\geq\eta_{\min},
\)
and hence
\(
    \Delta\geq0;
\)
that is, the parameter gap is non-negative.

The four models considered in this Letter---DIM, OIM, CIM, and SBM---are specific realizations of this framework, with CIMs and SBMs admitting a reduced phase-model description near the bifurcation, as detailed in Supplementary Notes~2 and~3. Figure~\ref{fig:ising_systems} summarizes their defining characteristics and illustrates their dynamics as the regularization parameter \(\eta\) is varied. Simulation parameters and experimental settings are provided in Supplementary Note~10.

Starting from \(\eta=0\), none of the considered systems is guaranteed to evolve to a valid Ising state. In CIMs, SBMs, and DIMs, the dynamics preferentially collapse into a trivial non-Ising state: \(x=0\) for CIMs and SBMs, and \(\phi=(\pi/2)\mathbf{1}\) for DIMs. In OIMs, the corresponding trivial state \(\phi=(\pi/2)\mathbf{1}\) has a threshold \(\eta_{\min}<0\), so it is already unstable over the physically relevant regime \(\eta\geq0\). Ising ground-state configurations become stable only beyond \(\eta_{\max}\). 

The existence of a parameter gap highlights an intrinsic limitation of such dynamical systems, with direct consequences for its functional capability as Ising machines. During practical operation, the regularization parameter is ramped up, forcing the system to traverse this parameter gap. Within this regime, the trivial state is unstable, and even small perturbations, unavoidable in the presence of noise, can deflect the dynamics toward non-optimal Ising solutions, as explored in the following section. Importantly, the failure to minimize the Ising Hamiltonian within the parameter gap does not stem from breakdown of the gradient-flow dynamics (with $\dot{E} \leq 0$ remaining valid), but rather from the fact that the ground state itself does not correspond to an Ising solution until $\eta \geq \eta_{\max}$.

\textit{Fate of the Dynamics in the Gap}---Since the system must traverse the parameter gap to reach the Ising solutions, we now examine the dynamics within this regime. The destabilization of the trivial state occurs when the largest eigenvalue (of the corresponding Jacobian)  crosses zero from negative to positive, rendering the corresponding eigenmode unstable.  In the ideal quasistatic, weak-noise limit, the initial post-bifurcation displacement is governed by \( \boldsymbol{v}_{\max}\), the dominant eigenvector of the Jacobian evaluated at the trivial state at
\(\eta=\eta_{\min}\), with the associated spin configuration \( \boldsymbol{\sigma}_{\max}=\mathrm{sign}(\boldsymbol{v}_{\max}) \)\cite{Kalinin2022,2023Bifurcation}.

When the parameter gap vanishes (\(\Delta = 0\)), the destabilization of the trivial state coincides with the stabilization of an optimal Ising configuration. In this case, the largest eigenvalue of the Jacobian evaluated about both the trivial state and the optimal configuration crosses zero at the same value of the control parameter ($\eta_{\Delta=0}$). Crucially, since the optimal Ising configuration now constitutes the system’s ground state, the dominant linearized mode—associated with the leading eigenvector of the Jacobian—aligns with this configuration. This alignment can be quantified via the normalized projection 
\(
\frac{\boldsymbol{\sigma}^{\ast}\cdot \boldsymbol{v}_{\max}}
{\|\boldsymbol{\sigma}^{\ast}\|\,\|\boldsymbol{v}_{\max}\|},
\)
where \(\boldsymbol{\sigma}^{\ast}\) denotes the ground-state Ising configuration. Importantly, this quantity equals one only when \(\boldsymbol{v}_{\max}\) is a positive scalar multiple of \(\boldsymbol{\sigma}^{\ast}\).

Furthermore, we will argue (in Supplementary Note~4) that when $\frac{\boldsymbol{\sigma}_{\max}\cdot \boldsymbol{v}_{\max}}{\|\boldsymbol{\sigma}_{\max}\|\,\|\boldsymbol{v}_{\max}\|} = 1$, the eigenvector corresponding to the largest eigenmode will represent an optimal Ising configuration whose energy is given by, $H=-\tfrac{N}{2}{\lambda_{\rm{max}}} - \# \mathcal{E}(\mathcal{G})$, where $\lambda_{\max}$ denotes the largest eigenvalue of the Jacobian corresponding to the trivial state for $\eta_{\Delta=0}$ and $\#\mathcal{E}(\mathcal{G})$ represents the number of edges in the graph $\mathcal{G}$. Such behavior, where $\boldsymbol{v}_{\max}$ perfectly aligns with an Ising configuration arises only in special cases, including bipartite graphs (see Supplementary Note~4 for a detailed derivation). From a functional standpoint, this implies that the dynamics converge to the ground state with high probability. We verify this behavior through numerical simulations in Supplementary Note~4, demonstrating that  DIMs, OIMs, CIMs, and SBMs consistently converge to the ground state for bipartite graphs.

For generic graphs, however, the parameter gap is nonzero. In this regime, destabilization of the trivial state does not typically coincide with stabilization of an optimal Ising configuration, and the leading eigenvector \(\boldsymbol{v}_{\max}\) of the Jacobian need not align with any ground-state Ising configuration.  More precisely, solution quality is governed by the weighted excess subdominant-eigenspace projection of the bifurcation-selected spin configuration relative to an optimal configuration, as shown in Supplementary Note~4.

The operational dynamics during traversal of the gap can be viewed as a stochastic mode-selection process governed by the interplay among noise, the temporal profile of the regularization parameter $\eta$, and the eigenspectrum at bifurcation. While a detailed stochastic analysis is presented in Supplementary Note~5, qualitatively, the role of these quantities can be understood as follows: noise seeds perturbations along the eigendirections, the ramp rate of \(\eta\) controls how long these perturbations are amplified before \(\eta_{\max}\) is reached, and the local eigengap \(\gamma_{N-1}\) controls the susceptibility of the
mode-selection process to noise-induced mixing with competing modes. 

While noise characteristics and the temporal protocol are important, here we focus on spectral features intrinsic to the bifurcation, namely, the parameter-gap width and the alignment of \(\boldsymbol{v}_{\max}\) with low-energy Ising configurations. The former controls the synchronization threshold for certification, whereas the latter determines the bifurcation-selected spin configuration.  Below, we introduce a
hybrid framework through which these properties can be  modulated.

\textit{Bridging the Gap}--- We now explore a pathway to improve the alignment between  \(\boldsymbol{v}_{\max}\) and a optimal spin configuration ($\boldsymbol{\sigma}^{\ast}$). To understand the foundation of our approach, we consider---using the approach proposed by Wang \emph{et al.},~\cite{2023Bifurcation} the degree of synchronization,
\begin{align}
    S(\boldsymbol{v}_{\max}) = \dfrac{1}{N}(\boldsymbol{v}_{\max}^{\top}\boldsymbol{\sigma}_{\mathrm{max}})^{2}, 
\end{align}
where, $\boldsymbol{\sigma}_\mathrm{max}=\operatorname{sign}(\boldsymbol{v}_{\mathrm{max}})$.
The significance of this metric, as quantified further on, is that it helps define a lower bound for the dominant eigenvector, $\boldsymbol{v}_{\mathrm{max}}$, to point to a state with an energy below a given threshold. Since $\boldsymbol{v}_{\mathrm{max}}$ is determined by the Jacobian at the trivial solution at the point of bifurcation, improving $S(\boldsymbol{v}_{\max})$ requires reshaping the eigenbasis of this Jacobian.

To this end, we consider a convex combination of existing dynamics designed to reshape the Jacobian spectrum and selectively enhance the projection of $\boldsymbol{v}_{\max}$ onto an optimal spin configuration while suppressing competing eigendirections. Specifically, we introduce the Hybrid Ising Machine (HyIM) defined by,
\begin{align}
    \dot{\phi}_{i} &= K\sum_{j=1}^{N}W_{ij}\bigl[\alpha \sin(\phi_{i} + \phi_{j}) +(1-\alpha)\sin(\phi_{i} - \phi_{j})\bigr] \nonumber \\ 
    & \;\;\;\;\; - K_{s}\sin(2\phi_{i}).
    \label{eqn: alpha-OIM-DIM}
\end{align}
where $\alpha \in [0,1]$ interpolates between the OIM and DIM limits. 

From an implementation perspective, this interpolation can be interpreted as an oscillator network with two coupling paths. The phase-difference term is produced by a direct coupling path, while the phase-sum term can be produced by coupling to a phase-conjugated copy of the neighboring oscillator signal~\cite{2024IMWTL..34..671I,article,inproceedings}, since
\(
\sin(\phi_i+\phi_j)=\sin(\phi_i-(-\phi_j)).
\)
The hyperparameter \(\alpha\) therefore represents the relative coupling strength assigned to the direct and phase-conjugate paths.

As in the pure cases, the HyIM admits an energy function $E_{\text{HyIM}}$ whose minimizers correspond to the minimizers of the Ising Hamiltonian (see Supplementary Note 6). Additionally, $\boldsymbol{v}_{\mathrm{max}}$ can be selectively destabilized while all other eigenmodes remain stable, which is critical to ensuring that the initial post-bifurcation evolution is governed by the dominant eigenmode. This selective destabilization occurs only for $\alpha \in [\alpha_c,1]$ (see Supplementary Note 7). Consequently, $\alpha_c$ determines the lower bound on $\alpha$ above which the initial HyIM evolution is governed by the dominant eigenmode considered in the analysis below. Computing $\alpha_c$ does not require prior knowledge of the ground state.

The gap in the dynamics can be defined as, 
\begin{align}
\Delta_{\mathrm{HyIM}}(\alpha) &= K^{\{0,\pi\}}_{s,\mathrm{HyIM}}(\alpha) - K_{s,\mathrm{HyIM}}^{\{\frac{\pi}{2}\}}(\alpha),\label{eqn: dynamic_gap_defined}
\end{align} and it can be shown that $\Delta_{\text{HyIM}} \geq 0$ and convex in $\alpha\in[0,1]$ (see Supplementary  Note~7). The convexity implies that if there exists an $\alpha_0 \in [\alpha_c,1]$ such that $\Delta_{\text{HyIM}}(\alpha_0)<\Delta_{\text{HyIM}}(\alpha=1)$, then \(\Delta_{\text{HyIM}}(\alpha) \le \Delta_{\text{HyIM}}(\alpha=1) \quad \text{for all } \alpha \in [\alpha_0, 1]\). While the value of $\alpha$ that gives rise to the minimum parameter gap is not known \emph{a priori}, $\alpha$ provides an additional \emph{design knob} to optimize the parameter gap; in practice, we observe that the minimum parameter gap does not usually occur at the end point ($\alpha=1)$ implying that the standalone DIM dynamics do not provide the minimum gap. Scaling of the minimum-gap in HyIM across graph sizes and graph classes, together with the associated solution-quality certificates, is reported in Supplementary Note~8. In the following, we will show that the gap sets the synchronization threshold required to certify the bifurcation-selected solution.

From a functional standpoint, for \(\alpha\in[0,\alpha_m^k]\), the dominant bifurcation mode certifies a solution below the \(k^{\mathrm{th}}\) excited-state energy level whenever
\begin{align}
    S(\boldsymbol{v}_{\max})
    \geq
    S_{\mathrm{crit},k}(\boldsymbol{v}_{\max})
    =
    1-\frac{\Delta H^k}{N\Delta_{\mathrm{HyIM}}(\alpha)} .
    \label{eqn:doscondition}
\end{align}
Here, \(S_{\mathrm{crit},k}(\boldsymbol{v}_{\max})\) is the critical synchronization threshold required to certify that the spin configuration induced by the dominant mode,
\(\boldsymbol{\sigma}_{\max}=\operatorname{sign}(\boldsymbol{v}_{\max})\), satisfies
\(
    H(\boldsymbol{\sigma}_{\max})\leq H^k,
\)
where, \(H^k\) denotes the \(k^{\mathrm{th}}\) excited-state energy level, \(H_0\) denotes the ground-state energy, and \(\Delta H^k=H^k-H_0\). The quantity \(\alpha_m^k\) is defined as the supremum of the set of \(\alpha\) values for which Eq.~\eqref{eqn:doscondition} is satisfied; see Supplementary Note~9.


Equation~\eqref{eqn:doscondition} provides a quantitative link between synchronization and solution quality: if the dominant eigendirection satisfies
the inequality, the resulting spin configuration is guaranteed to lie below the \(k^{\mathrm{th}}\) excited state. Crucially, the critical threshold
\(S_{\mathrm{crit},k}\) decreases with the parameter gap \(\Delta_{\mathrm{HyIM}}(\alpha)\). Thus, tuning \(\alpha\) to reduce the gap lowers the degree of synchronization required to certify a low-energy solution. This structural role of \(\Delta_{\mathrm{HyIM}}\) is distinct from that of the temporal protocol \(\eta(t)\), which governs how the trajectory traverses the gap. Consequently, for \(k\) such that \(\alpha_m^{k}>\alpha_c\), if \(\boldsymbol{v}_{\max}\) satisfies Eq.~\eqref{eqn:doscondition}, then it identifies a spin configuration \(\boldsymbol{\sigma}_{\max}\) satisfying
\(
    H(\boldsymbol{\sigma}_{\max})\leq H^{k}.
\) 
\begin{figure}[!h]
    \centering
    \includegraphics[scale = 0.4]{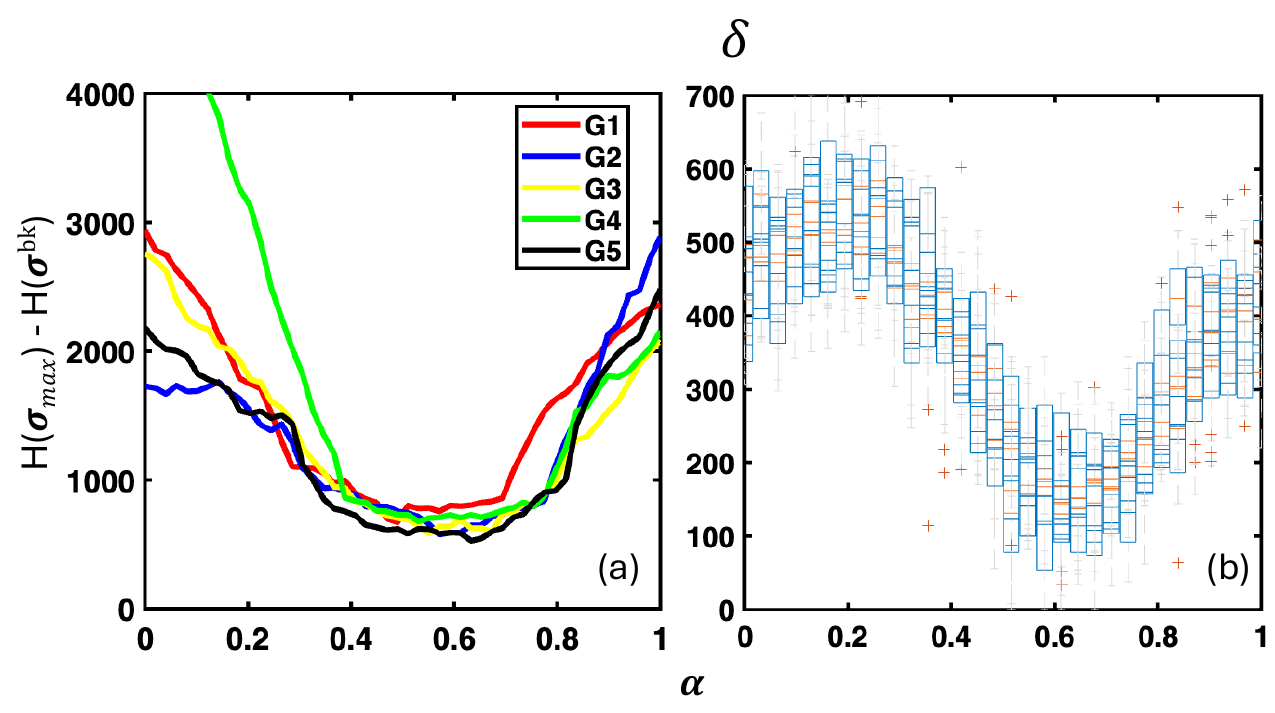}
    \caption{\justifying (a) Energy offset of the dominant-eigenvector configuration,
\(H(\boldsymbol{\sigma}_{\max})-H(\boldsymbol{\sigma}^{\mathrm{bk}})\), versus \(\alpha\) for the first five G-set instances, with
\(\boldsymbol{\sigma}_{\max}=\operatorname{sign}(\boldsymbol{v}_{\max})\). \(\boldsymbol{\sigma}^{\mathrm{bk}}\) denotes the best-known reference spin configuration for the corresponding G-set instance
(b)Box plots of the energy deviation
\(\delta_r(\alpha)=H_r(\alpha)-H_{\min}(\alpha)\), where
\(H_r(\alpha)\) is the trial-\(r\) Ising energy and
\(H_{\min}(\alpha)=\min_r H_r(\alpha)\). Each graph is simulated for 10 independent trials.}
    \label{fig:gap_vs_alpha}
\end{figure}

To illustrate the effect of the hybrid dynamics, we test five graphs, G1--G5, from the G-set benchmark \cite{Gset} each consisting of 800 nodes and 19,176 unweighted edges. Fig.~\ref{fig:gap_vs_alpha}(a) shows the evolution of
\(
H(\boldsymbol{\sigma}_{\max})-H(\boldsymbol{\sigma}^{\mathrm{bk}})
\)
as a function of \(\alpha\), where
\(\boldsymbol{\sigma}_{\max}=\mathrm{sign}(\boldsymbol{v}_{\max})\) is the
spin configuration obtained from the dominant bifurcation eigenvector and
\(\boldsymbol{\sigma}^{\mathrm{bk}}\) denotes the best-known reference spin
configuration for the corresponding G-set instance. Since a certified ground
state is not available for these 800-node instances, this quantity should be
interpreted as a lower bound on the deviation from the true ground-state
energy,
\(
H(\boldsymbol{\sigma}_{\max})-H(\boldsymbol{\sigma}^{\ast})
\). Subsequently, we also solve each graph for varying values of $\alpha$. Each graph is evaluated over 10 independent trials. 
Fig.~\ref{fig:gap_vs_alpha}(b) shows box plots of
\(\delta_r(\alpha)=H_r(\alpha)-H_{\min}(\alpha)\), where
\(H_r(\alpha)\) is the Ising energy from trial \(r\) and
\(H_{\min}(\alpha)=\min_r H_r(\alpha)\), over 10 trials at each \(\alpha\).


It can be observed that the Ising energy of the binary state, estimated from the sign of the leading eigenvector at the bifurcation point—attains a minimum for $\alpha \approx 0.5$--$0.75$, indicating that the hybrid dynamics can help point the system dynamics evolve towards a lower energy state. More importantly, this improvement translates into improved functional performance (i.e., lower Ising energy) as observed in Fig.~\ref{fig:gap_vs_alpha}(b).

\textit{Conclusion}--- In summary, we identify the parameter gap as a fundamental structural constraint intrinsic to a broad class of analog Ising machines. Although these systems evolve as gradient flows, destabilization of the trivial state and stabilization of Ising-encoded solutions occur at distinct thresholds. This separation creates a regime in which optimal spin configurations are not yet stabilized, rendering solution quality dependent on the spectral structure of the dynamics at bifurcation. 

Beyond the pairwise Ising interactions analyzed in this Letter, the extension to regularized gradient flows with higher-order Ising interactions, presented in Supplementary Note~11, suggests that parameter-gap analysis can serve as a broader design principle for such platforms. The parameter gap thereby provides a structural lens for analyzing and designing physical dynamical systems for discrete optimization, while the hybrid framework suggests new directions for controlling bifurcation structure and spectral alignment to improve the performance of analog computing architectures.\\

\textit{Acknowledgments}---This material is based upon work supported by ARO award W911NF-24-1-0228.
\\[0.5em]
E.M.~Hasantha~Ekanayake and Arvind~R.~Venkatakrishnan contributed equally to
this work.
\\[0.5em]
\textit{Data availability}--- The data that support the findings of this study are available
from the corresponding author upon reasonable request.\\[-1.4em]

\def\bibsection{\section*{References}}  
\bibliography{ArXiv}

\newpage
\appendix
\onecolumngrid

\newpage
\appendix
\onecolumngrid

\section*{Supplemental Material}

\subsection*{Definition of Symbols}
The following symbols are used throughout this Supplementary Material:
\begin{itemize}
    \item $G$: Signed adjacency matrix of the underlying graph.
    \item $W$: Weight matrix, defined by the relation $W_{ij} = -G_{ij}$.
    \item $\#\mathcal{E}(\mathcal{G})$: The number of edges in the graph $\mathcal{G}$
    \item $D$: Degree matrix, defined as the diagonal matrix with entries $D_{ii} = \sum_{j} W_{ij}$ (the weighted degree of node $i$), and $D_{ij} = 0$ for $i \neq j$.
    \item $E(.)$: Energy function.
    \item $J_f$: Jacobian matrix of the dynamical system.
    \item $S(.)$: Synchronization.
    \item $\boldsymbol{\sigma}^{\ast}$: An Ising ground state.
    \item $\boldsymbol{\sigma}_{\max}$: The sign of the dominant eigenvector at the bifurcation point, $\boldsymbol{\sigma}_{\max} = \mathrm{sign}(\boldsymbol{v}_{\max})$. 
    \item $f_{\mathrm{coup}}(.)$ :  Coupling function.
    \item $f_{\mathrm{reg}}(.)$ :  Regularization function.
    \item $\gamma_i$: Eigenvalue gap between the dominant eigenvalue and the $i^{\mathrm{th}}$ eigenvalue, defined as $\gamma_i={\lambda}_{\max}-{\lambda}_i$.    
    \item $\Delta_{(\cdot)}$: Parameter gap, where the subscript specifies the corresponding dynamical model.
    item $\alpha$: HyIM hyperparameter, defined as the mixing coefficient between OIM and DIM, where $\alpha\in[0,1]$.  
    \item $\Delta H^{k}$: Certified energy gap, defined as $\Delta H^{k}=H^{k}-H_{0}$. 
    Where $H^{k}$ is the certified Ising energy at the synchronization threshold, and $H_{0}$ is the ground-state Ising energy.
    \item $CR_{k}$: Certifiable ratio, defined as $CR_{k}=C_{\mathrm{cert}}/C_{\mathrm{ref}}$, where $C_{\mathrm{cert}}$ is the certified cut value obtained from the HyIM certification procedure, and $C_{\mathrm{ref}}$ is the reference cut value, equal to the certified optimal cut or the best-known cut depending on the graph class.
\end{itemize}
\clearpage

\section{Supplementary Note 1}
\subsection*{{Framework for Using Dynamical Systems to Minimize the Ising Hamiltonian}}

\noindent Here, we provide a framework for describing the dynamics of the Ising machines considered in this work. The problem of minimizing the Ising Hamiltonian can be mathematically expressed as, 
\begin{align}
    \boldsymbol{\sigma^{\ast}}(\mathcal{G}) = \underset{\sigma\in\{-1,1\}^{N}}{\operatorname{argmin}} - \frac{1}{2}\boldsymbol{\sigma}^{\top}G\boldsymbol{\sigma},
\end{align}
where $\mathcal{G}$ denotes the underlying graph, $G$ is the adjacency matrix, $N$ is the number of nodes, and $\boldsymbol{\sigma}_i \in \{\pm1\}^N$ is the Ising spin vector. Computing $\boldsymbol{\sigma^*}$ is NP-hard in the general case. The systems (Dynamical Ising Machine (DIM), Oscillator Ising Machine (OIM), Coherent Ising Machine (CIM) and Simulated Bifurcation Machine (SBM)) implement a continuous relaxation of the Ising Hamiltonian by defining an energy function $E(\phi)$ whose minimizers coincide with those of the Ising Hamiltonian. In many cases, a readout map $h(x)$ is additionally defined such that the optimal Ising configuration is recovered as 
$\boldsymbol{\sigma^*}=h(\boldsymbol{\phi}^\ast)$, where $\boldsymbol{\phi}^*$ denotes a minimizer of $E(\phi)$. The associated gradient-flow dynamics, in the presence of noise, can be expressed as,
\begin{align}
    \mathrm{d}\phi(t) = -\nabla_{\phi}E(\phi)\;\mathrm{d}t + \mathrm{d}B_t\label{eqn: IM-SDE}
\end{align}
where $\mathrm{d}B$ is Brownian noise. For the systems considered here,  the energy function $E(\phi,\eta\;; W)$ is parameterized by $\eta \in\mathcal{T}\subseteq \mathbb{R}$ where the mapping to the Ising Hamiltonian is valid only for some range of $\eta \in \mathcal{T}_{\text{compatible}}\subseteq \mathcal{T}$. The compatibility condition is given by 
\begin{align}
    \boldsymbol{\sigma^{\ast}}(\mathcal{G}) = \{z\in\mathbb{R}^{n}\;|\;z = h(\phi^{\ast}), \phi^{\ast} \in \underset{\phi}{\operatorname{argmin}}\;E(\phi,\eta,W)\}
\end{align}
for all $\eta\in\mathcal{T}_{\text{compatible}}$, where $z$ denotes the discrete Ising spin vector obtained via the readout map $h(\cdot)$.

\newpage 
\section{Supplementary Note 2}
\subsection*{Stability Threshold Ordering in Regularized Phase Encoded Ising Gradient Flows}

\noindent Here, we establish the stability threshold ordering in the class of regularized bifurcation-based Ising gradient flows of the form
\begin{equation}
    \dot{\phi}_i=
    \sum_jW_{ij}f_{\mathrm{coup}}(\phi_i\pm\phi_j)
    +
    f_{\mathrm{reg}}(\phi_i,\eta)=-\nabla_{\phi_{i}} E(\phi,\eta;W), \label{eq:unified_form}
\end{equation}
where $f_{\mathrm{coup}}$ encodes the interaction between spins; the sign in \(\phi_i\pm\phi_j\) is fixed by the model, with \(+\) denoting phase-sum coupling and \(-\) denoting phase-difference coupling. The term $f_{\mathrm{reg}}$ encodes the regularization. For the phase-encoded Ising flows considered here, this regularization is implemented through the standard second-harmonic injection term
\(
f_{\mathrm{reg}}(\phi_i;\eta)=-\eta\sin(2\phi_i),
\)
which stabilizes the binary phase manifold \(\phi_i\in\{0,\pi\}\). Furthermore, $\boldsymbol{W} \in \mathbb{R}^{N\times N}$ represents the quadratic interactions between the Ising spins. . \\[0.8em]
Within this generalized formulation, we show that the threshold, $\eta_{\mathrm{min}}$, at which the trivial state loses stability necessarily precedes, or coincides with, the threshold, $\eta_{\mathrm{max}}$, at which an Ising-encoded ground state becomes stable i.e., $\eta_{\mathrm{max}}\geq \eta_{\mathrm{min}}$, implying the existence of a non-negative parameter gap $\Delta \geq 0$.\\[0.8em]
\noindent To formulate the stability threshold ordering for the general class of  regularized phase-encoded Ising gradient flows, we proceed as follows:
\begin{enumerate}
\item \textbf{Unified Phase-Encoded Form of the Ising Dynamics.} We show that the admissible interaction functions in Eq.~\eqref{eq:unified_form} take the Fourier form
\begin{equation}
\dot{\phi}_i
=
K\sum_{\substack{j=1,j\neq i}}^{N}W_{ij}
\left(
\sum_{\substack{n \geq 1\\ n\ \mathrm{odd}}}a_n
\sin\big(n(\phi_i \pm \phi_j)\big)
\right)
-
\eta\sin(2\phi_i).
\label{eq:general_fourier_summary_odd}
\end{equation}
Here, the \(+\) sign corresponds to sum-phase interactions $(\phi_i+\phi_j)$, and the \(-\) sign corresponds to difference-phase interactions $(\phi_i-\phi_j)$. $K$ and $\eta$ represent the strength of the coupling between the oscillators and the SHI, respectively. $a_n$ represents the coefficient of the $n^{\text{th}}$ harmonic. \\[0.8em]
Thus, DIM and OIM are direct phase-flow realizations of Eq.~\eqref{eq:general_fourier_summary_odd}. Furthermore, we show that, under the  minimal-amplitude-heterogeneity conditions required for faithful implementation of the target Ising Hamiltonian, the CIM and the SBM dynamics reduce to the same Kuramoto-style phase-flow description.
\item \textbf{Stability Threshold Ordering in the Unified Formulation.} We then show that the ordering of the stability thresholds holds for the generalized dynamics presented in Eq.~\eqref{eq:general_fourier_summary_odd}, establishing the presence of a non-negative gap ($\Delta\geq0$). 
\end{enumerate}
\clearpage
\subsection{Unified Phase-Encoded Form of the Ising Dynamics}
\noindent \textbf{(a) Phase difference interactions [$f_{\text{coup}}(\phi_i-\phi_j)$]}\\[0.8em]
Consider the phase difference interaction of the form,
\begin{equation}
\dot{\phi}_i
=
\sum_{\substack{j=1,j\neq i}}^N
W_{ij}
f_{\text{coup}}(\phi_i-\phi_j),
\label{eq:OIM_phase_appendix}
\end{equation}
where
\(f_{\text{coup}}(\cdot)\)
denotes the interaction function in the weakly coupled regime and
\(W_{ij}=W_{ji}\)
represents symmetric real coupling. For the dynamics to admit a gradient-flow representation through a scalar energy function \(E\), with
\(\dot{\phi}_i=-\partial E/\partial \phi_i\), the induced vector field must be
curl-free:
\begin{equation}
\frac{\partial \dot{\phi}_i}{\partial \phi_j}
=
\frac{\partial \dot{\phi}_j}{\partial \phi_i},
\qquad
\forall i,j.
\label{eq:curl_condition_appendix_OIM}
\end{equation}
This condition follows from Clairaut's theorem: if a scalar potential
\(E(\boldsymbol{\phi})\) exists, then the mixed partial derivatives
\(\partial^2 E/\partial \phi_j\partial \phi_i\) and
\(\partial^2 E/\partial \phi_i\partial \phi_j\) must be equal. In the Ising-machine setting, this scalar energy is constructed so that, on the binary phase manifold, its minimizers correspond to minimizers of the target Ising Hamiltonian. Evaluating the two sides in  Eq.~\eqref{eq:curl_condition_appendix_OIM} gives
\begin{equation}
\frac{\partial \dot{\phi}_i}{\partial \phi_j}
=
-W_{ij}f_{\text{coup}}'(\phi_i-\phi_j),
\end{equation}
and
\begin{equation}
\frac{\partial \dot{\phi}_j}{\partial \phi_i}
=
-W_{ji}f_{\text{coup}}'(\phi_j-\phi_i).
\end{equation}
Under symmetric coupling,
\(W_{ij}=W_{ji}\),
the curl-free condition in Eq.~\eqref{eq:curl_condition_appendix_OIM}
therefore requires 
\begin{equation}
f_{\text{coup}}'(\phi_i-\phi_j)=f_{\text{coup}}'(-(\phi_i-\phi_j)),
\end{equation}
Equivalently, up to an additive constant,
\begin{equation}
f_{\text{coup}}(\phi_i-\phi_j)=-f_{\text{coup}}(-(\phi_i-\phi_j)).
\end{equation}
Thus, phase difference-based interactions admit a gradient-flow representation only when the nonconstant part of the interaction function is odd~\cite{hasan2026breakdowngradientflowdynamicsoscillator}. Since the interaction \(f_{\text{coup}}\) must be odd in its phase argument, its Fourier expansion contains only $\sin(.)$ terms. Consequently, 
\begin{equation}
f_{\text{coup}}(\phi_i-\phi_j)= \sum_{n\geq1} a_n \sin(n(\phi_i-\phi_j))
\end{equation}
Furthermore, the interaction function must preserve the Ising spin-flip symmetry. For phase difference-based encoding, $\phi_i-\phi_j=0
\quad \leftrightarrow \quad
\sigma_i\sigma_j=+1,$ and $\phi_i-\phi_j=\pm\pi
\quad \leftrightarrow \quad
\sigma_i\sigma_j=-1$. A spin flip corresponds to $\phi_i\rightarrow\phi_i+\pi$ which transforms the phase difference as $(\phi_i-\phi_j)
\rightarrow
(\phi_i-\phi_j)+\pi$. To remain consistent with the Ising pairwise interaction
\begin{equation}
H_{\mathrm{Ising}}
=
-\sum_{i<j}G_{ij}\sigma_i\sigma_j,
\qquad
\sigma_i\in\{-1,+1\},
\end{equation}
the interaction must reverse sign under a spin flip. Hence,
\begin{equation}
f_{\text{coup}}(\phi_i-\phi_j+\pi)=-f_{\text{coup}}(\phi_i-\phi_j).
\label{eq:half_wave_OIM}
\end{equation}
Considering the Fourier interaction of the $\pi$ shifted function
\begin{equation}
f_{\text{coup}}(\phi_i-\phi_j+\pi)= \sum_{n\geq1} a_n (-1)^n\sin(n(\phi_i-\phi_j)) \label{eq:shifted_Fourier}
\end{equation}
For Eq.~\eqref{eq:shifted_Fourier} to equal $-f_{\text{coup}}(\phi_i-\phi_j)$ as required by Eq.~\eqref{eq:half_wave_OIM}, 
\[(-1)^n=-1\]
for every harmonic with nonzero coefficient \(a_n\). Hence, only odd harmonics $\sin(.)$ are permitted, implying $a_n = 0$ for all even $n$.  Therefore, phase-difference based interaction functions that preserve both gradient-flow structure and Ising spin-flip symmetry can be expressed in the form,
\begin{equation}
\dot{\phi}_i
=
K\sum_{\substack{j=1,j\neq i}}^{N}W_{ij}
\left(
\sum_{\substack{n \geq 1\\ n\ \mathrm{odd}}}a_n
\sin\big(n(\phi_i - \phi_j)\big)
\right)
-
\eta\sin(2\phi_i).
\label{eq:general_fourier_OIM}
\end{equation}\\
\textbf{(b) Additive Phase Interactions ($f_{\text{coup}}(\phi_i+\phi_j)$)}\\[0.8em]
\noindent Consider interaction functions with additive phase dynamics,
\begin{equation}
\dot{\phi}_i
=
\sum_{\substack{j=1,j\neq i}}^N
W_{ij}
f_{\text{coup}}(\phi_i+\phi_j),
\label{eq:DIM_phase_appendix}
\end{equation}
where
\(f_{\text{coup}}(.)\)
denotes the interaction function and
\(W_{ij}=W_{ji}\)
represents symmetric real coupling. Using the same approach as illustrated above in the case of phase difference-based interaction functions, a gradient-flow interpretation of the dynamics entails that,
\begin{equation}
\frac{\partial \dot{\phi}_i}{\partial \phi_j}
=
\frac{\partial \dot{\phi}_j}{\partial \phi_i},
\qquad
\forall\; i,j \;\;
\Rightarrow\,\,\, 
W_{ij}f_{\text{coup}}'(\phi_i+\phi_j)=W_{ji}f_{\text{coup}}'(\phi_j+\phi_i)
\label{eq:curl_condition_appendix_DIM}
\end{equation}
Since \(\phi_i+\phi_j=\phi_j+\phi_i\), the curl-free condition in
Eq.~\eqref{eq:curl_condition_appendix_DIM} is satisfied identically under
symmetric coupling. Thus, additive phase interactions admit a gradient-flow
representation for any differentiable interaction function
\(f_{\mathrm{coup}}(\phi_i+\phi_j)\).\\[0.8em]
However, gradient-flow structure alone does not ensure compatibility with the Ising encoding. For \(\phi_i,\,\phi_j\in\{0,\pi\}\), the additive phase argument satisfies \(\phi_i+\phi_j\in\{0,\pi,2\pi\}\). $\sin(.)$ harmonics vanish at all of these points, whereas cosine harmonics take values \(\pm1\). Therefore, cosine components in the coupling force do not vanish on the binary phase manifold and would shift the intended Ising equilibria. Requiring the binary phase manifold to remain an equilibrium set therefore restricts the coupling force to sine harmonics,
\begin{equation}
f_{\mathrm{coup}}(\phi_i+\phi_j)
=
\sum_{n\geq 1} a_n
\sin\big(n(\phi_i+\phi_j)\big).
\end{equation}
The corresponding dynamics are
\begin{equation}
\dot{\phi}_i
=
K\sum_{\substack{j=1,j\neq i}}^{N}W_{ij}
\left(
\sum_{n \geq 1}a_n
\sin\big(n(\phi_i + \phi_j)\big)
\right)
-
\eta\sin(2\phi_i).
\label{eq:general_fourier_summary_sine}
\end{equation}
Finally, imposing the Ising spin-flip symmetry
\(f_{\mathrm{coup}}(\phi+\pi)=-f_{\mathrm{coup}}(\phi)\)
eliminates the even harmonics as shown earlier, leaving
\begin{equation}
\dot{\phi}_i
=
K\sum_{\substack{j=1,j\neq i}}^{N}W_{ij}
\left(
\sum_{\substack{n \geq 1\\ n\ \mathrm{odd}}}a_n
\sin\big(n(\phi_i + \phi_j)\big)
\right)
-
\eta\sin(2\phi_i).
\label{eq:general_fourier_DIM}
\end{equation}\\
\noindent\textbf{(c) Phase-Encoded Form of CIM and SBM Dynamics}\\[0.8em]
\noindent
Building on the phase-reduction analysis of parametric-oscillator Ising
machines presented in~\cite{61zx-gs91}, we recast the CIM and SBM dynamics
in amplitude--phase form. Under the approximate amplitude-homogeneity
condition required for faithful implementation of the target Ising
Hamiltonian, their phase dynamics reduce to the same phase-encoded structure derived above.\\[0.8em]
\noindent To establish this correspondence, we consider a network of weakly coupled
Stuart--Landau oscillators, which provide the normal form for weakly nonlinear
oscillations near a Hopf bifurcation, subject to parametric pumping near
\(2\omega\):
\begin{equation}
\dot{\nu}_i
= (\mu_i-\zeta_i |\nu_i|^2)\nu_i
+\kappa_i e^{i\phi_p}\nu_i^{*} \label{eq:slowflow}
+\xi\sum_{\substack{j=1,j\neq i}}\left(A_{ij}\nu_j+B_{ij}\nu_j^{*}\right),
\end{equation}
where, \(\nu_i=r_i e^{i\phi_i}\). Here \(A_{ij}\) and \(B_{ij}\) denote the normal and conjugate coupling pathways, respectively. $\boldsymbol{A}$ and $\boldsymbol{B}$ are real, symmetric and constructed as scaled version of $\boldsymbol{G}\, (=-\boldsymbol{W})$ i.e., $\boldsymbol{A}=\Xi\boldsymbol{G}$ and $\boldsymbol{B}=(1-\Xi)\boldsymbol{G}$, where $\Xi \in [0,1]$. $\mu_i$ is the net linear gain–loss term, $\zeta_i (>0, \text{here})$ is the nonlinear saturation, $\kappa_i e^{i\phi_p}$ represents the parametric pump at $2\omega$ with phase $\phi_p$. This form provides a reduced description of parametric oscillator networks in the near-threshold, single-mode regime, including CIM-type DOPO (degenerate optical parametric oscillator) implementations (also analyzed in \cite{61zx-gs91}). The corresponding complex energy function 
\begin{equation}
\begin{aligned}
E(\boldsymbol{\nu},\boldsymbol {\nu^{\!*}})
&=\; \sum_{i=1}^N
\Big[\,
-\mu_i\,|\nu_i|^2
\;+\;\frac{\zeta_i}{2}\,|\nu_i|^4
\;-\;\frac{\kappa_i}{2}\,\big(e^{i\phi_p} \nu_i^{\!*2}+e^{-i\phi_p} \nu_i^{2}\big)
\Big] \\
&-\;\frac{\xi}{2}\sum_{\substack{i,j=1\\ i\neq j}}^N
\Big[\,
A_{ij}\,\big(\nu_i \nu_j^{\!*}+\nu_i^{\!*}\nu_j\big)
\;+\; B_{ij}\,\big(\nu_i \nu_j+\nu_i^{\!*}\nu_j^{\!*}\big)
\Big]
\end{aligned}
\label{eq:E}
\end{equation}
satisfies \(\dot{E}\leq 0\). Expressing \(\nu_i=r_i e^{i\phi_i}\) and separating amplitude and phase gives
\begin{equation}
\begin{split}
\dot{r}_i
&=(\mu_i-\zeta_i r_i^2)r_i
+\kappa_i r_i\cos(2\phi_i)
+\xi\sum_{j\neq i}
\left[
A_{ij}r_j\cos(\phi_j-\phi_i)
+
B_{ij}r_j\cos(\phi_i+\phi_j)
\right],
\\
r_i\dot{\phi}_i
&=
-\kappa_i r_i\sin(2\phi_i)
+\xi\sum_{j\neq i}
\left[
A_{ij}r_j\sin(\phi_j-\phi_i)
-
B_{ij}r_j\sin(\phi_i+\phi_j)
\right].
\end{split}
\label{eq:amp_phase_CIM_SBM}
\end{equation}
At binary phase states \(\phi_i'\in\{0,\pi\}\), both
\(\cos(\phi_i-\phi_j)\) and \(\cos(\phi_i+\phi_j)\) reduce to the same Ising product \(\sigma_i \sigma_j\), with \(\sigma_i=\cos\phi_i'\in\{-1,+1\}\). Since the normal and conjugate coupling matrices, $\boldsymbol{A}$ and $\boldsymbol{B}$ respectively, are scaled versions of the Ising weight matrix,
the interaction energy reduces to
\begin{equation}
E(\boldsymbol{\nu},\boldsymbol {\nu^{\!*}})
=
C-\xi\sum_{i\neq j} G_{ij} r_i r_j \sigma_i \sigma_j \equiv C+\xi\sum_{i\neq j} W_{ij} r_i r_j \sigma_i \sigma_j,
\label{eq:CIM_SBM_weighted_ising}
\end{equation}
up to terms independent of the spin configuration. Therefore, amplitude heterogeneity rescales the programmed couplings as
\[
\bar{W}_{ij}=r_i^\star r_j^\star W_{ij}.
\]
The intended Ising Hamiltonian is recovered most directly \emph{when the saturated amplitudes are approximately homogeneous, \(r_i^\star\simeq r^\star\)}. Under this phase-dominated limit, the dynamics reduce to
\begin{equation}
\dot{\phi}_i
=
-K_s\sin(2\phi_i)
-
K\sum_{j\neq i}
\left[
A_{ij}\sin(\phi_i-\phi_j)
+
B_{ij}\sin(\phi_i+\phi_j)
\right].
\label{eq:phase_only_CIM_SBM}
\end{equation}

Thus, with respect to the Ising functionality of the dynamics, CIM and SBM-type parametric oscillator systems reduce to the same phase-encoded structure: normal coupling generates difference-phase interactions, while conjugate coupling generates additive phase interactions. This shows that the unified formulation used for the stability-threshold ordering also captures the phase-reduced Ising-compatible limit of parametric oscillator machines. \\[0.8em]
\textit{Summary:} Phase-encoded regularized Ising gradient flows with phase-sum and phase-difference interactions can be written in the unified form
\begin{equation}
\dot{\phi}_i
=
K\sum_{j=1}^{N}W_{ij}
\left(
\sum_{\substack{n \geq 1\\ n\ \mathrm{odd}}}a_n
\sin\big(n(\phi_i \pm \phi_j)\big)
\right)
-
\eta\sin(2\phi_i).
\label{eq:general_fourier_summary_odd1}
\end{equation}
The \(+\) and \(-\) signs correspond, respectively, to phase-sum and phase-difference interactions. DIM and OIM are recovered as the first-harmonic limits of these two cases. Furthermore, in the regime required for faithful Ising implementation of the target Ising Hamiltonian, CIM and SBM-type parametric-oscillator dynamics reduce to the same phase-encoded structure. Therefore, Eq.~\eqref{eq:general_fourier_summary_odd1} provides a unified phase-encoded framework for the regularized bifurcation-based Ising gradient flows considered here. \\[0.8em]
Next, we use this formulation to establish the stability-threshold ordering.
\subsection{Stability Threshold Ordering in the Unified Formulation.}
\noindent To establish the stability threshold ordering in the unified form established above~(Eq.~\eqref{eq:general_fourier_summary_odd1}), we define:
\begin{equation}
    \Gamma_0 = f_{\mathrm{coup}}'(0)
    =
    \sum_{\substack{n\ge 1\\ n\;\mathrm{odd}}} n a_n,
    \qquad
    \Gamma_\pi = f_{\mathrm{coup}}'(\pi)
    =
    \sum_{\substack{n\ge 1\\ n\;\mathrm{odd}}} n a_n(-1)^n .
    \label{eq:gamma_definitions}
\end{equation}
Since only odd harmonics are present, 
\begin{equation}
\Gamma_{\pi}=-\Gamma_{0},
\qquad
\Gamma_{0}+\Gamma_{\pi}=0
\end{equation}
Here, we exclude the \(\Gamma_0=0\), for which the interaction Jacobian vanishes at the binary and trivial phase arguments. \\[0.8em]
The Jacobian for the unified form evaluated at an equilibrium point
\(
\boldsymbol{\phi}^{\ast}
\)
is given by
\begin{align}
J_f(\boldsymbol{\phi}^{\ast},\eta)
=
KJ_{\mathrm{U}}(\boldsymbol{\phi}^{\ast})
-
2\eta
\operatorname{diag}
\left(
\cos(2\phi_1^\ast),
\ldots,
\cos(2\phi_N^\ast)
\right),
\end{align}
where 
\begin{equation}
J_{\mathrm{U}}(\phi^\ast)
=
\begin{bmatrix}
\displaystyle
\sum_{j=1,j\neq1}^{N}
W_{1j}F(\phi_1^\ast \pm \phi_j^\ast)
&
\pm W_{12}F(\phi_1^\ast \pm \phi_2^\ast)
&
\cdots
&
\pm W_{1N}F(\phi_1^\ast \pm \phi_N^\ast)
\\[10pt]

\pm W_{21}F(\phi_2^\ast \pm \phi_1^\ast)
&
\displaystyle
\sum_{j=1,j\neq2}^{N}
W_{2j}F(\phi_2^\ast \pm \phi_j^\ast)
&
\cdots
&
\pm W_{2N}F(\phi_2^\ast \pm \phi_N^\ast)
\\[10pt]

\vdots
&
\vdots
&
\ddots
&
\vdots
\\[10pt]

\pm W_{N1}F(\phi_N^\ast \pm \phi_1^\ast)
&
\pm W_{N2}F(\phi_N^\ast \pm \phi_2^\ast)
&
\cdots
&
\displaystyle
\sum_{j=1,j\neq N}^{N}
W_{Nj}F(\phi_N^\ast \pm \phi_j^\ast)
\end{bmatrix},
\end{equation}
with\\
\begin{equation}
F(\phi_i^\ast \pm \phi_j^\ast)
=
\sum_{\substack{n\geq1\\ n\ \mathrm{odd}}}
na_n
\cos\big(n(\phi_i^\ast \pm \phi_j^\ast)\big).
\end{equation}
We denote the trivial fixed point by $\boldsymbol{\phi}_{\rm triv}^{\ast}=\frac{\pi}{2}\mathbf{1}$. Let \(\mathcal{G}_{\phi}\subseteq\{0,\pi\}^N\) denote the set of binary phase configurations corresponding to Ising ground states. For any \(\boldsymbol{\Phi}^{\ast}\in\mathcal{G}_{\phi}\), define $\sigma_i=\cos\Phi_i^\ast\in\{\pm1\}$, then 
the trivial-state threshold and the ground-state stabilization threshold are
\begin{equation}
    \eta_{\min}^{\mathrm{U}}
    =
    -\frac{K}{2}
    \lambda_{\max}\!\left(J_{\mathrm{U}}^{\{\frac{\pi}{2}\}}\right),
    \qquad
    \eta_{\max}^{\mathrm{U}}
    =
    \min_{\boldsymbol{\Phi}^{\ast}\in\mathcal{G}_{\phi}}
    \frac{K}{2}
    \lambda_{\max}\!\left(J_{\mathrm{U}}^{\Phi^\ast}\right).
\end{equation}
For a fixed ground-state configuration \(\boldsymbol{\Phi}^{\ast}\), define
\begin{equation}
    \Delta_\mathrm{U}(\boldsymbol{\Phi}^{\ast})
    =
    \eta_{\max}^{\mathrm{U}}(\boldsymbol{\Phi}^{\ast})-\eta_{\min}^{\mathrm{U}}
    =
    \frac{K}{2}
    \left[
    \lambda_{\max}\!\left(J_{\mathrm{U}}^{\Phi^\ast}\right)
    +
    \lambda_{\max}\!\left(J_{\mathrm{U}}^{\{\frac{\pi}{2}\}}\right)
    \right], \label{eq:gap_general}
\end{equation}
where $\eta_{\max}^{\mathrm{U}}(\boldsymbol{\Phi}^{\ast})
    =
    \frac{K}{2}
    \lambda_{\max}\!\left(J_{\mathrm{U}}^{\Phi^\ast}\right)$. Next, we establish that $\Delta_{\mathrm{U}} \geq 0$. It is therefore sufficient to show that
\(\Delta_\mathrm{U}(\boldsymbol{\Phi}^{\ast})\ge0\) for every
\(\boldsymbol{\Phi}^{\ast}\in\mathcal{G}_{\phi}\), since
\begin{equation}
    \Delta_\mathrm{U}
    =
    \eta_{\max}^{\mathrm{U}}-\eta_{\min}^{\mathrm{U}}
    =
    \min_{\boldsymbol{\Phi}^{\ast}\in\mathcal{G}_{\phi}}
    \Delta_\mathrm{U}(\boldsymbol{\Phi}^{\ast}).
\end{equation}Next, define
\begin{equation}
    \Delta J_{\mathrm{U}} = J_{\mathrm{U}}^{\Phi^\ast}+J_{\mathrm{U}}^{\{\frac{\pi}{2}\}}.\label{eq:deltaJ_general}
\end{equation}We first consider phase-sum interactions. Let $\sigma_i=\cos\Phi_i^\ast\in\{\pm 1\}$. At an Ising ground state, \(f_{\mathrm{coup}}'(\Phi_i^\ast+\Phi_j^\ast)=\Gamma_0 \sigma_i \sigma_j\), while at the trivial state \(f_{\mathrm{coup}}'(\pi)=\Gamma_\pi=-\Gamma_0\). Therefore,
\begin{equation}
    [\Delta J_{\mathrm{U},+}]_{ii}
    =
    \Gamma_0\sum_j W_{ij}(\sigma_i \sigma_j-1),
    \qquad
    [\Delta J_{\mathrm{U},+}]_{ij}
    =
    \Gamma_0 W_{ij}(\sigma_i \sigma_j-1)
    \quad (i\ne j).
    \label{eq:deltaJ_dim_entries}
\end{equation}
Acting on the spin vector \(\boldsymbol{\sigma}\) gives
\begin{equation}
    [\Delta J_{\mathrm{U},+}\boldsymbol{\sigma}]_i
    =
    \Gamma_0\sum_j W_{ij}(\sigma_i \sigma_j-1)(\sigma_i+\sigma_j)=0.
    \label{eq:dim_null_vector}
\end{equation}
Indeed, if \(\sigma_i=\sigma_j\), then \(\sigma_i \sigma_j-1=0\), whereas if \(\sigma_i=-\sigma_j\), then \(\sigma_i+\sigma_j=0\). Thus \(\Delta J_{\mathrm{U},+}\) has a zero eigenvalue. For phase-difference interactions, at an Ising ground state \(f_{\mathrm{coup}}'(\Phi_i^\ast-\Phi_j^\ast)=\Gamma_0 \sigma_i \sigma_j\), while at the trivial state \(f_{\mathrm{coup}}'(0)=\Gamma_0\). In this case,
\begin{equation}
    [\Delta J_{\mathrm{U},-}]_{ii}
    =
    \Gamma_0\sum_j W_{ij}(\sigma_i \sigma_j+1),
    \qquad
    [\Delta J_{\mathrm{U},-}]_{ij}
    =
    -\Gamma_0 W_{ij}(\sigma_i \sigma_j+1)
    \quad (i\ne j).
    \label{eq:deltaJ_oim_entries}
\end{equation}
Consequently, every row of \(\Delta J_{\mathrm{U},-}\) sums to zero,
\begin{equation}
    \Delta J_{\mathrm{U},-}\mathbf 1=0,  \label{eq:oim_null_vector}
\end{equation}
implying that \(\Delta J_{\mathrm{U},-}\) also has a zero eigenvalue. In both cases \(\Delta J_{\mathrm{U}}\) is real and symmetric. Therefore, the existence of a zero eigenvalue implies
\begin{equation}
    \lambda_{\max}(\Delta J_{\mathrm{U}})\ge 0.   \label{eq:lmax_deltaJ_nonnegative}
\end{equation}
Applying Weyl's inequality to \(J_{\mathrm{U}}^{\Phi^\ast}+J_{\mathrm{U}}^{\{\frac{\pi}{2}\}}\) gives
\begin{equation}
    \lambda_{\max}(J_{\mathrm{U}}^{\Phi^\ast})
    +
    \lambda_{\max}(J_{\mathrm{U}}^{\{\frac{\pi}{2}\}})
    \ge
    \lambda_{\max}(J_{\mathrm{U}}^{\Phi^\ast}+J_{\mathrm{U}}^{\{\frac{\pi}{2}\}})
    =
    \lambda_{\max}(\Delta J_{\mathrm{U}})
    \ge 0.
    \label{eq:weyl_gap_ordering_general}
\end{equation}
Multiplying by \(K/2>0\) and using Eq.~\eqref{eq:gap_general}, we obtain
\begin{equation}
    \Delta_{\mathrm{U}}(\boldsymbol{\Phi}^{\ast})
    =
    \frac{K}{2}
    \left[
    \lambda_{\max}(J_{\mathrm{U}}^{\Phi^\ast})+\lambda_{\max}(J_{\mathrm{U}}^{\{\frac{\pi}{2}\}})
    \right] \geq 0.    
\end{equation}
Since this holds for every ground-state configuration
\(\boldsymbol{\Phi}^{\ast}\in\mathcal{G}_{\phi}\), taking the minimum over
ground-state configurations gives
\(
    \Delta_{\mathrm{U}}\ge0 .
\)\\[0.8em]
\noindent \textit{Summary:} The unified formulation reveals the common Jacobian structure of phase-encoded regularized bifurcation-based Ising gradient flows and shows that the non-negative parameter gap follows from this shared stability structure. The four models---DIM, OIM, CIM, and SBM, can be considered as specific realizations within this framework.

\subsection*{Illustrative Example}
\noindent As an illustrative example, we consider different coupling functions, $f_{\mathrm{coup}}(.)$, beyond the four systems considered in the main text. Starting with the phase-sum dynamics of DIM, we progressively incorporate higher-order odd $\sin(.)$ harmonics in the coupling function. Figures~\ref{fig:Harmonic}(a-c) show the corresponding phase dynamics realized by ramping $\eta$, the magnitude of the SHI term. The simulations were performed on the same \(N=15\)-node graph considered in Fig.~1 of the main manuscript.\\
\begin{figure}[!h]
    \centering
    \includegraphics[width=1\linewidth]{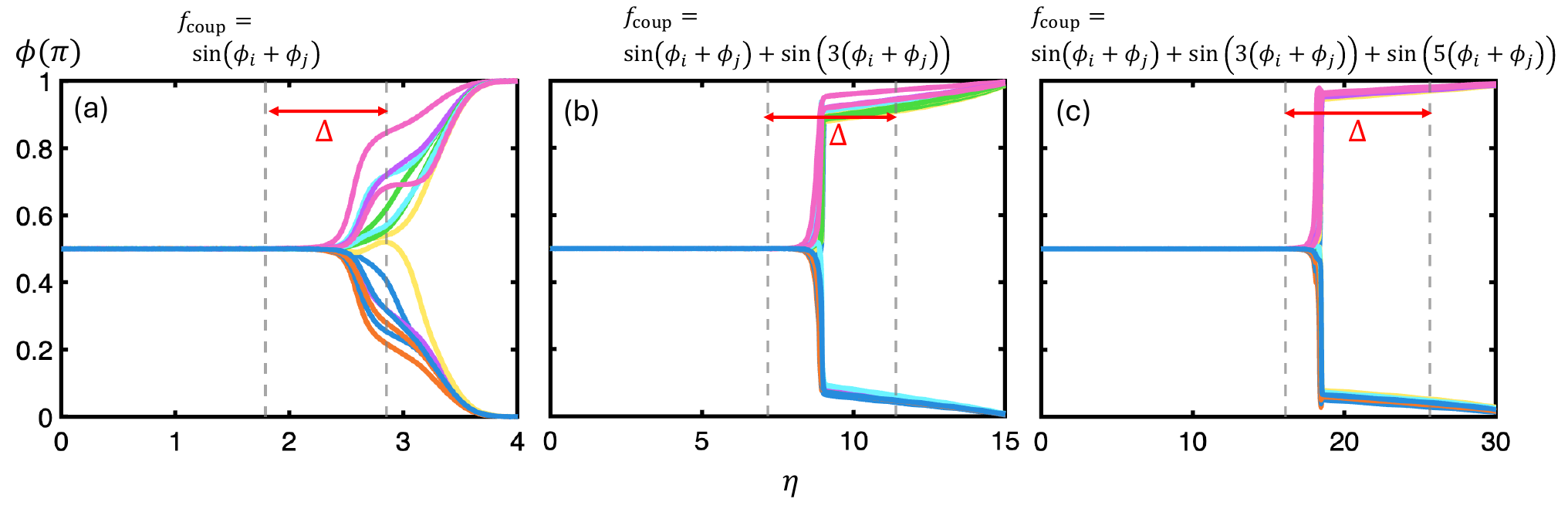}
    \caption{\justifying \smaller
Phase evolution for three different coupling functions, $f_{\mathrm{coup}}(.)$. 
The three cases considered are:
(a) \(f_{\mathrm{coup}}(\phi_i+\phi_j)=\sin(\phi_i+\phi_j)\);
(b) \(f_{\mathrm{coup}}(\phi_i+\phi_j)=\sin(\phi_i+\phi_j)+\sin\big(3(\phi_i+\phi_j)\big)\);
and (c) \(f_{\mathrm{coup}}(\phi_i+\phi_j)=\sin(\phi_i+\phi_j)+\sin\big(3(\phi_i+\phi_j)\big)+\sin \big(5(\phi_i+\phi_j)\big)\).
The corresponding parameter gaps are \(\Delta=1.06,4.23,\) and \(9.53\), respectively. The simulations were performed on the same \(N=15\)-node graph considered in Fig.~1 of the main manuscript.}
    \label{fig:Harmonic}
\end{figure}

\clearpage

\section{Supplementary Note 3}
\subsection*{Parameter Gap Analysis of Native CIM and SBM Dynamics}

\noindent Here, we analyze the stability thresholds in the native CIM and SBM dynamics (including amplitude dynamics) and show the presence of a non-negative parameter gap.

\subsubsection{\textbf{Coherent Ising Machine (CIM)}}
\noindent Following the formulation of Wang \textit{et al.} \cite{2023Bifurcation}, the dynamics of a CIM can be expressed as,
\begin{align}
\frac{d x_i}{d t} &= \left(-1 + p(t)\right)x_i - x_i^3 - \xi \sum_{j=1}^{N} W_{ij} x_j,
\end{align}
where \(x_i\in\mathbb{R}\) denotes the amplitude of the $i^{th}$ oscillator and the parameters \(p\) and \(\xi\) represent the control parameter ($p \equiv \eta$), and the coupling strength, respectively. At a fixed value of $p$, the stability of the equilibrium point $x^*$ is determined by the Jacobian matrix,
\begin{align}
J_f(x^*,p) = \big[\,(-1+p)I \;-\; 3\,\mathrm{diag}(x_1^{*2}, \ldots, x_N^{*2}) \;- \; \xi W \,\big]
\end{align}
It can be observed that the trivial state $x^\star = \mathbf{0}_N$ is always a \emph{unique} fixed point of the CIM dynamics, until it looses stability when $p$ exceeds a critical value given by $\eta_{min}\equiv p_{\{0\}}=1 - \lambda_{\text{max}}(-\xi W)$. Similarly, the critical value of $p$ at which at least one optimal Ising configuration, $\{x^*\}^N$, becomes stable can be expressed as, $\eta_{max}\equiv p_{\{\text{Ising}\}}=1- \lambda_{\text{max}}\big(- \xi W - 3\mathrm{diag}(x_1^{*2}, \ldots, x_N^{*2}) \big)$. Consequently, the parameter gap in the CIM dynamics can be expressed as,
\begin{equation}
\Delta _{\text{CIM}}=p_{\{\text{Ising}\}}-p_{\{0\}}=- \lambda_{\text{max}}\big(- \xi W - 3\mathrm{diag}(x_1^{*2}, \ldots, x_N^{*2}) \big) +  \lambda_{\text{max}}(-\xi W).  
\end{equation}
Furthermore, using Weyl's inequality, we obtain the following relationship, 
\begin{align*}
    -\lambda_{\max}(- \xi W - 3\operatorname{diag}(x_{1}^{*2},\hdots,x_{N}^{*2})) +  \lambda_{\max}(-\xi W) &= \lambda_{\min}(\xi W + 3\operatorname{diag}(x_{1}^{*2},\hdots,x_{N}^{*2})) +  \lambda_{\max}(-\xi W)\\
    & \geq \lambda_{\min}(3\operatorname{diag}(x_{1}^{*2}, \hdots, x_{N}^{*2}))= 3\min_{1\leq i\leq n}x_{i}^{*2}\geq 0.
\end{align*}
which shows that the parameter gap is always non-negative.

\label{appendix:D}

\subsubsection{\textbf{Simulated Bifurcation Machine (SBM)}}
\noindent The dynamics (simplified form) of the SBM can be expressed as~\cite{doi:10.1126/sciadv.aav2372},
\begin{equation}
\frac{d x_i}{d t} \;=\; \bigl(-\Delta_i^{\textnormal{SBM}} + p(t)\bigl)\,x_i 
\;-\;K_e\,x_i^3 
\;-\;\xi_0 \sum_{j=1}^N W_{ij}\,x_j \label{appendix:SBM_final}
\end{equation}
where $\Delta_i^{\textnormal{SBM}}$ is the positive detuning frequency between the resonance frequency of the $i^{th}$ oscillator and half the pump frequency, $p(t)$ is the time‐dependent two‐photon pumping amplitude which serves as the control parameter $\eta$, $K_e$ is the positive Kerr coefficient, and $\xi_0$ is the coupling strength. For a fixed $p$, the stability of the equilibrium point $x^*$ is determined by the Jacobian matrix,
\begin{align}
J_f(x^*,p) = \big[\,(-\Delta_i^{\textnormal{SBM}}+p)I \;-\; 3 K_e\,\mathrm{diag}(x_1^{*2},, \ldots, x_N^{*2}) \;-\; \xi_0 W \,\big]
\end{align}
When the detuning is assumed to be identical across all oscillators, the SBM dynamics reduce to those of the CIM. Consequently, the parameter gap in the SBM dynamics can be expressed as,
\begin{equation}
\Delta _{\text{SBM}}=p_{\{\text{Ising}\}}-p_{\{0\}}=- \lambda_{\text{max}}\big(- \xi_0 W - 3 K_e \mathrm{diag}(x_1^{*2}, \ldots, x_N^{*2}) \big) +  \lambda_{\text{max}}(-\xi_0 W).  
\end{equation}
Furthermore, using Weyl's inequality, we obtain the following relationship, 
\begin{align*}
    -\lambda_{\max}(- \xi_0 W - 3 K_e\operatorname{diag}(x_{1}^{*2},\hdots,x_{N}^{*2})) +  \lambda_{\max}(-\xi_0 W) &= \lambda_{\min}(\xi_0 W + 3 K_e \operatorname{diag}(x_{1}^{*2},\hdots,x_{N}^{*2})) +  \lambda_{\max}(-\xi_0 W)\\
    & \geq \lambda_{\min}(3 K_e \operatorname{diag}(x_{1}^{*2}, \hdots, x_{N}^{*2}))= 3 K_e \min_{1\leq i\leq n}x_{i}^{*2}\geq 0.
\end{align*}
which shows that the parameter gap is always non-negative.

\newpage 
\section{Supplementary Note 4}
\subsection{Sufficient condition for \boldmath{$v_{\max} \propto \sigma^{\ast}$}}\noindent 
In this section, we establish a sufficient condition under which 
$\boldsymbol{v}_{\max}$---the dominant eigenvector at the bifurcation point---aligns with an Ising ground-state configuration 
$\boldsymbol{\sigma^{\ast}} $, i.e., 
$\boldsymbol{v}_{\max} \propto \boldsymbol{\sigma^{\ast}}$. To establish this sufficient condition, we begin with the Ising Hamiltonian,
\begin{align}
H(\boldsymbol{\sigma}) &= -\frac{1}{2} \boldsymbol{\sigma^T} G \boldsymbol{\sigma}
\end{align}
where, $\boldsymbol{\sigma} \in \{-1, 1\}^N$. As detailed in the main text, we consider antiferromagnetic coupling i.e., $W_{ij}=-G_{ij}$, for which $H(\boldsymbol{\sigma}) = -\frac{1}{2}\boldsymbol{\sigma^{\top}}G\boldsymbol{\sigma} = \frac{1}{2}\boldsymbol{\sigma^{\top}}W\boldsymbol{\sigma}$. Using the relationship, $J_{\text{ts}}^{\ast}\, =\, -D -W \Rightarrow W= -D -J_{\text{ts}}^{\ast}$, $H(\boldsymbol{\sigma})$ can be reformulated as,

\begin{align}
H(\boldsymbol{\sigma}) &= -\frac{1}{2} \boldsymbol{\sigma}^{\top} (J^{\ast}_{\mathrm{ts}}) \boldsymbol{\sigma}-\frac{1}{2} \boldsymbol{\sigma}^{\top}  (D) \boldsymbol{\sigma}.
\end{align}

\noindent Furthermore, using the identity
\(
\boldsymbol{\sigma}^{\top} D \boldsymbol{\sigma}
= \sum_{i=1}^{N} D_i \boldsymbol{\sigma_i^{2}}
= 2\,\#\mathcal{E}(\mathcal{G}),
\)
where \(\#\mathcal{E}(\mathcal{G})\) denotes the number of edges of the graph
\(\mathcal{G}\), we define,
\begin{equation}
H'(\boldsymbol{\sigma}) = H(\boldsymbol{\sigma}) + \#\mathcal{E}(\mathcal{G}),
\end{equation}
which corresponds to an offset version of the original Ising Hamiltonian. Since the matrix \(J^{\ast}_{\mathrm{ts}}\) is symmetric, it admits an orthonormal
set of eigenpairs \(\{(\lambda_i,\boldsymbol{v}_i)\}_{i=1}^{N}\), ordered as
\(\lambda_{\min}=\lambda_1 \le \lambda_2 \le \cdots \le \lambda_N=\lambda_{\max}\).
Accordingly, \(J^{\ast}_{\mathrm{ts}}\) admits the spectral decomposition
\(
J^{\ast}_{\mathrm{ts}} = \sum_{i=1}^{N} \lambda_i \boldsymbol{v}_i \boldsymbol{v}_i^{\top}.
\)
Substituting this decomposition into the definition of \(H'(\boldsymbol{\sigma})\)
and simplifying yields,
\begin{align}
H'(\boldsymbol{\sigma})
&= -\frac{1}{2}\,
\boldsymbol{\sigma}^{\top}
\left(\sum_{i=1}^{N} \lambda_i \boldsymbol{v}_i \boldsymbol{v}_i^{\top}\right)
\boldsymbol{\sigma}
\nonumber \\
&= -\frac{1}{2}\sum_{i=1}^{N} \lambda_i
\left(\boldsymbol{\sigma}^{\top}\boldsymbol{v}_i\right)^2
\nonumber \\
&= -\frac{1}{2}
\left(
\lambda_{\max}\left(\boldsymbol{\sigma}^{\top}\boldsymbol{v}_{\max}\right)^2
+
\sum_{i=1}^{N-1} \lambda_i
\left(\boldsymbol{\sigma}^{\top}\boldsymbol{v}_i\right)^2
\right),
\label{eqn:Hprime_definition}
\end{align}
To simplify further, note that
\begin{equation}
\sum_{i=1}^{N} (\boldsymbol{\sigma}^{\top}\boldsymbol{v}_i)^2
= \|\boldsymbol{\sigma}\|^2 = N,
\end{equation}
since \(\{\boldsymbol{v}_i\}_{i=1}^{N}\) forms an orthonormal basis and
\(\boldsymbol{\sigma} \in \{-1,1\}^{N}\).
Using this identity, Eq.~\eqref{eqn:Hprime_definition} reduces to
\begin{align}
H'(\boldsymbol{\sigma})
&= -\frac{1}{2}
\left(
\lambda_{\max}
\left(
N - \sum_{i=1}^{N-1}
(\boldsymbol{\sigma}^{\top}\boldsymbol{v}_i)^2
\right)
+ \sum_{i=1}^{N-1}
\lambda_i (\boldsymbol{\sigma}^{\top}\boldsymbol{v}_i)^2
\right)
\nonumber \\
&= -\frac{N}{2}\lambda_{\max}
+ \frac{1}{2}\sum_{i=1}^{N-1}
\gamma_i(\boldsymbol{\sigma}^{\top}\boldsymbol{v}_i)^2,
\label{eq:Hising}
\end{align}
where \(\gamma_i = \lambda_{\max} - \lambda_i\) and
\(\gamma_1 \ge \gamma_2 \ge \cdots \ge \gamma_{N-1} \ge 0\). Identifying a spin configuration \(\boldsymbol{\sigma}\) that minimizes
\(H'(\boldsymbol{\sigma})\) is equivalent to minimizing the Ising Hamiltonian
\(H\), and is therefore NP-hard in general.

\noindent Nevertheless, if a spin configuration
\(\boldsymbol{\sigma}^{\ast} \in \{-1,1\}^{N}\) satisfies
\begin{equation}
(\boldsymbol{\sigma}^{\ast\top}\boldsymbol{v}_i)^2 = 0,
\qquad i = 1,2,\ldots,N-1,
\label{eqn:sufficiency_condition}
\end{equation}
then \(H'(\boldsymbol{\sigma}^{\ast}) = -\tfrac{N}{2}\lambda_{\max}
\le H'(\boldsymbol{\sigma})\) for all other \(\boldsymbol{\sigma}\),
and \(\boldsymbol{\sigma}^{\ast}\) is an Ising ground state.
Consequently, for
\(\boldsymbol{\sigma}_{\max} = \mathrm{sign}(\boldsymbol{v}_{\max})\)
to be a ground-state configuration, a sufficient condition is
\begin{equation}
(\boldsymbol{\sigma}_{\max}^{\top}\boldsymbol{v}_i)^2 = 0,
\qquad i = 1,2,\ldots,N-1.
\label{eq:sufficient_condition}
\end{equation}
We show that this condition Eq. \ref{eq:sufficient_condition} is satisfied for bipartite graphs.\\

\noindent 
\textbf{Bipartite graphs.}  
Let $\mathcal{G} = (\mathcal{V}, \mathcal{E})$ be a connected bipartite graph, and let $G \in \mathbb{R}^{N \times N}$ denote its symmetric adjacency matrix. 
\begin{lemma}
    Let $(\boldsymbol{v}_i)_{i=1}^N$ be the eigenvectors of the Jacobian matrix $J_{ts}^{*}$. We establish the following properties:
\begin{itemize}
  \item[(a)] The eigenvectors $(\boldsymbol{v}_i)_{i=1}^N$ of $J_{ts}^{*}$ form an orthonormal basis of $\mathbb{R}^N$.
  \item[(b)] The sign vector $\boldsymbol{\sigma}_{\max}$ coincides (up to normalization) with the dominant eigenvector $\boldsymbol{v}_{\max}$ of $J_{ts}^{*}$.
\end{itemize}
\label{lemma: sufficient_conditions_at_trivial_equilibrium}
\end{lemma}
An immediate consequence of  Lemma~\ref{lemma: sufficient_conditions_at_trivial_equilibrium}, it follows that
\[
\bigl(\boldsymbol{\sigma}_{\max} \cdot \boldsymbol{v}_i\bigr)^2 =
\begin{cases}
0, & i = 1, 2, \ldots, N-1, \\[6pt]
N, & i = N,
\end{cases}
\]
which implies that the sufficient condition stated in Eq.~\eqref{eq:sufficient_condition} is satisfied. As a result the leading eigenvector at the point of bifurcation can be used to obtain the Ising ground state in bipartite graphs. 

\begin{proof}
\mbox{}\\
\begin{enumerate}
    \item[\textbf{(a)}] Since $G$ is undirected, its weight matrix $W$ is symmetric, i.e., $W =W^{\top}$. Consequently, $J_{ts}^{*} = -D - W$ is also symmetric. By the spectral theorem for real symmetric matrices, $J_{ts}^{*}$ is diagonalizable and admits an an orthonormal basis of eigenvectors in $\mathbb{R}^n$. Eigenvectors associated with distinct eigenvalues are orthogonal, and for repeated eigenvalues, the eigen-space admits an orthonormal bases. Therefore, the eigenvectors of $J_{ts}^{*}$ can be chosen to form an orthonormal basis of $\mathbb{R}^n$, establishing property (a).

\item[\textbf{(b)}]
For a bipartite graph $G=(V_1 \cup V_2,E)$, the matrix $J_{ts}^{\ast}$ admits the block decomposition (refer to the critical points in the bifurcation dynamics of the DIM in the Supplementary Note~2)
\[
J_{ts}^{\ast}
= \begin{pmatrix}
-D_1 & -W \\
-W^{\top} & -D_2
\end{pmatrix}=
\begin{pmatrix}
-D_1 & -B \\
-B^{\top} & -D_2
\end{pmatrix},
\]
where $B$ denotes the biadjacency matrix and $D_1$, $D_2$ are diagonal degree matrices
associated with the vertex sets $V_1$ and $V_2$, respectively.

Define the sign vector $\boldsymbol{\sigma} \in \mathbb{R}^N$ by
\[
\boldsymbol{\sigma_i} =
\begin{cases}
+1, & i \in V_1,\\
-1, & i \in V_2.
\end{cases}
\]

Since $G$ is bipartite, all the neighbors of any vertex $i \in V_1$ lie in $V_2$, and vice versa.
As a direct consequence of this bipartite structure, the action of $J_{ts}^{\ast}$ on $\boldsymbol{z}$ satisfies
\[
J_{ts}^{\ast} \boldsymbol{\sigma}
=
-\begin{pmatrix}
D_1 & B \\
B^{\top} & D_2
\end{pmatrix}
\begin{pmatrix}
\mathbf{1}_m \\
-\mathbf{1}_n
\end{pmatrix}
=-\begin{pmatrix}
D_1 \mathbf{1}_m - B\mathbf{1}_n \\
B^{\top}\mathbf{1}_m - D_2 \mathbf{1}_n
\end{pmatrix}=
0.
\]

Hence, $\boldsymbol{\sigma}$ is an eigenvector of $J_{ts}^{\ast}$ associated with the eigenvalue $0$.
Since $J_{ts}^{\ast}$ is negative semidefinite—being the negative of the signless Laplacian—this eigenvalue is the largest eigenvalue of the matrix. Therefore, $\boldsymbol{\sigma}$ coincides with the dominant eigenvector 
$\boldsymbol{v}_{\max}$ of $J_{ts}^{\ast}$ (up to normalization), implying 
$\boldsymbol{\sigma} = \boldsymbol{\sigma}_{\max}$.
This completes the proof of property~(b).
\end{enumerate}
\end{proof}

\textbf{Simulations}: From a functional standpoint, when the dominant eigenvector aligns with an optimal Ising ground-state configuration, the dynamics are expected to converge to the ground state with high probability. To corroborate this, we analyzed a representative randomly-generated bipartite graph with $100$ nodes. The graph was solved using each of the four models—OIM, DIM, CIM, and SBM—over $50$ independent trials using the same noise amplitude ($A_n$). In every case, all models consistently found the ground state (\text{Ising-Energy} $= -764$). The corresponding results and the simulation parameters have been summarized in Table \ref{tb:bipartite}.

\begin{table}[h]
\centering
\setlength{\tabcolsep}{12pt}        
\renewcommand{\arraystretch}{1.15}  
\caption{\justifying Results for a 100-node bipartite graph. Across 50 independent trials, DIM, OIM, CIM, and SBM consistently converged to the optimal Ising-energy (=$-$764).}
\label{tab:model_results}
\begin{tabular}{c c c c c}
\hline
\textbf{Model} & \textbf{Ising-Energy} & \textbf{Count} & 
\textbf{Parameters} & 
\textbf{Control Parameter ($\eta(t)$)} \\ 
\hline

DIM & -764 & 50 &
\begin{tabular}[c]{@{}c@{}}
$K = 1$ \\[2pt]
$t_{\text{stop}} = 40$ \\[2pt]
$A_n = 10^{-3}$
\end{tabular}
& $K_s(t) = \dfrac{4t}{t_{\text{stop}}}$ \\ 
\hline

OIM & -764 & 50 &
\begin{tabular}[c]{@{}c@{}}
$K = 1$ \\[2pt]
$t_{\text{stop}} = 40$ \\[2pt]
$A_n = 10^{-3}$
\end{tabular}
& $K_s(t) = \dfrac{4t}{t_{\text{stop}}}$ \\ 
\hline 

CIM & -764 & 50 &
\begin{tabular}[c]{@{}c@{}}
$\xi = 0.05$ \\[2pt]
$t_{\text{stop}} = 40$ \\[2pt]
$A_n = 10^{-3}$
\end{tabular}
& $p(t) = \dfrac{4t}{t_{\text{stop}}}$ \\ 
\hline

SBM & -764 & 50 &
\begin{tabular}[c]{@{}c@{}}
$K_e = 1.0$ \\[2pt]
$\xi_0 = 0.05$ \\[2pt]
$\Delta_i^{\textnormal{SBM}} = 1, \forall i$ \\[2pt]
$t_{\text{stop}} = 40$ \\[2pt]
$A_n = 10^{-3}$
\end{tabular}
& $p(t) = \dfrac{4t}{t_{\text{stop}}}$ \\ 
\hline\\

\end{tabular}
\label{tb:bipartite}
\end{table}

\clearpage

\noindent \textbf{Impact of Subdominant Eigen Modes} We also analyze how the weighted projection of the spin configuration selected by the dynamics onto the subdominant eigenspaces impacts the computational performance. We begin by considering the the spectral decomposition of the Ising energy. As discussed earlier, up to an additive
    constant, the Ising energy can be expressed as
    \begin{align}
    H'(\boldsymbol{\sigma})
    &=
    -\frac{N}{2}\lambda_{\max}
    +
    \frac{1}{2}
    \sum_{i=1}^{N-1}
    \gamma_i
    (\boldsymbol{\sigma}^{\top}\boldsymbol{v}_i)^2,
    \label{eq:Hising}
    \end{align}
    where $\gamma_i=\lambda_{\max}-\lambda_i\geq 0$
    and $\{\boldsymbol{v}_i\}$ are the eigenvectors of the Jacobian at the
    bifurcation point. Here, $\boldsymbol{\sigma}\in\{\pm1\}^N$ denotes an arbitrary
    Ising spin configuration. In the context of the dynamics, it may be identified
    with the binarized state selected by the trajectory. Equation~\eqref{eq:Hising} shows that since the first term is independent of \(\boldsymbol{\sigma}\), minimizing the Ising energy is equivalent to minimizing the weighted projection of \(\boldsymbol{\sigma}\) onto the subdominant eigenspaces,
\[
    \sum_{i=1}^{N-1}
    \gamma_i
    (\boldsymbol{\sigma}^{\top}\boldsymbol{v}_i)^2 .
\]
Moreover, for an optimal Ising configuration $\boldsymbol{\sigma^\ast}$, $    \sum_{i=1}^{N-1}
    \gamma_i
    (\boldsymbol{\sigma}^{\ast\top}\boldsymbol{v}_i)^2$ is minimized.  Thus, for a selected trajectory, the relevance of the subdominant modes is manifested through the \emph{excess  weighted subdominant projection} in comparison to an optimal configuration. For example, comparing the configuration predicted by the dominant bifurcation mode,
\(\boldsymbol{\sigma}_{\max}=\operatorname{sign}(\boldsymbol{v}_{\max})\), with an optimal Ising configuration \(\boldsymbol{\sigma}^{\ast}\), gives
\begin{equation}
\begin{split}
\mathcal{H} 
&=
H'(\boldsymbol{\sigma}_{\max})
-
H'(\boldsymbol{\sigma}^{\ast}) \\
&=
\frac{1}{2}
\sum_{i=1}^{N-1}
\gamma_i
\left[
(\boldsymbol{\sigma}_{\max}^{\top}\boldsymbol{v}_i)^2
-
(\boldsymbol{\sigma}^{\ast\top}\boldsymbol{v}_i)^2
\right].
\label{eq:Hising_diff}
\end{split}
\end{equation}
Equation~\eqref{eq:Hising_diff} therefore shows that analog Ising machine performance relative to the optimum is governed by the weighted excess subdominant projection of the dynamically selected configuration relative to the projection of the optimal configuration.\\

\clearpage
\section{Supplementary note 5}
\subsection*{Stochastic mode selection during gap traversal}
\noindent The operational dynamics within the gap are governed by the interplay between noise, the temporal profile of the regularization parameter \(\eta\), and the eigenspectrum of the linearized dynamics. Moreover, the trajectory selected as the system traverses this interval can be viewed as the outcome of a stochastic mode-selection process. To elucidate the interplay, we consider a locally linear ramp,
\[
\eta(t)=\eta_{\min}+\rho_\eta t,
\qquad
\rho_\eta=\frac{\Delta}{\Delta t},
\]
where $\Delta t$ denotes the time taken by the system to traverse the gap $\Delta$ and \(\rho_\eta\) is the ramp rate of the regularization parameter. In close proximity of the trivial-state bifurcation, the noisy linearized dynamics can be expressed as
\begin{equation}
    d(\delta \phi)=J_f(\eta(t))\delta \phi\,dt+A_n\,dB_t ,
\end{equation}
where \(A_n\) is the noise amplitude and \(B_t\) is a standard Wiener process, i.e., Brownian motion. Since the dynamics considered here are gradient flows, \(J_f(\eta)\) is symmetric at the trivial equilibrium and admits an orthonormal eigen-basis. Neglecting eigen-basis rotation terms over the local interval near the crossing, and projecting onto the instantaneous eigenmodes gives
\begin{equation}
    d\Omega_k=\lambda_k(\eta(t))\Omega_k\,dt+A_n\,dB_k(t),
\end{equation}
where \(\Omega_k\) is the modal amplitude of the \(k^{\mathrm{th}}\) eigenmode. For a mode that becomes unstable at \(\eta=\eta_k\), we approximate the corresponding eigenvalue locally as
\begin{equation}
    \lambda_k(\eta)\simeq s_k(\eta-\eta_k),
    \qquad
    s_k=\left.\frac{d\lambda_k}{d\eta}\right|_{\eta=\eta_k}.
\end{equation}
The corresponding crossing time is
\begin{equation}
    t_k=\frac{\eta_k-\eta_{\min}}{\rho_\eta}.
\end{equation}
For modes that cross within the parameter gap, i.e.,
\(
\eta_k<\eta_{\max}
\quad\Longleftrightarrow\quad
t_k<\Delta t,
\)
the modal SDE after crossing becomes
\begin{equation}
    d\Omega_k=s_k\rho_\eta(t-t_k)\Omega_k\,dt+A_n\,dB_k(t),
    \qquad t>t_k .
\end{equation}
For \(t>t_k\), the mean modal amplitude satisfies
\begin{equation}
    \frac{d\langle \Omega_k\rangle}{dt}
    =
    s_k\rho_\eta(t-t_k)\langle \Omega_k\rangle .
\end{equation}
This follows by taking the expectation of the modal SDE and using
\(\langle dB_k(t)\rangle=0\). Therefore,
\begin{equation}
    \langle \Omega_k(t)\rangle
    =
    \Omega_k(t_k)
    \exp\left[
        \frac{s_k\rho_\eta}{2}(t-t_k)^2
    \right],
\end{equation}
and, at the end of the gap traversal,
\begin{equation}
    \langle \Omega_k(\Delta t)\rangle
    =
    \Omega_k(t_k)
    \exp\left[
        \frac{s_k\rho_\eta}{2}(\Delta t-t_k)^2
    \right].
\end{equation}
Thus, the modal mean reflects the deterministic growth of any initial projection along an unstable eigenmode. The variance is obtained by applying It\^{o}'s rule to \(\Omega_k^2\). Writing
\[
q_k(t)=s_k\rho_\eta(t-t_k),
\]
the scalar modal SDE becomes
\[
d\Omega_k=q_k(t)\Omega_k\,dt+A_n\,dB_k(t).
\]
Using
\(
d(\Omega_k^2)=2\Omega_k\,d\Omega_k+(d\Omega_k)^2,
\)
together with
\(
(dB_k)^2=dt,\,\, dt\,dB_k=0,\,\, (dt)^2=0,
\)
gives
\begin{equation}
    \frac{d\langle \Omega_k^2\rangle}{dt}
    =
    2q_k(t)\langle \Omega_k^2\rangle+A_n^2 .
\end{equation}
Since
\(
\mathrm{Var}[\Omega_k]=\langle \Omega_k^2\rangle-\langle \Omega_k\rangle^2\,\,
\)
and
\(
\,\,\frac{d\langle \Omega_k\rangle^2}{dt}
=
2q_k(t)\langle \Omega_k\rangle^2,
\)
the variance satisfies
\begin{equation}
    \frac{d\,\mathrm{Var}[\Omega_k]}{dt}
    =
    2q_k(t)\mathrm{Var}[\Omega_k]+A_n^2 .
\end{equation}
Substituting \(q_k(t)=s_k\rho_\eta(t-t_k)\), we obtain
\begin{equation}
    \frac{d\,\mathrm{Var}[\Omega_k]}{dt}
    =
    2s_k\rho_\eta(t-t_k)\mathrm{Var}[\Omega_k]+A_n^2 .
\end{equation}
Solving this linear ODE from \(t_k\) to \(\Delta t\) gives the general expression
\begin{align}
    \mathrm{Var}[\Omega_k(\Delta t)]
    &=
    \underbrace{
    \mathrm{Var}[\Omega_k(t_k)]
    \exp\left[
        s_k\rho_\eta(\Delta t-t_k)^2
    \right]
    }_{\text{amplified pre-existing fluctuation}}
    \nonumber\\
    &\quad+
    \underbrace{
    A_n^2
    \exp\left[
        s_k\rho_\eta(\Delta t-t_k)^2
    \right]
    \int_{t_k}^{\Delta t}
    \exp\left[
        -s_k\rho_\eta(\tau-t_k)^2
    \right]d\tau
    }_{\text{noise injected during gap traversal}} . \label{eq:var_general}
\end{align}
The first term represents thermal or stochastic fluctuations already present when the mode becomes unstable. The second term represents noise injected after the crossing and amplified during the remaining traversal of the gap.\\[0.8em]
We also consider the role of the local eigengap at the bifurcation,
\[
\gamma_{N-1}=\lambda_N-\lambda_{N-1},
\]
where \(\lambda_N=0\) is the dominant eigenvalue (at bifurcation) and \(\lambda_{N-1}\) is the weakest stable competing eigenvalue. As analyzed in the following section, the r.m.s. fluctuation of the strongest competing stable mode at the bifurcation is
\begin{equation}
    \Omega_{N-1,\mathrm{rms}}=\frac{A_n}{\sqrt{2\gamma_{N-1}}}.
\end{equation}
Thus, for the strongest competing mode,
\begin{equation}
    \mathrm{Var}[\Omega_{N-1}(t_{N-1})]
    \simeq
    \frac{A_n^2}{2\gamma_{N-1}}. \label{eq:var_a2}
\end{equation}
which impacts the pre-crossing fluctuations. Substituting Eq.~\eqref{eq:var_a2} into Eq.~\eqref{eq:var_general} gives,
\begin{align}
    \mathrm{Var}[\Omega_{N-1}(\Delta t)]
    &\simeq
    \frac{A_n^2}{2\gamma_{N-1}}
    \exp\left[
        s_{N-1}\rho_\eta(\Delta t-t_{N-1})^2
    \right]
    \nonumber\\
    &\quad+
    A_n^2
    \exp\left[
        s_{N-1}\rho_\eta(\Delta t-t_{N-1})^2
    \right]
    \int_{t_{N-1}}^{\Delta t}
    \exp\left[
        -s_{N-1}\rho_\eta(\tau-t_{N-1})^2
    \right]d\tau .
\end{align}
This expression separates two effects. The eigengap \(\gamma_{N-1}\) controls the magnitude of the pre-existing fluctuation in the competing mode at the time it becomes unstable through \(\Omega_{N-1,\mathrm{rms}}=A_n/\sqrt{2\gamma_{N-1}}\). The ramp rate \(\rho_\eta=\Delta/\Delta t\) then controls how strongly this fluctuation, together with noise injected during the traversal, is amplified before the Ising-encoded state is stabilized. Thus, a larger eigengap suppresses the initial competing-mode fluctuation, while a faster traversal reduces the time available for that fluctuation to grow.\\[0.8em]
A sufficient condition for robust dominant-mode selection is that the r.m.s. amplitude of each competing mode remain small compared with the readout margin of the dominant mode:
\begin{equation}
    \sqrt{\mathrm{Var}[\Omega_k(\Delta t)]}
    \ll
    |\langle \Omega_N(\Delta t)\rangle|
    \min_i |(v_{\mathrm{max}})_i|,
    \qquad 1\leq k< N.
\end{equation}
Here \(v_{\mathrm{max}}\) denotes the dominant bifurcation eigenvector. This condition states that noise-amplified competing modes must remain small enough that they do not alter the signs of the components selected by the dominant mode. For the strongest competing mode, this condition becomes
\begin{align}
    &\left[
    \frac{A_n^2}{2\gamma_{N-1}}
    \exp\left[
        s_{N-1}\rho_\eta(\Delta t-t_{N-1})^2
    \right]
    +
    A_n^2
    \exp\left[
        s_{N-1}\rho_\eta(\Delta t-t_{N-1})^2
    \right]
    \int_{t_{N-1}}^{\Delta t}
    \exp\left[
        -s_{N-1}\rho_\eta(\tau-t_{N-1})^2
    \right]d\tau
    \right]^{1/2}
    \nonumber\\
    &\hspace{3cm}
    \ll
    |\langle \Omega_N(\Delta t)\rangle|
    \min_i |(v_{\mathrm{max}})_i|.
\end{align}
The above results help elucidate the role of the ramp and the eigenspectrum. A mode that becomes unstable early in the gap has a longer growth time \((\Delta t-t_k)\) and can acquire appreciable amplitude even if it is initially seeded only by noise. A slower ramp increases \(\Delta t\), which can strengthen the growth of the desired dominant mode, but it also increases the time available for competing unstable modes to amplify. Conversely, a faster ramp reduces the residence time in the gap, but can also move the dynamics away from the quasi-static bifurcation picture.\\[0.8em]
Taken together, this analysis delineates when the trajectory remains governed by the dominant bifurcation mode even in the presence of noise and a time-dependent $\eta$. This corresponds to the operational regime in which the structural parameter-gap analysis, eigenvector-alignment condition, and synchronization results remain applicable.

\subsection*{Noise-induced seed fluctuation at the bifurcation}
\noindent Here, we derive the seed fluctuation of the strongest competing stable mode at the bifurcation. Near the trivial-state bifurcation, weakly stable competing modes are especially sensitive to noise, and their excitation can influence the post-bifurcation trajectory. To quantify this effect, we model the noisy dynamics near the trivial equilibrium \(\phi_0\) as
\begin{equation}
    d\phi
    =
    f(\phi,\eta)\,dt
    +
    A_n\,dB_t,
    \qquad
    f(\phi,\eta)=-\nabla_{\phi} E(\phi,\eta;W),
\end{equation}
where \(A_n\) is the scalar noise amplitude and \(B_t\) is a standard Brownian motion. Linearizing around \(\phi_0\), with \(\delta\phi=\phi-\phi_0\), gives
\begin{equation}
    d(\delta\phi)
    =
    J_f(\eta)\delta\phi\,dt
    +
    A_n\,dB_t,
\end{equation}
where
\begin{equation}
    J_f(\eta)
    =
    \left.
    \frac{\partial f}{\partial \phi}
    \right|_{\phi=\phi_0}.
\end{equation}
At the bifurcation point, let the eigenvalues of \(J_f\) be ordered as
\begin{equation}
    \lambda_1
    \le
    \lambda_2
    \le
    \cdots
    \le
    \lambda_{N-1}
    \le
    \lambda_N = \lambda_{\max}=0 .
\end{equation}
The eigenvector \(\boldsymbol{v}_N=\boldsymbol{v}_{\mathrm{max}}\) defines the marginal bifurcating mode, while \(\boldsymbol{v}_{N-1}\) is the weakest stable competing mode. We define the local eigengap as
\begin{equation}
    \gamma_{N-1}
    =
    \lambda_{N}-\lambda_{N-1}
    =
    -\lambda_{N-1}.
\end{equation}
Since the dynamics are a gradient flow, \(J_f\) is symmetric at the equilibrium point and admits an orthonormal eigenbasis. Expanding the perturbation in this basis,
\begin{equation}
    \delta\phi(t)
    =
    \sum_{k=1}^{N} \Omega_k(t)\boldsymbol{v}_k,
    \qquad
    \Omega_k(t)=\boldsymbol{v}_k^{T}\delta\phi(t),
\end{equation}
and approximating the eigenbasis as fixed over a sufficiently small interval near the bifurcation, the projected modal dynamics are
\begin{equation}
    d\Omega_k
    =
    \lambda_k \Omega_k\,dt
    +
    A_n\,dB_k(t),
\end{equation}
where \(B_k(t)=\boldsymbol{v}_k^T B_t\) is the Brownian noise projected onto the \(k^{\mathrm{th}}\) eigenmode. For the strongest competing stable mode, this reduces to
\begin{equation}
    d\Omega_{N-1}
    =
    -\gamma_{N-1}\,\Omega_{N-1}\,dt
    +
    A_n\,dB_{N-1}(t).
\end{equation}
This is an Ornstein--Uhlenbeck process. Its mean evolves as
\begin{equation}
    \mathbb{E}[\Omega_{N-1}(t)]
    =
    \Omega_{N-1}(0)e^{-\gamma_{N-1} t},
\end{equation}
showing that deterministic projections onto the competing mode decay at a rate set by the local eigengap. The stationary variance is
\begin{equation}
    \mathrm{Var}(\Omega_{N-1})
    =
    \frac{A_n^2}{2\gamma_{N-1}}.
\end{equation}
Equivalently, the root-mean-square seed fluctuation in the strongest competing stable mode is
\begin{equation}
    \Omega_{N-1,\mathrm{rms}}
    =
    \frac{A_n}{\sqrt{2\gamma_{N-1}}}.
\end{equation}
Thus, a larger local eigengap suppresses noise-induced excitation of the competing mode, whereas a smaller eigengap increases the seed fluctuation available for subsequent amplification during traversal of the parameter gap.

\subsection{Device-to-device mismatch} 

\noindent Device-to-device mismatch perturbs the Jacobian and therefore shifts the eigenvalue crossings that define the stability thresholds. If
                \(J_f\rightarrow J_f+\delta J_f\), then for a simple eigenvalue of a symmetric
                Jacobian,
                \[
                    \delta\lambda_i\simeq v_i^\top \delta J_f\,v_i .
                \]
                Since a stability threshold \(\eta_i^\ast\) is defined by \(\lambda_i(\eta_i^\ast)=0\),
                a first-order Taylor expansion yields
                \[
                    \delta\eta_i
                    \simeq
                    -\frac{\delta\lambda_i}
                    {\left.\partial\lambda_i/\partial\eta\right|_{\eta_i^\ast}}
                    =
                    -\frac{v_i^\top \delta J_f\,v_i}
                    {\left.\partial\lambda_i/\partial\eta\right|_{\eta_i^\ast}} .
                \]
                Thus mismatch can shift the trivial-state and Ising-state thresholds by
                different amounts, thereby narrowing or broadening the effective parameter gap.

                \subsection{Signal attenuation} 
                \noindent Signal attenuation admits a similar spectral interpretation. Uniform attenuation of all coupling strengths primarily rescales the effective coupling strength and hence the corresponding stability thresholds. When the attenuation is nonuniform, it changes the effective interaction matrix \(W\), thereby modifying the Jacobian spectrum, the width of the parameter gap, and the alignment of the dominant bifurcation mode with low-energy Ising configurations.

\clearpage
\section{Supplementary note 6}
\subsection*{Dynamics of the Hybrid Ising Machine (HyIM)}

\noindent Here, we establish that the HyIM dynamics minimize an energy function
whose minimizers coincide with those of the Ising Hamiltonian.
Consider the phase dynamics
\begin{align}
\frac{d \phi_i}{d t} &=
K\!\sum_{\substack{j=1\\j\neq i}}^{N} W_{ij}
\Bigl[
\alpha \sin(\phi_i+\phi_j)
+(1-\alpha)\sin(\phi_i-\phi_j)
\Bigr]
- K_{s}\sin(2\phi_i),
\label{eq:mixed_dyn}
\end{align}
with \(\alpha \in [0,1]\).
Define the energy function
\begin{align}
E(\boldsymbol{\phi}) &=
\frac{K}{2}\!\sum_{\substack{i,j=1\\j\neq i}}^{N} W_{ij}
\Bigl[
\alpha \cos(\phi_i+\phi_j)
+(1-\alpha)\cos(\phi_i-\phi_j)
\Bigr]
- \frac{K_{s}}{2}\sum_{i=1}^{N}\cos(2\phi_i).
\label{eq:Appendix_mixed_energy}
\end{align}
Computing the gradient of \(E(\boldsymbol{\phi})\) yields
\begin{align}
\frac{\partial E}{\partial \phi_k}
&=
-\frac{K}{2}\!\sum_{\substack{l=1\\l\neq k}}^{N} W_{kl}
\Bigl[
\alpha \sin(\phi_k+\phi_l)
+(1-\alpha)\sin(\phi_k-\phi_l)
\Bigr]
\nonumber \\
&\quad
-\frac{K}{2}\!\sum_{\substack{l=1\\l\neq k}}^{N} W_{lk}
\Bigl[
\alpha \sin(\phi_l+\phi_k)
-(1-\alpha)\sin(\phi_l-\phi_k)
\Bigr]
+ K_s \sin(2\phi_k).
\end{align}
Since \(W_{kl}=W_{lk}\) and
\(\sin(\phi_l\pm\phi_k)=\pm\sin(\phi_k\pm\phi_l)\),
the two sums are identical, giving
\begin{align}
\frac{\partial E}{\partial \phi_k}
=
- K\!\sum_{\substack{l=1\\l\neq k}}^{N} W_{kl}
\Bigl[
\alpha \sin(\phi_k+\phi_l)
+(1-\alpha)\sin(\phi_k-\phi_l)
\Bigr]
+ K_s \sin(2\phi_k).
\end{align}
Comparing with Eq.~\eqref{eq:mixed_dyn}, we obtain the gradient-flow relation
\begin{equation}
\frac{d\phi_k}{dt}
=
-\,\frac{\partial E}{\partial \phi_k}.
\end{equation}
Consequently,
\begin{align}
\frac{dE}{dt}
&=
\sum_{k=1}^{N}
\frac{\partial E}{\partial \phi_k}
\frac{d\phi_k}{dt}
=
-\sum_{k=1}^{N}
\left(\frac{d\phi_k}{dt}\right)^2
\;\le\; 0,
\end{align}
and the system performs gradient descent on
\(E(\boldsymbol{\phi})\). When each phase settles at \(0\) or \(\pi\)
(so that \(\cos(2\phi_i)=1\)),
the energy reduces to
\begin{align}
E(\boldsymbol{\phi})
&=
\frac{K}{2}\!\sum_{\substack{i,j=1\\j\neq i}}^{N} W_{ij}
\Bigl[
\alpha \cos(\phi_i+\phi_j)
+(1-\alpha)\cos(\phi_i-\phi_j)
\Bigr]
-  \frac{K_s}{2}N,
\end{align}
where \(-\frac{K_s}{2} N\) is a constant offset.
By choosing \(K = 1\),
Eq.~\eqref{eq:Appendix_mixed_energy} becomes equivalent to the Ising Hamiltonian
up to an additive constant, and therefore shares the same minimizers.

\clearpage

\section{Supplementary Note 7} 
\subsection*{Properties of parameter gap in the HyIM dynamics}

The HyIM dynamics can be expressed as, 
\begin{align}
    \frac{d \phi_i}{d t} = K\sum_{j=1}^{N}W_{ij}\bigl[\alpha \sin(\phi_{i} + \phi_{j}) +(1-\alpha)\sin(\phi_{i} - \phi_{j})\bigr] - K_{s}\sin(2\phi_{i}).
    \label{eqn: alpha-OIM-DIM}
\end{align}
The Jacobian matrix of this system is given by 
\begin{align}
    J_f(\phi^{\ast}, K_s) = KJ_{\text{HyIM}}(\phi^{\ast},\alpha) - 2K_{s}\operatorname{diag}\big(\cos(2\phi^{\ast}_{1}),\hdots,\cos(2\phi^{\ast}_{N})\big)
\end{align}
where the matrix $J_{\text{HyIM}}(\phi^{\ast},\alpha)$ is defined as, 
\begin{align}
    [J_{\text{HyIM}}(\phi^{\ast}, \alpha)]_{ij} = \begin{cases}
        \alpha K \sum\limits_{k=1}^{N}W_{ik}\cos(\phi^{\ast}_{i} + \phi^{\ast}_{k}) + (1-\alpha)K\sum\limits_{k=1}^{N}W_{ik}\cos(\phi^{\ast}_{i} - \phi^{\ast}_{k}) & i = j\\
        \alpha K W_{ij}\cos(\phi^{\ast}_{i} + \phi^{\ast}_{j}) - (1-\alpha)KW_{ij}\cos(\phi^{\ast}_{i} - \phi^{\ast}_{j}) & i\neq j.
    \end{cases}\label{eqn: D matrix}
\end{align}
It can be observed that 
\begin{align}
    J_{\text{HyIM}}(\phi^{\ast},\alpha) = \alpha J_{\text{DIM}}(\phi^{\ast}, \alpha) + (1-\alpha)J_{\text{OIM}}(\phi^{\ast}, \alpha), \label{eqn: jacobian convex combination}
\end{align}

\noindent implying that the Jacobian is a linear combination of the Jacobians of the OIM and the DIM dynamics. Consequently, the parameter gap can then be defined using the same approach as that used for calculating parameter gap in the OIM and DIM dynamics. Specifically, the critical $K_s$ at which the ground state 
$\boldsymbol{\sigma}^\ast$ becomes stable can be expressed as,
\begin{align}
K^{\{0,\pi\}}_{s,\text{HyIM}}(\alpha)=\min_{\phi \in \text{Ising-ground state}(\mathcal{G})}\frac{K}{2}\lambda_{\max}\big({J_{\text{HyIM}}}(\phi^{\ast},\alpha )\big) = \dfrac{K}{2}\lambda_{\max}\big(J_{\text{HyIM}}^{\{0,\pi\}}(\phi^{\ast},\alpha)\big),
\label{eq:Ks_critical_zero_pi_OIM_DIM}
\end{align} 
while the threshold at which the trivial $(\frac{\pi}{2})\mathbf{1}_{N}$ state loses stability is given by,
\begin{align}
K_{s,\text{HyIM}}^{\{\frac{\pi}{2}\}}(\alpha)=-\frac{K}{2}\lambda_{\max}\big({J_{\text{HyIM}}\big((\frac{\pi}{2}})\mathbf{1}_{N},\alpha\big)\big)=-\frac{K}{2}\lambda_{\max}((1-2\alpha){D} - {W}).
\label{eq:Ks_critical_half_pi_OIM_DIM}
\end{align}
The parameter gap, therefore, can be calculated as,
\begin{align}
\Delta_{\text{HyIM}}(\alpha) &= K^{\{0,\pi\}}_{s,\text{HyIM}}(\alpha) - K_{s,\text{HyIM}}^{\{\frac{\pi}{2}\}}(\alpha) \notag \\
&= \dfrac{K}{2}\,\Big[ \lambda_{\max}\big(J_{\text{HyIM}}^{\{0,\pi\}}(\phi^{\ast},\alpha)\big)
- \big(-\lambda_{\max}((1-2\alpha){D} - {W})\big) \Big].\label{eqn: alpha-gap defined}
\end{align}
Next we analyze key properties of the parameter gap defined above in Eq.~\eqref{eqn: alpha-gap defined}. 
\begin{lemma}
    The parameter gap $\Delta_{\textnormal{HyIM}}(\alpha)$ is non-negative and convex in $\alpha\in[0,1]$.  
\end{lemma}

\begin{proof}
    We first demonstrate that the parameter gap is non-negative. We proceed in a similar manner to the analysis of this property in OIMs and DIMs. Define the matrix,

   \begin{align}
        \Delta J_{\text{HyIM}}(\phi^{\ast}, \alpha) &= J_{\text{HyIM}}^{\{0,\pi\}}(\phi^{\ast},\alpha) + {J_{\text{HyIM}}((\frac{\pi}{2}})\mathbf{1}_{N},\alpha) \nonumber \\ 
      &= J_{\text{HyIM}}^{\{0,\pi\}}(\phi^{\ast},\alpha) +\Big[ (1-2\alpha)D - W\Big] \nonumber \\ 
      &= \begin{cases}
        \alpha\big(\#\mathcal{N}^{+}_{i}(\phi^{\ast}) - \#\mathcal{N}_{i}^{-}(\phi^{\ast})\big) + (1-\alpha)\big(\#\mathcal{N}^{+}_{i}(\phi^{\ast}) - \#\mathcal{N}_{i}^{-}(\phi^{\ast})\big) + (1-2\alpha)D& i = j \\ 
        \alpha W_{ij} + (1-\alpha) (-W_{ij}) - W_{ij}& i\neq j, j\in\mathcal{N}_{i}^{+} \\ 
        \alpha(-W_{ij}) + (1-\alpha)W_{ij} - W_{ij}& i\neq j, j\in\mathcal{N}_{i}^{-} \\ 
        0 & \text{otherwise}
    \end{cases} \nonumber \\ 
      &= \begin{cases}
        \big(\#\mathcal{N}^{+}_{i}(\phi^{\ast}) - \#\mathcal{N}_{i}^{-}(\phi^{\ast})\big)  + (1-2\alpha) \big(\#\mathcal{N}^{+}_{i}(\phi^{\ast}) + \#\mathcal{N}_{i}^{-}(\phi^{\ast})\big)& i = j \\ 
        2\alpha W_{ij} - 2W_{ij}& i\neq j, j\in\mathcal{N}_{i}^{+} \\ 
         - 2\alpha W_{ij}& i\neq j, j\in\mathcal{N}_{i}^{-} \\ 
        0 & \text{otherwise}
    \end{cases} \nonumber \\ 
      &= \begin{cases}
        2\#\mathcal{N}^{+}_{i}(\phi^{\ast})-2\alpha \big(\#\mathcal{N}^{+}_{i}(\phi^{\ast}) + \#\mathcal{N}_{i}^{-}(\phi^{\ast})\big)& i = j \\ 
        2\alpha W_{ij} - 2W_{ij}& i\neq j, j\in\mathcal{N}_{i}^{+} \\ 
         - 2\alpha W_{ij}& i\neq j, j\in\mathcal{N}_{i}^{-} \\ 
        0 & \text{otherwise}
    \end{cases} \label{eqn: jacobian difference OIM-DIM}
    \end{align}

By selecting a vector $\boldsymbol{x}$ whose components satisfy $x_i = x_j$ for all $j \in N_i^{+}(\phi^{\ast})$ and $x_i = -x_j$ for all $j \in N_i^{-}(\phi^{\ast})$, it follows directly that
$\Delta J_{\text{HyIM}}(\phi^{\ast}, \alpha)\,\boldsymbol{x} = 0.$, which implies that $\Delta J_{\text{HyIM}}(\phi^{\ast}, \alpha)$ has a zero eigenvalue.
    Using Weyl's inequality and Eq.~\eqref{eqn: jacobian difference OIM-DIM}, we obtain 
    \begin{align*}\lambda_{\max}\big(J_{\text{HyIM}}^{\{0,\pi\}}(\phi^{\ast},\alpha)\big)
+ \lambda_{\max}((1-2\alpha){D} - {W})\big) \geq \lambda_{\max}\big(J_{\text{HyIM`}}^{\{0,\pi\}}(\phi^{\ast},\alpha) + ((1-2\alpha){D} - {W})\big) \geq 0.
    \end{align*} 
This establishes that the parameter gap $\Delta_{\text{HyIM}}(\alpha)$ is non-negative.
\\

\noindent Next, we show that $\Delta_{\text{HyIM}}(\alpha)$ is convex in $\alpha\in[0,1]$. To establish this, we prove that,

\begin{align}
    G(\beta x + (1-\beta)y) \leq \beta G(x) + (1-\beta)G(y), \forall \;\;x,y,\beta\in [0,1]
\end{align}
Using Eq.~\eqref{eqn: jacobian convex combination}, 
\begin{align*}
    J_{\text{HyIM}}(\phi^{\ast}, \beta x+ (1-\beta)y) = \beta J_{\text{HyIM}}(\phi^{\ast}, x) + (1-\beta)J_{\text{HyIM}}(\phi^{\ast}, y).
\end{align*}

\noindent  Since the largest eigenvalue \(\lambda_{\max}(M)\) is a convex function of the
matrix \(M\), for any \(\beta \in [0,1]\) we have
\begin{equation}
\lambda_{\max}\!\left(
J_{\text{HyIM}}\big(\phi^{\ast},\, \beta x + (1-\beta)y\big)
\right)
\le
\beta\,\lambda_{\max}\!\left(
J_{\text{HyIM}}(\phi^{\ast},x)
\right)
+
(1-\beta)\,\lambda_{\max}\!\left(
J_{\text{HyIM}}(\phi^{\ast},y)
\right).
\end{equation}
It follows that \(G(\alpha)\) is convex for \(\alpha \in [0,1]\),
which concludes the proof.
\end{proof}

\subsection{Threshold $\alpha_c$ to achieve \boldmath$K_{s}^{\{\frac{\pi}{2}\}}\geq 0$} 

\noindent We establish the critical value of \(\alpha\) that ensures stability of the trivial
state in the absence of regularization, i.e., for \(K_s = 0\).
Ensuring \(K_s \ge 0\) is significant because, under this condition, the dynamics
evolve along the dominant eigenmode at the point where the trivial state loses
stability. We first show that the critical threshold
\(K_{s,\text{HyIM}}^{\{\frac{\pi}{2}\}}(\alpha)\) is nondecreasing in \(\alpha\).
Recall that, for the HyIM dynamics, the Jacobian evaluated at the trivial
state yields
\begin{equation}
K_{s,\text{HyIM}}^{\{\frac{\pi}{2}\}}(\alpha)
=
-\frac{K}{2}\,
\lambda_{\max}\!\big((1-2\alpha)D - W\big).
\end{equation}
Thus, monotonicity of
\(K_{s,\text{HyIM}}^{\{\frac{\pi}{2}\}}(\alpha)\)
follows from the monotonicity of
\(\lambda_{\max}((1-2\alpha)D - W)\). Consider any \(0 \le \alpha_1 \le \alpha_2 \le 1\).
Then,
\begin{align}
\lambda_{\max}\!\big((1-2\alpha_1)D - W\big)
- \lambda_{\max}\!\big((1-2\alpha_2)D - W\big)
&=
\lambda_{\max}\!\big((1-2\alpha_1)D - W\big)
+ \lambda_{\min}\!\big(-(1-2\alpha_2)D + W\big)
\nonumber \\
&\ge
\lambda_{\min}\!\big(2(\alpha_2-\alpha_1)D\big)
\nonumber \\
&=
2(\alpha_2-\alpha_1)\min_{1\le i\le N} d_i
\;\ge\; 0,
\label{eq:inequality_alpha}
\end{align}
where, the inequality follows from Weyl’s eigenvalue inequality.
Equation~\eqref{eq:inequality_alpha} shows that
\(\lambda_{\max}((1-2\alpha)D - W)\) is nonincreasing in \(\alpha\), and hence
\(K_{s,\text{HyIM}}^{\{\frac{\pi}{2}\}}(\alpha)\) is nondecreasing in \(\alpha\). To define the critical value \(\alpha_c \in [0,1]\), note that
\begin{equation}
K_{s,\text{HyIM}}^{\{\frac{\pi}{2}\}}(0)
= K_{s,\mathrm{OIM}}^{\{\frac{\pi}{2}\}} < 0,
\qquad
K_{s,\text{HyIM}}^{\{\frac{\pi}{2}\}}(1)
= K_{s,\mathrm{DIM}}^{\{\frac{\pi}{2}\}} > 0.
\end{equation}
Since eigenvalues depend continuously on \(\alpha\), there exists a unique
\(\alpha_c \in [0,1]\) such that
\begin{equation}
K_{s,\text{HyIM}}^{\{\frac{\pi}{2}\}}(\alpha_c) = 0.
\end{equation}
Finally, the optimal mixing parameter that minimizes the parameter gap is obtained
by solving
\begin{equation}
\alpha^{\ast}
=
\underset{\alpha \in [\alpha_c,\,1]}{\operatorname*{arg\,min}}
\;
\Delta_{\text{HyIM}}(\alpha).
\label{eqn:final_optimization_problem}
\end{equation}

\newpage

\clearpage

\section{Supplementary note 8} 
\subsection*{Parameter-Gap and Certification Characteristics of HyIM for Planted Frustrated Instances}
\noindent Here, we analyze the evolution of the minimum achievable parameter gap in the HyIM \((\Delta_{\mathrm{HyIM}}^{\min})\) and its dependence on $\alpha$ and the resulting certification using planted-instances. We use planted instances with real signed symmetric weights since computing \(\Delta_{\mathrm{HyIM}}(\alpha)\) requires evaluating \(\eta_{\max}(\alpha)\)-- the upper boundary of the parameter gap, which depends on the stability of Ising ground-state configurations and is therefore generally intractable for arbitrary graphs. In the planted instances, a ground-state configuration is known by construction, enabling direct evaluation of the corresponding parameter gap. Nevertheless, since not all degenerate ground-state configurations are necessarily known, the computed \(\eta_{\max}\), and hence \(\Delta_{\mathrm{HyIM}}^{\min}\), should be interpreted as an upper bound on the true gap because the exact \(\eta_{\max}(\alpha)\) involves minimization over all ground-state configurations.
\begin{figure}[!h]
    \centering
    \includegraphics[width=.89\linewidth]{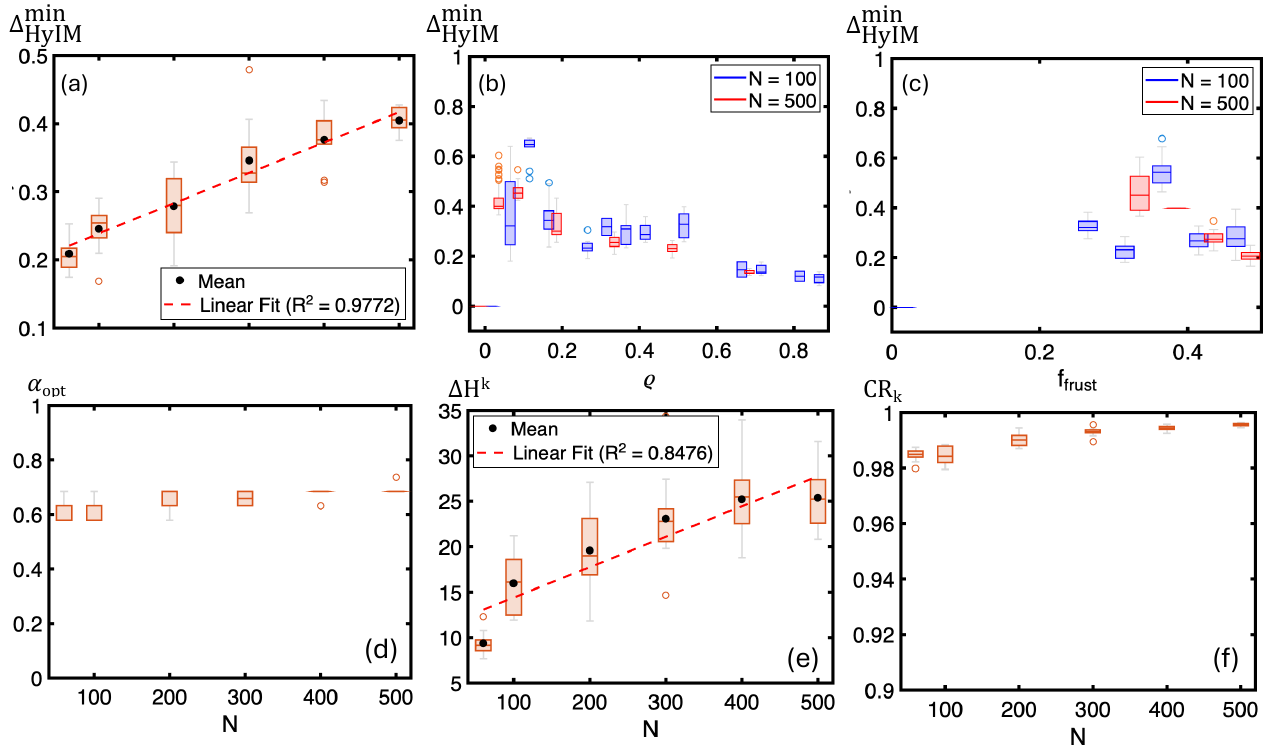}
   \caption{\justifying \smaller
(a) Statistics of the minimum parameter gap, \(\Delta_{\mathrm{HyIM}}^{\min}\), as a function of graph size \(N\). Black markers denote the mean across 10 instances, and the red dashed line shows a linear fit to the mean values, with the corresponding coefficient of determination \((R^2)\). 
(b) Dependence of \(\Delta_{\mathrm{HyIM}}^{\min}\) on graph density \(\varrho\) for \(N=100\) and \(N=500\). 
(c) Dependence of \(\Delta_{\mathrm{HyIM}}^{\min}\) on the frustration metric \(f_{\mathrm{frust}}\) for \(N=100\) and \(N=500\). 
(d) Graph-size dependence of the optimal interpolation parameter \(\alpha_{\mathrm{opt}}\), corresponding to the minimum parameter gap.
(e) Distribution of the certified energy deviation, \(\Delta H^{k}=H^{k}-H_{0}\), as a function of \(N\). Black markers denote the mean across instances, and the red dashed line shows a linear fit to the mean values. (f) Distribution of the certifiable ratio ($CR_K$) as a function of graph size $N$. 
}
    \label{fig:planted_instances}
\end{figure}
While a detailed pseudocode has been presented in Algorithm~\ref{algo:planted_instances} below, the construction of the planted-solution instances can be summarized as follows. We use connected tile-planted instances in which the edge set is decomposed into frustrated motifs with known local optima. The planted spin configuration minimizes every motif simultaneously, thereby certifying the global optimum by construction. The instances considered here use real and symmetric weights. Fig.~\ref{fig:planted_instances}(a) shows the minimum HyIM gap ($\Delta_{\mathrm{HyIM}}^{\min}$) for $\alpha \in [\alpha_c,1]$ as a function of graph size, computed over 10 planted instances for each \(N\in\{60,100,200,300,400,500\}\). The instances used in the graph-size scaling study were generated with a tile size of 20. The observed trend indicates that the minimum achievable gap increases approximately linearly with graph size for these planted instances.\\[0.8em]
To quantify the impact of graph connectivity, we evaluate $\Delta_{\mathrm{HyIM}}^{\min}$ as a function of the edge density
\(
    \varrho=\frac{\#\mathcal{E}(\mathcal{G})}{N(N-1)/2} .
\)
Fig.~\ref{fig:planted_instances}(b) shows \(\Delta_{\mathrm{HyIM}}^{\min}\) as a function of \(\varrho\) for planted instances with \(N=100\) and \(N=500\). Fig.~\ref{fig:planted_instances}(c) shows the dependence of \(\Delta_{\mathrm{HyIM}}^{\min}\) on the frustration in the graph for planted instances of the same size. For weighted signed graphs, we quantify frustration as the weighted fraction of interaction strength that remains unsatisfied by the planted optimum,
{\small \[
    f_{\mathrm{frust}}
    =
    \frac{
    \sum_{i<j}
    |W_{ij}|\,
    \mathbb{I}\!\left[W_{ij}\sigma_i^{\ast}\sigma_j^{\ast}>0\right]
    }{
    \sum_{i<j}|W_{ij}|
    } .
\]}
$f_{\mathrm{frust}} \in [0,0.5]$ with $f_{\mathrm{frust}}=0$ indicating a bipartite graph. In both the cases, the tile size and the probability of forming an edge between two tiles are varied to tune $\varrho$ and $f_{\mathrm{frust}}$. 
It can be observed that the dependence of \(\Delta_{\mathrm{HyIM}}^{\min}\) on edge density is non-monotonic over the sampled range: the minimum attainable gap increases slightly at low densities, but decreases as the graph becomes denser. No clear monotonic trend is observed with the weighted frustration level \(f_{\mathrm{frust}}\) for the instances considered here. This indicates that the optimal operating regime depends on graph structure and that intermediate-\(\alpha\) behavior is problem dependent rather than universal across all instances. Fig.~\ref{fig:planted_instances} shows the optimal parameter $\alpha_{\mathrm{opt}}$ resulting in minimum parameter gap (panel (d)), the corresponding statistics and their scaling for the certified energy deviation, $\Delta H^{k}=H^{k}-H_{0}$ (panel (e)), and the certifiable ratio
    \(
        CR_K =
        \frac{C_{\mathrm{cert}}}{C_{\mathrm{ref}}}
    \) 
    (panel (f)), where \(C_{\mathrm{ref}}\) denotes the certified optimal cut or the best-known cut, depending on the instance class. For the planted instances, \(C_{\mathrm{ref}}\) denotes the certified optimal cut.

{ \smaller 
\noindent\rule{\textwidth}{0.5pt}
\vspace{0.3em}
\refstepcounter{algorithm}
\label{algo:planted_instances}
\noindent\textbf{Algorithm S\arabic{algorithm}. Generating Planted Frustrated Ising Instances}\\[-0.7em]
\noindent\rule{\textwidth}{0.5pt}
\begin{algorithmic}[1]
\Require Number of nodes $N$; tile size $q_t$ with $q_t \mid N$; intra-tile AFM weight $w_{\mathrm{tile}}>0$; forced connector weight $\lambda_{\mathrm{forced}}>0$; extra positive-edge probability/weight $p_{\mathrm{ep}},\lambda_{\mathrm{ep}}>0$; extra signed-edge probability/weight $p_{\mathrm{es}},\lambda_{\mathrm{es}}>0$; tolerance $\epsilon$.
\Ensure Interaction matrix $W$; signed adjacency matrix $G=-W$; planted spin configuration $\boldsymbol{\sigma}^*$; planted energy $H^*$.
\State Set $M \gets N/q_t$.
\State Partition the $N$ nodes into $M$ disjoint tiles $T_1,T_2,\ldots,T_M$, each containing $q_t$ nodes.
\State Initialize $W \gets 0_{N\times N}$, $\boldsymbol{\sigma}^* \gets 0_N$, and $H^* \gets 0$.
\For{$m=1,2,\ldots,M$}
    \If{$q_t$ is even}
        \State Assign $q_t/2$ spins in $T_m$ to $+1$ and $q_t/2$ spins to $-1$.
    \Else
        \State Assign $(q_t+1)/2$ spins in $T_m$ to $+1$ and $(q_t-1)/2$ spins to $-1$.
    \EndIf
    \State Randomly permute the spin labels within $T_m$, then apply a random gauge flip $g\in\{-1,+1\}$ to all spins in $T_m$.
    \State Store the resulting spins in $\boldsymbol{\sigma}^*$.
    \For{each unordered pair $(i,j)$ with $i,j\in T_m$, $i<j$}
        \State Set $W_{ij} \gets w_{\mathrm{tile}}$ and $W_{ji} \gets w_{\mathrm{tile}}$.
    \EndFor
    \State Set $H^* \gets H^* - w_{\mathrm{tile}}q_t/2$ if $q_t$ even, else $H^* \gets H^* - w_{\mathrm{tile}}(q_t-1)/2$.
\EndFor
\For{$m=1,2,\ldots,M-1$}
    \State Choose $i\in T_m$, $j\in T_{m+1}$ with $\sigma_i^*\sigma_j^*=-1$.
    \State Set $W_{ij}=W_{ji}\gets \lambda_{\mathrm{forced}}$; set $H^* \gets H^* - \lambda_{\mathrm{forced}}$.
\EndFor
\For{each unordered pair $(i,j)$, $i<j$, with $W_{ij}=0$ and $\sigma_i^*\sigma_j^*=-1$}
    \State With probability $p_{\mathrm{ep}}$: set $W_{ij}=W_{ji}\gets \lambda_{\mathrm{ep}}$; set $H^* \gets H^* - \lambda_{\mathrm{ep}}$.
\EndFor
\For{each unordered pair $(i,j)$, $i<j$, with $W_{ij}=0$}
    \State With probability $p_{\mathrm{es}}$: set $W_{ij}=W_{ji}\gets -\lambda_{\mathrm{es}}\sigma_i^*\sigma_j^*$; set $H^* \gets H^* - \lambda_{\mathrm{es}}$.
\EndFor
\State Symmetrize $W \gets (W+W^{\top})/2$ and set $W_{ii}\gets 0$ for all $i$.
\State Set $G\gets -W$.
\State Compute $H_{\mathrm{check}} = \frac{1}{2}(\boldsymbol{\sigma}^*)^{\top}W\boldsymbol{\sigma}^*$.
\State Accept the instance only if $|H_{\mathrm{check}}-H^*|<\epsilon$, $W$ is connected, all weighted degrees $d_i\ge0$, and $D+W\succeq0$ and the HyIM gap exhibits an interior minimum with $\alpha_{\mathrm{opt}}>\alpha_c$ (else retry with a new random seed).
\State \Return $W$, $G$, $\boldsymbol{\sigma}^*$, and $H^*$.
\Statex \hrulefill 
\end{algorithmic}}

\subsection*{Other Graph Families}
\noindent Besides the planted instances, we also evaluate other graph families. We consider two categories of graphs:\\

\noindent \textbf{(A) Instances with certified energy optima.} This category includes cubic regular graph instances, and \(4N\)-type M\"obius ladder graphs. For all instances in this category, the ground-state energy was independently certified using Gurobi~\cite{gurobi}. Since all degenerate ground-state configurations were not enumerated by Gurobi, the parameter gap was evaluated using the certified ground-state configuration returned by the solver. As the true gap is defined by a minimum over all ground-state branches, the reported values should be interpreted as upper bounds on the true minimum parameter gap.   

\begin{figure}[!h]
    \centering
    \includegraphics[width=0.5\linewidth]{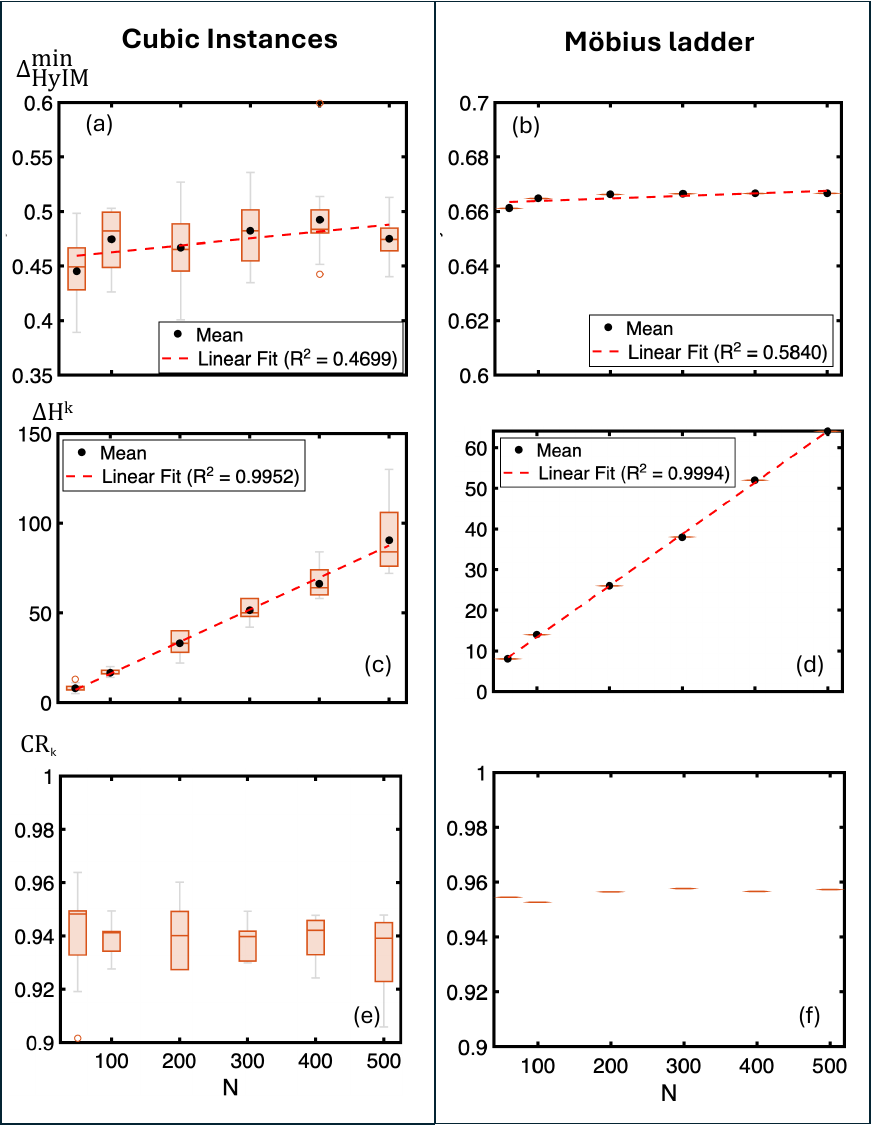}
    \caption{\justifying \smaller Scaling of the minimum HyIM gap, $\Delta_{\mathrm{HyIM}}^{\min}$, the certified energy gap $\Delta H^{k}$, and Certifiable Ratio ($CR_k$) with graph size for (a) cubic regular graph instances, and (b) M\"obius ladder graph instances. Black markers denote the mean across instances, and the red dashed line shows a linear fit to the mean values. }
    \label{fig:Scaling_behavior}
\end{figure}

\noindent \textbf{(B) Instances without certified optima.} This category includes Erd\H{o}s--R\'enyi (ER) graphs with edge probabilities \(p=\{0.05,0.2,0.4\}\) and random 4-regular graphs, where Gurobi could not certify optimality within a reasonable time budget. We therefore used a common best-known reference energy defined as the lowest Ising energy found across OIM, DIM, CIM, SBM, HyIM, and the population-annealing Monte Carlo approach presented in Ref.~\cite{khan2026leveragingpopulationdynamicssteer}. Since the true ground-state energy and the full set of ground-state configurations are not certifiably known for this category, we use the best common best-known energy for the analysis. \\[0.8em]
For each graph class, we evaluated sizes
    \(N=60,100,200,300,400,500\), using 10 independent instances at each size except for the M\"obius ladder, where the graph is fixed for a given \(N\).\\[0.8em] 
    \begin{figure}
    \centering
    \includegraphics[width=1\linewidth]{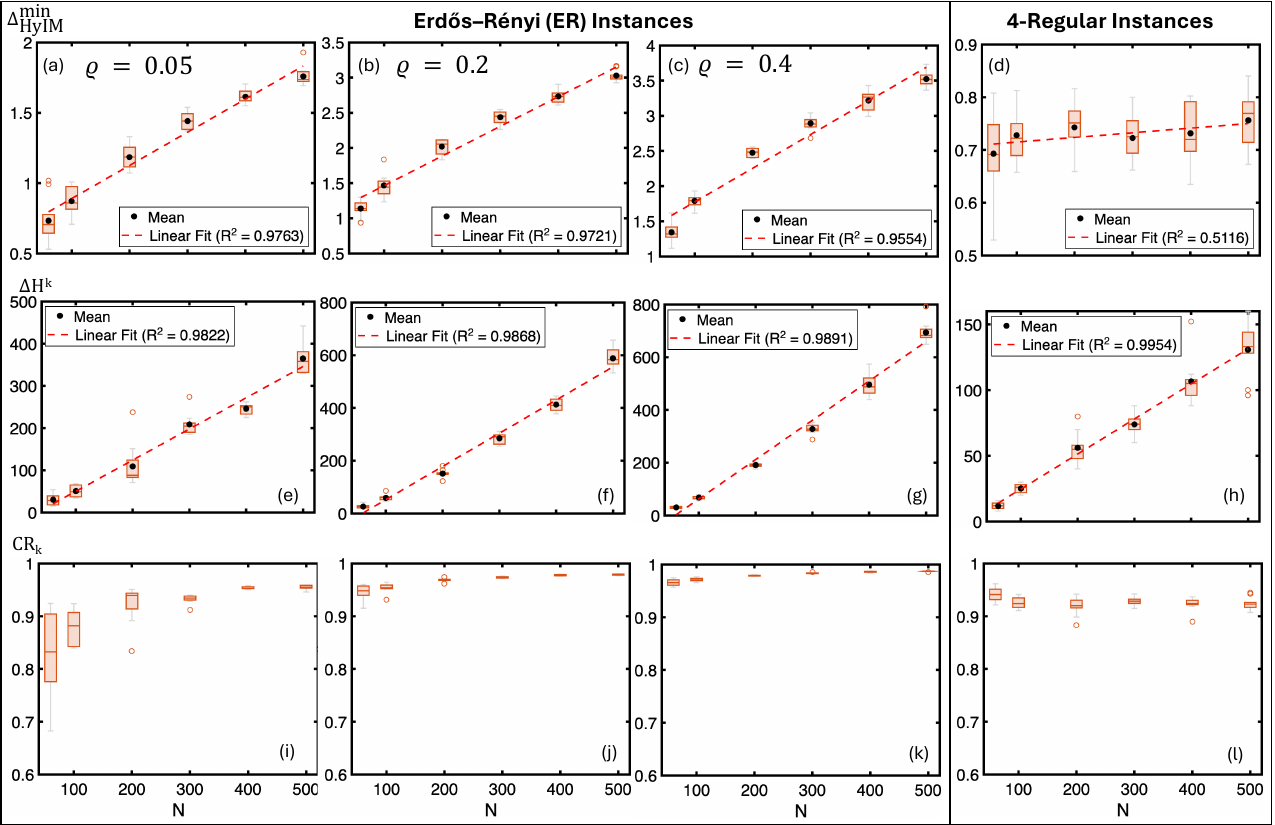}
    \caption{\justifying \smaller Scaling of the minimum HyIM gap, $\Delta_{\mathrm{HyIM}}^{\min}$, the certified energy gap, $\Delta H^{k}$, and the Certifiable Ratio ($CR_k$) with graph size for Erd\H{o}s--R\'enyi (ER) graph instances with edge probabilities $p=0.05$, $p=0.20$, and $p=0.40$, and 4-regular graph instances. Black markers denote the mean across instances, and the red dashed line shows a linear fit to the mean values. }
    \label{fig:Best_cut_from_PAMC}
\end{figure}    
\noindent   Figures~\ref{fig:Scaling_behavior} and~\ref{fig:Best_cut_from_PAMC} report
    the scaling of the minimum achievable HyIM parameter gap
    \(\Delta_{\mathrm{HyIM}}^{\min}\), the corresponding certifiable energy deviation
    \(\Delta H^k\), and the certifiable ratio
    \(
        CR_k =
        \frac{C_{\mathrm{cert}}}{C_{\mathrm{ref}}},
    \) where \(C_{\mathrm{ref}}\) denotes either the certified optimal cut or the best-known cut, depending on the instance class. The observed scaling is graph-dependent. The planted and ER instances show a clear increase of \(\Delta_{\mathrm{HyIM}}^{\min}\) with graph size, whereas the trend is much weaker for the M\"obius ladder family. Correspondingly, the absolute certified energy deviation does not scale uniformly across graph classes. \\[0.8em]
    Nevertheless, the normalized certifiable ratio remains high across the graph classes studied. In particular, for the ER and planted instances it improves with increasing graph size, while for the cubic and M\"obius ladder instances it remains approximately constant over the sampled range. 

\subsection*{Comparison}
\noindent We also compare  the HyIM with the other models, namely DIM, OIM, CIM, and SBM, for all graph classes considered in the above analysis. Table~\ref{tab:N_comparison} summarizes the optimization performance on the instances with certified optima: signed planted instances, cubic regular graphs, and \(4N\)-type M\"obius ladder graphs. For each graph size, the table reports the certified ground-state energy \(H_0\), the lowest Ising energy obtained by each architecture, and the optimal HyIM mixing parameter \(\alpha_{\mathrm{opt}}\); the result is reported as average over 10 instances. Across these instances, HyIM is competitive with the best endpoint model and attains the lowest energy among the tested architectures in a majority of the cases considered.\\[0.8em]

\begin{table}[h]
\centering
\setlength{\tabcolsep}{12pt}
\caption{\justifying Comparison of the optimization performance of DIM, OIM, CIM, SBM, and HyIM for planted instances, cubic regular graph instances, and Möbius ladder graph instances of varying sizes. For each graph size, the table reports the ground-state energy, $H_0$, the lowest Ising energy obtained by each architecture, and the optimal HyIM mixing parameter, $\alpha_{\mathrm{opt}}$; the result is reported as average over 10 instances. Boldface entries denote the ground-state energy, while entries marked with $\star$ indicate the lowest Ising energy achieved among the five Ising machine architectures (DIM, OIM, CIM, SBM, and HyIM).}
\label{tab:N_comparison}

\begin{tabular*}{\textwidth}{@{\extracolsep{\fill}} c c | c c c c | c c}
\hline
\textbf{Size} & $\mathbf{H_{0}}$ & \textbf{DIM} & \textbf{OIM} & \textbf{CIM} & \textbf{SBM} & $\boldsymbol{\alpha_{\mathrm{opt}}}$ & \textbf{HyIM} \\
\hline
\multicolumn{8}{c}{\textbf{Planted Graph Instances}}\\
\hline
60  & \textbf{-38.01}   & -37.20   & -35.61   & -36.15   & -36.72   & 0.62 & $-37.77^{\star}$   \\
100 & \textbf{-77.90}   & -71.15   & -57.89   & -69.84   & -69.82   & 0.61 & $-74.37^{\star}$   \\
200 & \textbf{-213.19}  & -212.79  & -106.21  & -199.90  & -198.97  & 0.65 & $-212.99^{\star}$   \\
300 & \textbf{-418.68}  & $-418.68^{\star}$   & -156.35 & $-418.68^{\star}$   & $-418.68^{\star}$   & 0.66 & $-418.68^{\star}$   \\
400 & \textbf{-677.54}  & $-677.54^{\star}$   & -219.77 & $-677.54^{\star}$   & $-677.54^{\star}$   & 0.67 & $-677.54^{\star}$   \\
500 & \textbf{-1014.10} & $-1014.10^{\star}$   & -286.10 & $-1014.10^{\star}$   & $-1014.10^{\star}$   & 0.69 & $-1014.10^{\star}$   \\
\hline
\multicolumn{8}{c}{\textbf{Cubic Instances}}\\
\hline
60  & \textbf{-72.00}   & -71.80   & -65.00   & -71.00   & -70.80   & 0.98 & $-72.00^{\star}$   \\
100 & \textbf{-123.40}  & $-123.00^{\star}$    & -104.00  & -119.00  & -119.60  & 0.98 & -122.60 \\
200 & \textbf{-250.40}  & -246.80  & -207.60  & -236.40  & -237.80  & 0.98 & $-247.20^{\star}$   \\
300 & \textbf{-377.20}  & -371.80  & -306.40  & -350.20  & -355.60  & 0.98 & $-372.20^{\star}$   \\
400 & \textbf{-504.20}  & -492.20  & -403.40  & -468.60  & -472.20  & 0.98 & $-492.60^{\star}$   \\
500 & \textbf{-629.80}  & $-616.40^{\star}$    & -503.00  & -584.60  & -584.00  & 0.98 & -612.40 \\
\hline
\multicolumn{8}{c}{\textbf{Möbius ladder}}\\
\hline
60  & \textbf{-86.00}  & $-86.00^{\star}$    & -78.00  & $-86.00^{\star}$    & $-86.00^{\star}$    & $\approx 1$ & $-86.00^{\star}$   \\
100 & \textbf{-146.00} & $-146.00^{\star}$  & -138.00 & -138.00 & $-146.00^{\star}$  & $\approx 1$ & $-146.00^{\star}$   \\
200 & \textbf{-296.00} & $-296.00^{\star}$  & -264.00 & -280.00 & -272.00 & $\approx 1$ & $-296.00^{\star}$  \\
300 & \textbf{-446.00} & $-438.00^{\star}$    & -380.00 & -414.00 & -422.00 & $\approx 1$ & $-438.00^{\star}$   \\
400 & \textbf{-596.00} & $-588.00^{\star}$   & -508.00 & -556.00 & -556.00 & $\approx 1$ &  $-588.00^{\star}$  \\
500 & \textbf{-746.00} & $-730.00^{\star}$   & -630.00 & -698.00 & -690.00 & $\approx 1$ & $-730.00^{\star}$   \\

\hline

\end{tabular*}
\end{table}

\begin{table}[!t]
\centering
\setlength{\tabcolsep}{12pt}
\caption{\justifying Comparison of the optimization performance of DIM, OIM, CIM, SBM, and HyIM for Erd\H{o}s--R\'enyi (ER) graph instances with edge densities $\rho=0.05$, $\rho=0.20$, and $\rho=0.40$, together with 4-regular graph instances of varying sizes. For each graph size, the table reports best-known reference energy \(H_{\mathrm{BK}}\), the lowest Ising energy achieved by each architecture, and the optimal HyIM mixing parameter, $\alpha_{\mathrm{opt}}$; the result is reported as average over 10 instances. Boldface entries denote the ground-state energy, while entries marked with $\star$ indicate the lowest Ising energy achieved among the five Ising machine architectures (DIM, OIM, CIM, SBM, and HyIM).}
\label{tab:PAMC_comparison}

\begin{tabular*}{\textwidth}{@{\extracolsep{\fill}} c c | c c c c | c c}
\hline
\textbf{Size} & $\mathbf{H_{\mathrm{BK}}}$ & \textbf{DIM} & \textbf{OIM} & \textbf{CIM} & \textbf{SBM} & $\boldsymbol{\alpha_{\mathrm{opt}}}$ & \textbf{HyIM} \\
\hline
\multicolumn{8}{c}{\textbf{Erd\H{o}s--R\'enyi (ER) graph instances with edge density $\rho=0.05$}}\\
\hline
60  & \textbf{-74.90}    & -74.10    & -64.70    & -72.30    & -72.90    & 0.96 & $-74.90^{\star}$ \\
100 & \textbf{-159.20}   & -156.40   & -124.20   & -155.40   & -154.40   & 0.90 & $-158.40^{\star}$ \\
200 & \textbf{-453.80}   & -435.40   & -313.40   & -441.20   & -438.40   & 0.82 & $-449.20^{\star}$ \\
300 & \textbf{-845.50}   & -796.30   & -536.90   & -824.90   & -811.50   & 0.79 & $-832.30^{\star}$ \\
400 & \textbf{-1305.10}  & -1222.30  & -784.10   & -1268.30  & -1264.30  & 0.74 & $-1288.50^{\star}$ \\
500 & \textbf{-1827.50}  & -1700.70  & -1059.50  & -1784.10  & -1755.10  & 0.74 & $-1802.70^{\star}$ \\
\hline
\multicolumn{8}{c}{\textbf{Erd\H{o}s--R\'enyi (ER) graph instances with edge density $\rho=0.2$}}\\
\hline
60  & \textbf{-134.30}   & -129.90   & -99.50    & -132.10   & -132.50   & 0.78 & $-132.70^{\star}$ \\
100 & \textbf{-298.70}   & -280.10   & -204.10   & -293.50   & -292.50   & 0.74 & $-294.30^{\star}$ \\
200 & \textbf{-853.00}   & -798.20   & -498.80   & $-842.20^{\star}$   & -834.80   & 0.68 & $-842.20^{\star}$ \\
300 & \textbf{-1580.80}  & -1461.40  & -793.20   & -1554.00  & -1547.00  & 0.67 & $-1555.40^{\star}$ \\
400 & \textbf{-2432.50}  & -2237.30  & -1106.30  & -2394.10  & -2371.70  & 0.66 & $-2394.90^{\star}$ \\
500 & \textbf{-3397.90}  & -3117.90  & -1436.70  & -3336.10 & -3310.90  & 0.64 & $-3351.30^{\star}$ \\
\hline
\multicolumn{8}{c}{\textbf{Erd\H{o}s--R\'enyi (ER) graph instances with edge density $\rho=0.4$}}\\
\hline
60  & \textbf{-169.60}   & -159.80   & -125.00   & -167.60   & -167.20   & 0.69 & $-168.00^{\star}$ \\
100 & \textbf{-373.50}   & -347.70   & -230.90   & $-368.90^{\star}$   & $-368.90^{\star}$   & 0.66 & -368.10 \\
200 & \textbf{-1064.40}  & -986.00   & -507.40   & $-1051.80^{\star}$  & -1042.80  & 0.64 & -1049.60 \\
300 & \textbf{-1961.00}  & -1826.60  & -799.20   & -1927.60  & -1927.00  & 0.63 & $-1928.20^{\star}$ \\
400 & \textbf{-3017.70}  & -2787.50  & -972.30   & $-2964.50^{\star}$ & -2934.30  & 0.63 & -2959.90 \\
500 & \textbf{-4208.60}  & -3819.00  & -1268.20  & $-4139.40^{\star}$  & -4095.00  & 0.63 & -4126.20 \\
\hline
\multicolumn{8}{c}{\textbf{4-Regular Instances}}\\
\hline
60  & \textbf{-84.40}   & $-84.40^{\star}$     & -75.60   & $-84.40^{\star}$   & $-84.40^{\star}$   & 0.95 & $-84.40^{\star}$ \\
100 & \textbf{-138.80}  & $-138.00^{\star}$    & -117.60  & -136.00  & -136.00  & 0.95 & $-138.00^{\star}$ \\
200 & \textbf{-286.80}  & $-285.20^{\star}$    & -240.40  & -274.80  & -273.60  & 0.95 & $-285.20^{\star}$ \\
300 & \textbf{-433.20}  & $-429.20^{\star}$    & -346.00  & -407.20  & -407.60  & 0.95 & -428.80 \\
400 & \textbf{-580.40}  & -573.20  & -457.20  & -540.80  & -543.60  & 0.95 & $-573.60^{\star}$\\
500 & \textbf{-726.00}  & $-720.40^{\star}$   & -571.60  & -679.60  & -680.00  & 0.95 & -719.60 \\
\hline
\end{tabular*}
\end{table}
\noindent Table~\ref{tab:PAMC_comparison} presents the corresponding results for the instances without certified optima, including Erd\H{o}s--R\'enyi graphs with edge probabilities \(\rho=0.05\), \(\rho=0.20\), and \(\rho=0.40\), as well as random 4-regular graphs. For these cases, the table reports the common best-known reference energy \(H_{\mathrm{BK}}\), the lowest Ising energy obtained by each architecture, and \(\alpha_{\mathrm{opt}}\). Since the true optimum is not certifiably known for these instances, the comparison should be interpreted relative to \(H_{\mathrm{BK}}\). Within this reference framework, HyIM achieves the lowest energy in a majority of the ER and 4-regular cases.\\[0.8em]

\clearpage

\section{Supplementary note 9} \subsection*{Relationship Between the Parameter Gap and Synchronization at Bifurcation}

\noindent In this section, we establish a quantitative relationship between the parameter gap
and the degree of synchronization at the point of bifurcation.
For the coherent Ising machine (CIM), Wang \emph{et al.}~\cite{2023Bifurcation}
showed that when the degree of bifurcation exceeds a threshold, the binarized dominant
eigenvector coincides with the optimal Ising solution.
Here, we derive an analogous result for the HyIM dynamics.
Specifically, we show that sufficiently strong synchronization at bifurcation ensures
that the binarized dominant eigenvector attains an energy strictly below that of the
first excited Ising state.
We further show that the synchronization threshold increases monotonically with the
HyIM parameter gap, motivating hybrid dynamics to reduce the parameter gap.

\subsubsection{\textbf{Setup and Alignment Metric}}

\noindent Consider the Ising Hamiltonian
\begin{equation}
H(\boldsymbol{\sigma})
=
-\frac{1}{2}\boldsymbol{\sigma}^{\top}G\boldsymbol{\sigma},
\qquad
\boldsymbol{\sigma}\in\{-1,1\}^{N}.
\label{eq:Alignment_hamiltonian}
\end{equation}
Let \(H_0\) and \(H_1\) denote the ground-state and first excited-state energies, and
define the Ising energy gap
\(\Delta H = H_1 - H_0\). Following~\cite{2023Bifurcation}, we quantify synchronization at bifurcation using
\begin{equation}
S(\boldsymbol{v}_{\max})
=
\frac{1}{N}
(\boldsymbol{v}_{\max}^{\top}\boldsymbol{\sigma}_{\max})^2,
\qquad
\boldsymbol{\sigma}_{\max} = \operatorname{sign}(\boldsymbol{v}_{\max}),
\label{eq:s2_def}
\end{equation}
where \(\boldsymbol{v}_{\max}\) is the dominant eigenvector of the
Jacobian evaluated at the trivial solution.

\begin{theorem}
Let the first nontrivial stable state at bifurcation be proportional to the dominant eigenvector 
$\boldsymbol{v}_{\max}$ of the Jacobian associated with the $\mathrm{HyIM}$ dynamics. 
Define the $\mathrm{HyIM}$ parameter gap as
\begin{align}
\Delta_{\mathrm{HyIM}}(\alpha)
&= \frac{K}{2}\Big[
\lambda_{\max}\!\big(J_{\mathrm{HyIM}}^{\{0,\pi\}}(\phi^{\ast},\alpha)\big)
+ 
\lambda_{\max}\!\big((1-2\alpha)D - W\big)
\Big].
\label{eqn: alpha-gap defined20}
\end{align}
Without loss of generality, by setting $K = 1$ and  
$J_{\mathrm{HyIM}}^{\{0,\pi\}}(\phi^{\ast},\alpha)
= J(\phi^{\ast},\alpha)$, 
we define the parameter gap as
\begin{equation}
\Delta_{\mathrm{HyIM}}(\alpha)
=
\frac{1}{2}\Big[
\lambda_{\max}\!\big(J(\phi^{\ast},\alpha)\big)
+
\lambda_{\max}\!\big((1-2\alpha)D - W\big)
\Big].
\end{equation}
If
\begin{equation}
S(\boldsymbol{v}_{\max})
\geq
S_{\mathrm{crit}}(\boldsymbol{v}_{\max})
=
1-\frac{\Delta H}{N\,\Delta_{\mathrm{HyIM}}(\alpha)},
\label{eq:dos_condition}
\end{equation}
then there exists an interval $\alpha \in [0,\alpha_m]$ such that the binarized configuration $
\boldsymbol{\sigma}_{\max}
=
\operatorname{sign}(\boldsymbol{v}_{\max})
$
satisfies
\[
H(\boldsymbol{\sigma}_{\max}) \leq H_1.
\]
\end{theorem}

\begin{proof}
    We show that that the upper bound to $H(\boldsymbol{\sigma}_{max})$ is lesser than a lower bound to $H_{1}$. To do this, we add and subtract the matrix $(1-2\alpha)D$ from Eq. \ref{eq:Alignment_hamiltonian} and rewriting $J_{\text{HyIM}}(\frac{\pi}{2},\alpha) = (1-2\alpha)D -W$, we obtain,
\begin{align}
H(\boldsymbol{\sigma}) &= -\frac{1}{2} \boldsymbol{\sigma}^{\top} J_{\text{HyIM}}(\frac{\pi}{2},\alpha) \boldsymbol{\sigma} + (1-2\alpha)\# \mathcal{E}(\mathcal{G}),
\end{align}
and define

\begin{align}
    H'(\boldsymbol{\sigma}) = H(\boldsymbol{\sigma}) - (1-2\alpha)\#\mathcal{E}(\mathcal{G}),
    \label{eq:off_set_Ising}
\end{align}
$H'$ can be considered as an $H$ with an offset. As the matrix $J_{\text{HyIM}}(\frac{\pi}{2},\alpha)$ is symmetric, its eigenvalue-eigenvector pair $\{(\widetilde\lambda_{i}(\alpha), \boldsymbol{v}_{i})\}_{i=1}^{N}$ with the following ordering $\widetilde\lambda_{\min}(\alpha) = \widetilde\lambda_{1}(\alpha)\leq \widetilde\lambda_{2}(\alpha)\leq \hdots\leq \widetilde\lambda_{N}(\alpha) = \widetilde\lambda_{\max}(\alpha)$.  Thus, $J_{\text{HyIM}}(\frac{\pi}{2},\alpha)$ admits the spectral decomposition $J_{\text{HyIM}}(\frac{\pi}{2},\alpha) = \sum_{i=1}^N \widetilde\lambda_i(\alpha) \boldsymbol{v}_i \boldsymbol{v}_i^{\top}$. Substituting this decomposition, we obtain,

\begin{align}
H'(\boldsymbol{\sigma})  &= -\dfrac{1}{2}  \sum_{i=1}^N \widetilde\lambda_i(\alpha) (\boldsymbol{\sigma}^{\top}\cdot \boldsymbol{v}_i)^2  \label{eq:part1}\\
&= -\dfrac{1}{2}  \left(\widetilde\lambda_{\max}(\alpha) (\boldsymbol{\sigma}^{\top}\cdot \boldsymbol{v}_{\max})^2 + \sum_{i=1}^{N-1} \widetilde\lambda_i(\alpha) (\boldsymbol{\sigma}^{\top}\cdot \boldsymbol{v}_i)^2\right) \notag \\ &= -\dfrac{1}{2} \left( \widetilde\lambda_{\max}(\alpha) \left( N - \sum_{i=1}^{N-1} (\boldsymbol{\sigma}^{\top} \boldsymbol{v}_i)^2 \right) 
    + \sum_{i=1}^{N-1} \widetilde\lambda_i(\alpha) (\boldsymbol{\sigma}^{\top} \boldsymbol{v}_i)^2 \right) \notag \\
    &= -\dfrac{1}{2} N \widetilde\lambda_{\max}(\alpha) + \dfrac{1}{2} \sum_{i=1}^{N-1} \gamma_i (\boldsymbol{\sigma}^{\top} \boldsymbol{v}_i)^2, \label{eq:Hising1}
\end{align}
where $\gamma_i = \widetilde\lambda_{\max}(\alpha)-\widetilde\lambda_i(\alpha)$ and $\gamma_i \geq \gamma_{i+1} \geq 0$, and from Eq.~\eqref{eq:off_set_Ising}, we can get
\begin{align}
    H_{0} \geq -\dfrac{N}{2}\widetilde\lambda_{\max}(\alpha) + (1-2\alpha)\#\mathcal{E}(\mathcal{G}). \label{eqn: inequality on the ground state}
\end{align}
To derive an upper bound on $H'(\boldsymbol{\sigma})$, we further analyze Eq.~\eqref{eq:part1} by adding and subtracting the term $\tfrac{N}{2}\,\lambda_{\max}\!\big(J_{\text{HyIM}}(\phi^{\ast}, \alpha)\big)$. For notational simplicity, we denote $
J_{\text{HyIM}}(\phi^{\ast}, \alpha) \equiv J^{\boldsymbol{\Phi}}_{\text{HyIM}} .
$

\begin{align}
H'(\boldsymbol{\sigma}) &= -\dfrac{1}{2} \boldsymbol{\sigma}^{\top} \Biggr( \lambda_{\max}\big(J^{\boldsymbol{\Phi}}_{\text{HyIM}}\big)I + \sum_{i=1}^N \widetilde\lambda_i(\alpha) \boldsymbol{v}_i \boldsymbol{v}_i^{\top} \Biggr) \boldsymbol{\sigma}  + \lambda_{\max}\big(J^{\boldsymbol{\Phi}}_{\text{HyIM}}\big)\dfrac{N}{2} \notag \\
&=   - \dfrac{1}{2} \lambda_{\max}\big(J^{\boldsymbol{\Phi}}_{\text{HyIM}} \big)\biggr(\sum_{i=1}^{N}(\boldsymbol{v}_{i}^{\top}\boldsymbol{\sigma})^{2}\biggr) -\dfrac{1}{2} \sum_{i=1}^{N}\widetilde\lambda_{i}(\alpha) (\boldsymbol{v}_{i}^{\top} \boldsymbol{\sigma})^{2} +  \lambda_{\max}\big(J^{\boldsymbol{\Phi}}_{\text{HyIM}}\big)\dfrac{N}{2} \notag \\
&= -\dfrac{1}{2} \sum_{i=1}^{N}\Big(\widetilde\lambda_{i}(\alpha) + \lambda_{\max}\big(J^{\boldsymbol{\Phi}}_{\text{HyIM}}\big)\Big)(\boldsymbol{v}_{i}^{\top} \boldsymbol{\sigma})^{2}  + \lambda_{\max}\big(J^{\boldsymbol{\Phi}}_{\text{HyIM}}\big)\dfrac{N}{2} \label{eqn:H'simplificationn} \\
&\leq -\dfrac{1}{2}\Big(\widetilde\lambda_{\max}(\alpha) + \lambda_{\max}\big(J^{\boldsymbol{\Phi}}_{\text{HyIM}}\big)\Big)(\boldsymbol{v}_{\max}^{\top}\boldsymbol{\sigma})^2 + \lambda_{\max}\big(J^{\boldsymbol{\Phi}}_{\text{HyIM}}\big)\dfrac{N}{2} \label{eqn:H'simplification2}
\end{align}
The inequality in Eq.~\eqref{eqn:H'simplification2} follows from Eq.~\eqref{eqn:H'simplificationn} by restricting $\alpha$ to values that satisfy the following inequality. 
\begin{equation}
\begin{aligned}
\sum_{i=1}^{N}
\Big(\widetilde{\lambda}_i(\alpha)
+ \lambda_{\max}\big(J^{\boldsymbol{\Phi}}_{\text{HyIM}}\big)\Big)
(\boldsymbol{v}_i^{\top}\boldsymbol{\sigma})^2
&\geq
\Big(\widetilde{\lambda}_{\max}(\alpha)
+ \lambda_{\max}\big(J^{\boldsymbol{\Phi}}_{\text{HyIM}}\big)\Big)
(\boldsymbol{v}_{\max}^{\top}\boldsymbol{\sigma})^2, 
\\
\sum_{i=1}^{N-1}
\Big(\widetilde{\lambda}_i(\alpha)
+ \lambda_{\max}\big(J^{\boldsymbol{\Phi}}_{\text{HyIM}}\big)\Big)
(\boldsymbol{v}_i^{\top}\boldsymbol{\sigma})^2
&\ge 0
\label{eq:suff}
\end{aligned}
\end{equation}

\noindent We now establish the existence of a critical value $\alpha \in [0,\alpha_m]$ such that Eq.~\eqref{eq:suff} holds. Since $(\boldsymbol{v}_i^{\top}\boldsymbol{\sigma})^2 \ge 0$ for all $i$ and $\lambda_{\max}\big(J^{\boldsymbol{\Phi}}_{\text{HyIM}}\big) \geq 0$, a sufficient condition for the
inequality to hold is
\begin{equation}
\widetilde{\lambda}_{\min}(\alpha)
+ \lambda_{\max}\big(J^{\boldsymbol{\Phi}}_{\text{HyIM}}\big)
\ge 0
\;\;\Longleftrightarrow\;\;
\widetilde{\lambda}_{\min}(\alpha) \ge -\lambda_{\max}\big(J^{\boldsymbol{\Phi}}_{\text{HyIM}})
\label{eq:suf_2}
\end{equation}
Similarly, we can show that $\widetilde\lambda_{\min}(\alpha)$ is decreasing in $\alpha$.
Consider $0 <\alpha_1 < \alpha_2<1$. Then,
\begin{align}
\widetilde\lambda_{\min}\big((1-2\alpha_1)D-W\big)
-
\widetilde\lambda_{\min}\big((1-2\alpha_2)D-W\big) &=
\widetilde\lambda_{\min}\big((1-2\alpha_1)D-W\big)
+
\widetilde\lambda_{\max}\big(-(1-2\alpha_2)D+W\big) \nonumber\\
& \geq
\widetilde\lambda_{\min}\big((1-2\alpha_1)D-W - (1-2\alpha_2)D+W\big)
\nonumber\\
&\geq
\widetilde\lambda_{\min}\big(2(\alpha_2-\alpha_1)D\big)
=
2(\alpha_2-\alpha_1)\min_{1\le i\le N} d_i
\;\ge\;0,
\end{align}
This shows that $\widetilde\lambda_{\min}(\alpha)$ is monotone decreasing in $\alpha$. We further note that when $\alpha = 0$, the HyIM model reduces to the pure OIM model. In this case, the minimum eigenvalue satisfies $\widetilde{\lambda}_{\min}(0) = 0$, (since $D - W \succeq 0).
$ Moreover, $\widetilde\lambda_{\min}(\alpha)$ is concave in $\alpha$, meaning that there exists
a range of  $\alpha \in [0,\alpha_m]$ such that the sufficient condition in Eq. \ref{eq:suf_2}
holds, which implies that there exists a range of $\alpha$ for which the 
inequality in Eq.~\eqref{eqn:H'simplification2} is satisfied.  Therefore, substituting Eq. \ref{eqn:H'simplification2}  to Eq. \ref{eq:off_set_Ising} we get,
\begin{align}
    H(\boldsymbol{\sigma}_{\max}) &= H'(\boldsymbol{\sigma}_{\max}) +(1-2\alpha) \#\mathcal{E}(\mathcal{G}) \notag \\ 
    &\leq -\dfrac{1}{2}\Big(\widetilde\lambda_{\max} + \widetilde\lambda_{\max}\big(J^{\boldsymbol{\Phi}}_{\text{HyIM}}\Big)(\boldsymbol{v}_{\max}^{\top}\boldsymbol{\sigma}_{\max})^{2} + \widetilde\lambda_{\max}\big(J^{\boldsymbol{\Phi}}_{\text{HyIM}})\dfrac{N}{2} +(1-2\alpha) \#\mathcal{E}(G) \notag \\ 
    & \leq -\dfrac{N}{2}\Big(\widetilde\lambda_{\max} + \widetilde\lambda_{\max}\big(J^{\boldsymbol{\Phi}}_{\text{HyIM}}\big)\Big) + \Delta H + \widetilde\lambda_{\max}(J^{\boldsymbol{\Phi}}_{\text{HyIM}})\dfrac{N}{2} + (1-2\alpha) \#\mathcal{E}(\mathcal{G})  \notag \\ 
    & \leq H_{1} - H_{0} - \dfrac{N}{2}\widetilde\lambda_{\max} + (1-2\alpha)\#\mathcal{E}(\mathcal{G}).\label{eq:refref}
\end{align}
Eq.~\ref{eq:refref} is obtained by substituting 
$(\boldsymbol{v}_{\text{max}}^{\top}\boldsymbol{\sigma}_{\max})^2$ using the 
relationships given in Eq.~\ref{eq:s2_def} and Eq.~\ref{eq:dos_condition}.Now we observe the following  inequality using $H_{0}\geq - \frac{N}{2}\widetilde\lambda_{\max} +  (1-2\alpha)\#\mathcal{E}(\mathcal{G})$ from \eqref{eqn: inequality on the ground state}, therefore 
\[
H(\boldsymbol{\sigma}_{\max}) \leq H_{1},
\]
which completes the proof.  
\end{proof}

\subsubsection{\textbf{Interplay Between \(\alpha_c\) and \(\alpha_m\)}}

\noindent Two critical parameters arise in the analysis:
$\alpha_c$, defined as the value at which the trivial state loses stability at $K_s = 0$,
and $\alpha_m$, defined as the maximal value of $\alpha$ for which, given the alignment condition
\begin{equation}
S(\boldsymbol{v}_{\max})
\geq S_{\mathrm{crit}}(\boldsymbol{v}_{\max})=
1-\frac{\Delta H}{N\,\Delta_{\mathrm{HyIM}}(\alpha)},
\label{eq:dos_condition_3}
\end{equation}
the dominant eigenmode satisfies
\[
H(\boldsymbol{\sigma}_{\max}) \leq H_1.
\]
Two distinct cases emerge for critical values of $\alpha_m$ and $\alpha_c$.

\begin{itemize}
\item \textbf{Case 1:} \(\alpha_m > \alpha_c\).  

In this case, there exists an interval 
\(\alpha \in [\alpha_c, \alpha_m]\) 
such that
\[
H(\boldsymbol{\sigma}_{\max}) \leq H_1.
\]
Thus, once the trivial state loses stability at $\alpha_c$, 
the dynamics evolve along the dominant eigenmode toward the ground state. 
Ground-state recovery is therefore guaranteed.

\item \textbf{Case 2:} \(\alpha_m < \alpha_c\).  

In this case, the dominant-eigenmode criterion does not guarantee ground-state recovery at the bifurcation point. 
However, the admissible range of $\alpha_m$ may extend beyond its initial threshold via modifying Eq. \ref{eq:dos_condition_3} as
\begin{align}
    S(\boldsymbol{v}_{\max})
    \geq S_{\mathrm{crit},k}(\boldsymbol{v}_{\max})=
    1 - \dfrac{\Delta H^{k}}{N\,\Delta(\mathcal{G}, \alpha)},
    \label{eqn:doscondition2}
\end{align}
where $H^{k}$ denotes the $k^{\text{th}}$ excited state and 
$\Delta H^{k} = H^{k} - H_{0}$. The extension of $\alpha_m$ can be characterized by the modified Eq. \ref{eq:refref} as
\begin{align}
    H(\boldsymbol{\sigma}_{\max})
    \leq
    H^k - H_{0}
    - \dfrac{N}{2}\widetilde{\lambda}_{\max}
    + (1-2\alpha)\#\mathcal{E}(\mathcal{G}).
    \label{eq:refref2}
\end{align}

As $\alpha$ increases, both $\widetilde{\lambda}_{\max}$ and the term 
$(1-2\alpha)\#\mathcal{E}(\mathcal{G})$ decrease monotonically. 
Consequently, the inequality can be satisfied successively for higher excited states for $k$ such that $\alpha_m^{k} > \alpha_c$, 
if $\boldsymbol{v}_{\max}$ satisfies Eq.~\eqref{eqn:doscondition2}, 
then it must correspond to a spin configuration 
$\boldsymbol{\sigma}_{\max}$ satisfying
\begin{equation}
H(\boldsymbol{\sigma}_{\max}) \leq H^{k}. \label{eq:H_sigma_less_H_k}
\end{equation}

Formally, we define $\alpha_m^{k}$ as follows:

\[
\mathcal{A}_k
=
\left\{
\alpha\in[0,1]:
S(\boldsymbol{v}_{\max}(\alpha))
\ge
1-\frac{\Delta H^k}{N\Delta_{\mathrm{HyIM}}(\alpha)}
\right\}.
\]
then, $\alpha_m^k=\sup \mathcal{A}_k $.
\end{itemize}

We now examine the relationship between $\alpha_c$ and $\alpha_m$ for bipartite graphs.

\begin{lemma}[Critical parameters ($\alpha_c,\alpha_m$) for bipartite graphs]
\label{lemma:alpha_c_alpha_m_bipartite}
Let $\mathcal{G} = (\mathcal{V}, \mathcal{E})$ be a connected bipartite graph with degree matrix $D$ and weight matrix $W$. 
Consider the $\mathrm{HyIM}$ model evaluated at $\phi = \frac{\pi}{2}$. 
\\

Then the following hold:

\begin{enumerate}
    \item[(i)] The critical parameter $\alpha_c$, defined as the $K_s$ at which the trivial state looses stability,
    \[
    \lambda_{\max}\!\big((1-2\alpha)D - W\big) = 0,
    \]
    satisfies
    \[
    \alpha_c = 1.
    \]

    \item[(ii)] The maximal admissible parameter $\alpha_m$, determined by the condition
    \[
    \sum_{i=1}^{N-1}
    \Big(\widetilde{\lambda}_i(\alpha)
    + \lambda_{\max}\big(J^{\boldsymbol{\Phi}}_{\mathrm{HyIM}}\big)\Big)
    (\boldsymbol{v}_i^{\top}\boldsymbol{\sigma})^2
    \ge 0,
    \]
    satisfies $\alpha_m = 1$ and consequently, $\alpha_c= \alpha_m=1$. 
    \end{enumerate}
\end{lemma}

\begin{proof}
\mbox{}\\
\begin{enumerate}
    \item[(a)] \textbf{ Derivation of $\alpha_c$ for a bipartite graph.} The critical value $\alpha_c$ is determined by the condition $K_s = 0$, 
which corresponds to the point where the trivial state loses stability 
and the dominant eigenmode emerges. From the definition,
\begin{equation}
K_{s,\text{HyIM}}^{\{\frac{\pi}{2}\}}(\alpha)
=
-\frac{K}{2}\,
\lambda_{\max}\!\big((1-2\alpha)D - W\big).
\end{equation}
Hence, $\alpha_c$ satisfies
\[
\lambda_{\max}\!\big((1-2\alpha)D - W\big) = 0.
\]
For a bipartite graph, the eigenvector associated with the largest eigenvalue 
is aligned with the ground-state configuration (see Lemma~\ref{lemma: sufficient_conditions_at_trivial_equilibrium} 
in Supplementary Note~3). Due to the block structure of a bipartite graph,
this eigenvector can be written as
\[
\boldsymbol{v}_{\max} \propto \boldsymbol{\sigma}^{\ast}
=
\begin{pmatrix}
\mathbf{1}_m \\
-\mathbf{1}_n
\end{pmatrix}.
\]
We evaluate the action of $(1-2\alpha)D - W$ on this vector:
\begin{align}
\big((1-2\alpha)D - W\big)
\begin{pmatrix}
\mathbf{1}_m \\
-\mathbf{1}_n
\end{pmatrix}
&=
\begin{pmatrix}
(1-2\alpha)D_1 & -W \\
- W^{\top} & (1-2\alpha)D_2
\end{pmatrix}
\begin{pmatrix}
\mathbf{1}_m \\
-\mathbf{1}_n
\end{pmatrix} \notag \\
&=
(1-2\alpha)
\begin{pmatrix}
D_1 \mathbf{1}_m \\
- D_2 \mathbf{1}_n
\end{pmatrix}
+
\begin{pmatrix}
W \mathbf{1}_n \\
- W^\top \mathbf{1}_m
\end{pmatrix}.
\end{align}
Here, $D_1$ and $D_2$ denote the diagonal degree matrices corresponding to the two partitions of the bipartite graph. Since the graph is bipartite, the degree relations satisfy
\[
D_1 \mathbf{1}_m = W \mathbf{1}_n,
\qquad
D_2 \mathbf{1}_n = W^\top \mathbf{1}_m.
\]
Substituting these identities yields
\begin{align}
\big((1-2\alpha)D - W\big)
\boldsymbol{\sigma}^{\ast}
&=
(1-2\alpha)
\begin{pmatrix}
D_1 \mathbf{1}_m \\
- D_2 \mathbf{1}_n
\end{pmatrix}
+
\begin{pmatrix}
D_1 \mathbf{1}_m \\
- D_2 \mathbf{1}_n
\end{pmatrix} \notag \\
&=
2(1-\alpha)
\begin{pmatrix}
D_1 \mathbf{1}_m \\
- D_2 \mathbf{1}_n
\end{pmatrix} \notag \\
&=
2(1-\alpha) D\,\boldsymbol{\sigma}^{\ast}.
\end{align}
Thus, $\boldsymbol{\sigma}^{\ast}$ is an eigenvector with eigenvalue
\[
\lambda_{\max} = 2(1-\alpha)d,
\]
where $d$ denotes the degree (or, more generally, the action of $D$ on the vector). The largest eigenvalue vanishes when $2(1-\alpha_c) = 0$ which implies $\alpha_c =1$.

\item[(b)] \textbf{ Derivation of $\alpha_m$ for a bipartite graph.}\\

The critical value $\alpha_m$ is determined by the sufficient condition
given in Eq.~\eqref{eq:suff}, namely
\begin{equation}
\sum_{i=1}^{N-1}
\Big(\widetilde{\lambda}_i(\alpha)
+ \lambda_{\max}\big(J^{\boldsymbol{\Phi}}_{\text{HyIM}}\big)\Big)
(\boldsymbol{v}_i^{\top}\boldsymbol{\sigma})^2
\ge 0.
\label{eq:suff2}
\end{equation}
From Lemma~\ref{lemma: sufficient_conditions_at_trivial_equilibrium}
(Supplementary Note~3), for bipartite graphs the case $\alpha=1$ (i.e., HyIM reduces to pure DIM)we obtain
\[
(\boldsymbol{v}_i^{\top}\boldsymbol{\sigma})^2 = 0,
\qquad i = 1,2,\ldots,N-1.
\]
We now verify whether the synchronization condition   
\begin{equation}
S(\boldsymbol{v}_{\max})
=\frac{1}{N}
(\boldsymbol{v}_{\max}^{\top}\boldsymbol{\sigma}_{\max})^2
\geq
1-\frac{\Delta H}{N\,\Delta_{\mathrm{HyIM}}(\alpha)}
\label{eq:dos_condition_bipartite}
\end{equation}
holds for bipartite graphs at $\alpha=1$. This is equivalent to showing that
\begin{equation}
\frac{1}{N}
(\boldsymbol{v}_{\max}^{\top}\boldsymbol{\sigma}_{\max})^2
-
1
+
\frac{\Delta H}{N\,\Delta_{\mathrm{HyIM}}(\alpha)}
\geq 0.
\end{equation}
For bipartite graphs, the dominant eigenvector $\boldsymbol{v}_{\max}$ coincides with 
$\boldsymbol{\sigma}_{\max}$ at $\alpha=1$, implying
\[
\frac{1}{N}
(\boldsymbol{v}_{\max}^{\top}\boldsymbol{\sigma}_{\max})^2 = 1.
\]
Hence, the condition reduces to
\begin{equation}
\frac{\Delta H}{N\,\Delta_{\mathrm{HyIM}}(\alpha)} \geq 0.
\end{equation}
Since both $\Delta H$ and $\Delta_{\mathrm{HyIM}}(\alpha)$ are non-negative,
the inequality satisfies. Therefore, the synchronization condition
holds for $\alpha =1 $, implying
\[
\alpha_m = 1.
\]

\end{enumerate}

From parts \textbf{(a)} and \textbf{(b)}, since $\alpha_c = 1$ and $\alpha_m = 1$, we conclude that $\alpha_c = \alpha_m = 1$ which completes the proof.
\end{proof}
\noindent We now illustrate the synchronization measures 
$S(\boldsymbol{v}_{\max})$ and 
$S_{\mathrm{crit}}(\boldsymbol{v}_{\max})$ 
for M\"obius ladder graphs for the first excited state 
as a function of $\alpha \in [0,1]$, 
for system sizes $6 \leq N \leq 16$ in Fig. \ref{fig:mobius}.

\begin{figure}
    \centering
    \includegraphics[ scale = 0.55]{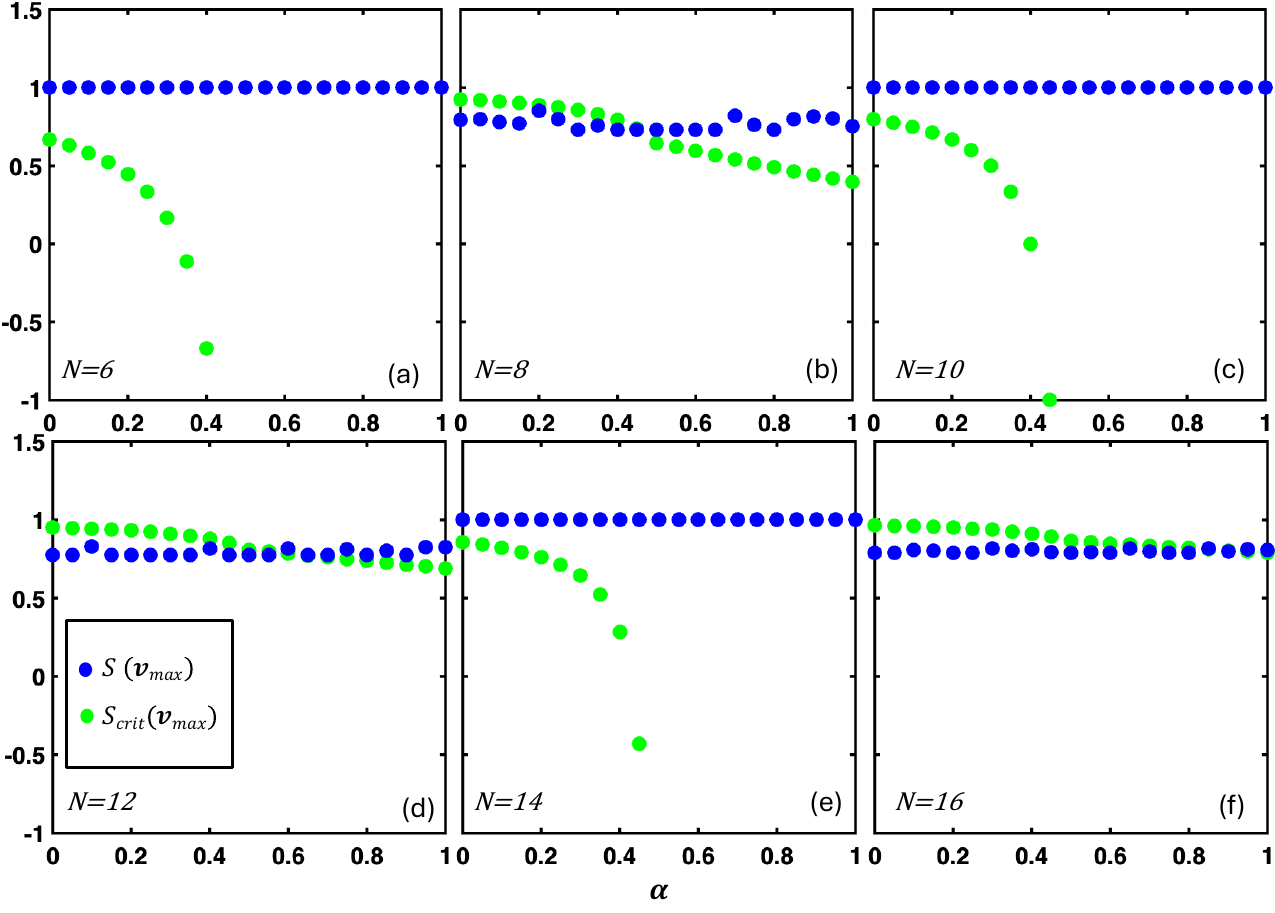}
    \caption{ \justifying \smaller Evolution of $S(\boldsymbol{v_{\mathrm{max}}})$ and $S_{crit}(\boldsymbol{v_{\mathrm{max}}})$ with $\alpha$ for M\"obius ladder graphs of different sizes 
    (a) $N=6$, (b) $N=8$, (c) $N=10$, (d) $N=12$, (e) $N=14$, and (f) $N=16$.
    }

    \label{fig:mobius}
\end{figure}
\vspace{3cm}

\noindent Figure~\ref{fig:random_graph_alignment} compares the synchronization of the dominant bifurcation mode, \(S_{\max}\equiv S(\boldsymbol{v}_{\max})\), with the critical thresholds \(S_{\mathrm{crit},k}\) associated with successive excited-state levels for the graph considered in Fig.~1 of the main manuscript. As \(\alpha\) is varied, \(S_{\max}\) exceeds several of these thresholds over an intermediate range of \(\alpha\), indicating that the bifurcation-selected configuration can be certified below the corresponding energy levels. 

\begin{figure}
    \centering
    \includegraphics[ scale = 0.6]{ 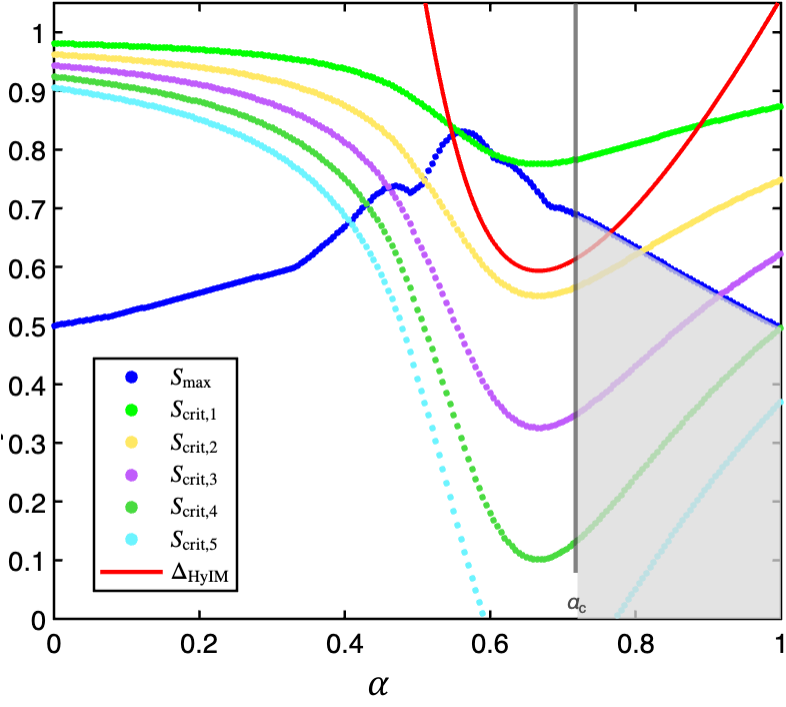}
    \caption{\justifying
     \smaller Evolution of the synchronization metric \(S_{\max}\equiv S(\boldsymbol{v}_{\max})\), the critical synchronization thresholds \(S_{\mathrm{crit},k}\), and the HyIM parameter gap \(\Delta_{\mathrm{HyIM}}\) as functions of the interpolation parameter \(\alpha\). The vertical gray line marks \(\alpha_c\), and the shaded region indicates the admissible HyIM regime \(\alpha\geq\alpha_c\), where the dominant bifurcation mode is selectively destabilized. The graph used in Fig.~1 of the main manuscript is used for the analysis.}

    \label{fig:random_graph_alignment}
\end{figure}

\clearpage
\section{Supplementary note 10}
\subsection{Simulation Parameters and Experimental Settings}
 
This section summarizes the numerical parameters and experimental settings used in all simulations presented in the main text. For each model, we specify the system parameters, initialization conditions, and time-dependent control parameters employed in the simulations. Unless otherwise stated, the parameters listed here are held fixed across all realizations to ensure consistency and reproducibility of the reported results. Table \ref{tb:suppl_table2} summarizes the simulation parameters and control protocols used to generate the results shown in Fig.~1 of the main text. For each model (DIM, OIM, CIM, and SBM).

\begin{table}[h]
\centering
\setlength{\tabcolsep}{12pt}        
\renewcommand{\arraystretch}{1.15}  
\caption{Simulation parameters used in Figure 1 of the main text.}
\label{tab:model_results}
\begin{tabular}{c c c}
\hline
\textbf{Model}  & 
\textbf{Parameters} & 
\textbf{Control Parameter ($\eta(t)$)} \\ 
\hline

DIM &  
\begin{tabular}[c]{@{}c@{}}
$K = 1$ \\[2pt]
$t_{\text{stop}} = 40$ \\[2pt]
$A_n = 10^{-3}$
\end{tabular}
& $K_s(t) = \dfrac{4t}{t_{\text{stop}}}$ \\ 
\hline

OIM &  
\begin{tabular}[c]{@{}c@{}}
$K = 1$ \\[2pt]
$t_{\text{stop}} = 40$ \\[2pt]
$A_n = 10^{-3}$
\end{tabular}
& $K_s(t) = \dfrac{4t}{t_{\text{stop}}}$ \\ 
\hline 

CIM & 
\begin{tabular}[c]{@{}c@{}}
$\xi = 0.05$ \\[2pt]
$t_{\text{stop}} = 40$ \\[2pt]
$A_n = 10^{-3}$
\end{tabular}
& $p(t) = \dfrac{4t}{t_{\text{stop}}}$ \\ 
\hline

SBM & 
\begin{tabular}[c]{@{}c@{}}
$K_e = 1.0$ \\[2pt]
$\xi_0 = 0.05$ \\[2pt]
$\Delta = 1$ \\[2pt]
$t_{\text{stop}} = 40$ \\[2pt]
$A_n = 10^{-3}$
\end{tabular}
& $p(t) = \dfrac{4t}{t_{\text{stop}}}$ \\ 
\hline

\end{tabular}
\label{tb:suppl_table2}
\end{table}
Table \ref{tb:sim_HyIM} summarizes the simulation parameters and control protocols used to generate the results shown in Fig.~2(b) of the main text.

\begin{table}[h]
\centering
\setlength{\tabcolsep}{12pt}        
\renewcommand{\arraystretch}{1.15}  
\caption{Simulation parameters used in Figure 2(b) of the main text.}
\label{tab:model_results}
\begin{tabular}{c c c}
\hline
\textbf{Model}  & 
\textbf{Parameters} & 
\textbf{Control Parameter ($\eta(t)$)} \\ 
\hline

HyIM & 
\begin{tabular}[c]{@{}c@{}}
$K = 1$ \\[2pt]
$t_{\text{stop}} = 4$ \\[2pt]
$A_n = 10^{-4}$ \\[2pt]
$K_s(t_0) = \max\!\left\{0,\; \lambda_{\max}\!\big(J_{\mathrm{HyIM}}(\tfrac{\pi}{2},\alpha)\big) - 3 \right\}.$
\end{tabular}
& $K_s(t) = K_s(t_0)+\dfrac{15t}{t_{\text{stop}}}$ \\ 
\hline

\end{tabular}
\label{tb:sim_HyIM}
\end{table}

\noindent Figure~\ref{fig:G1-G5} shows the distribution of the  $\delta$ as a function of the parameter $\alpha$ for individual graphs G1--G5 \cite{Gset}, complementing the collective distribution shown in Fig.~2(b) of the main text. For each graph, box plots summarize the results from 10 independent simulation trials at each value of $\alpha$. 

\begin{figure}[!h]
    \centering
    \includegraphics[ scale = 0.75]{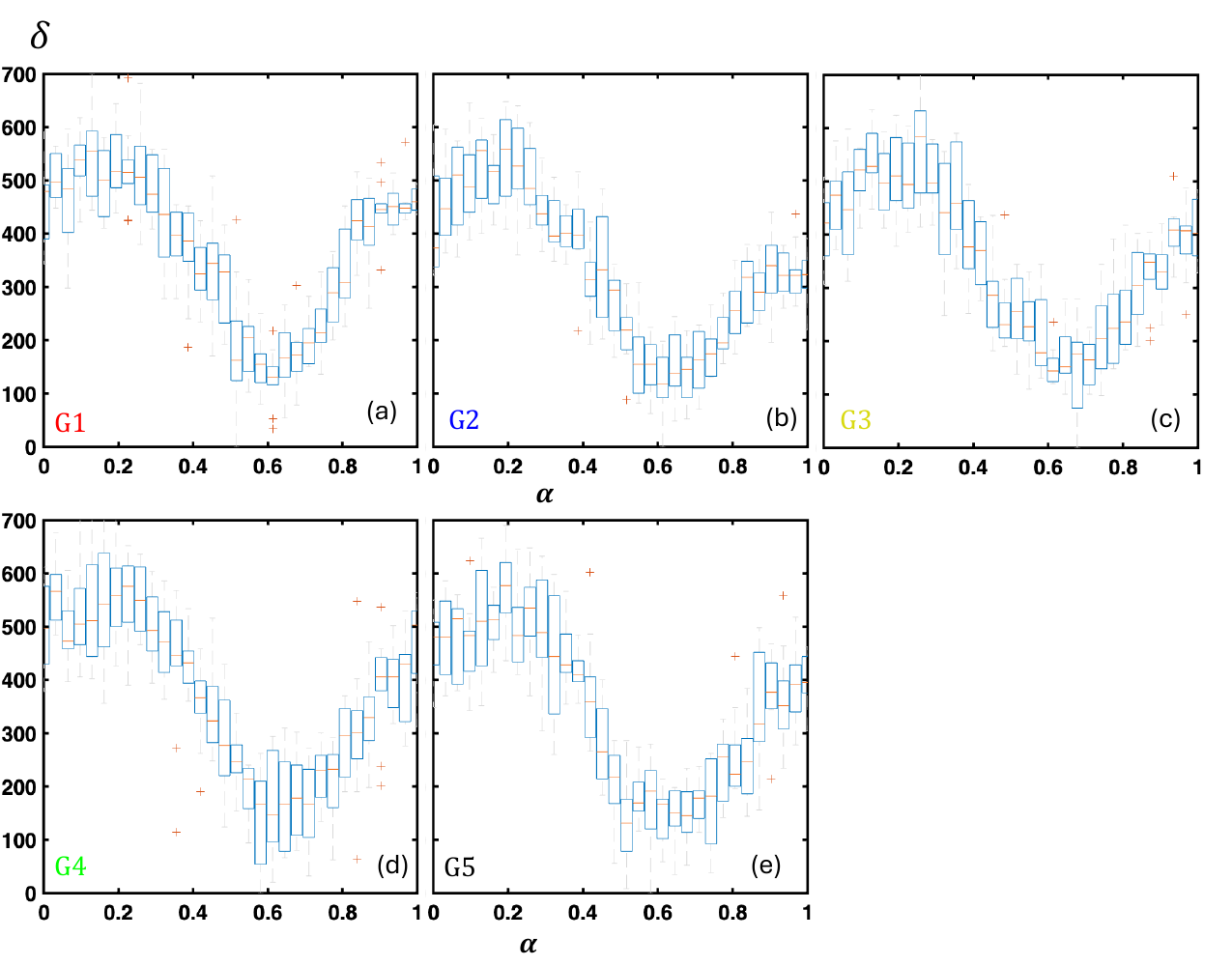}
    \caption{\justifying \smaller Individual box plots corresponding to the concatenated box-plot results presented in the main text. The box plots show the nonlinear simulation energy deviation,
\(\delta_r(\alpha)=H_r(\alpha)-H_{\min}(\alpha)\), where
\(H_r(\alpha)\) is the Ising energy obtained in the \(r^{\mathrm{th}}\) trial and
\(H_{\min}(\alpha)=\min_r H_r(\alpha)\) is the minimum energy across all trials for a given \(\alpha\). Panels (a)–(e) correspond to graphs G1–G5, respectively, each simulated over 10 independent trials.  
}
\label{fig:G1-G5}
\end{figure}
\label{appendix:A}

\clearpage

\section{Supplementary note 11}
\subsection*{Non-Negative Parameter Gap in Regularized Phase-Encoded Higher-Order Non-Quadratic Ising Gradient-Flows}
\noindent We consider a generalized higher-order Ising Hamiltonian of the form
\begin{equation}
    H_{\mathrm{HO}}(\boldsymbol{\sigma})
    =
    - \sum_{i<j} G_{ij}^{(2)} \sigma_i \sigma_j
    - \sum_{i<j<k} G_{ijk}^{(3)} \sigma_i \sigma_j \sigma_k
    - \sum_{i<j<k<l} G_{ijkl}^{(4)} \sigma_i \sigma_j \sigma_k \sigma_l
    - \cdots ,
\end{equation}
where \(G^{(k)}_{i_1\cdots i_k}\in\mathbb{R}\) denotes the coefficient of a
\(q\)-body interaction among the spins
\(\{\sigma_{i_1},\ldots,\sigma_{i_q}\}\). We consider these coefficients to be invariant under permutations of their indices, so that each unordered \(q\)-tuple represents a single interaction. Equivalently, we denote each \(q\)-body interaction by a hyperedge \(e=\{i_1,\ldots,i_q\}\), and let \(\mathcal{E}_q\) be the set of such \(q\)-body hyperedges. The Hamiltonian can then be written as
\begin{equation}
    H_{\mathrm{HO}}(\boldsymbol{\sigma})
    =
    -\sum_{q\ge 2}\sum_{e\in\mathcal{E}_q}
    G^{(q)}_e
    \prod_{i\in e}\sigma_i ,
\end{equation}
where \(\sigma_i\in\{\pm1\}\). To remain consistent with the convention used in the manuscript, we define \(W^{(q)}_e=-G^{(q)}_e\). We now consider two phase-encoded gradient-flow representations for higher-order Ising interactions.\\

\textbf{A. Balanced Phase Interactions.}
We first consider even-order interactions and then show how native odd-order interactions can be incorporated through a reference-spin embedding.\\

\textbf{Even-Order Interactions.}
As described in our prior work~\cite{Bashar2023}, higher-order Ising interactions can be represented using phase-encoded gradient flows. Consider an even-order hyperedge \(e=(i_1,\ldots,i_q)\), with \(q\in 2\mathbb{N}\), and define a signed incidence vector \(\epsilon_e\in\{0,\pm1\}^{N}\). The corresponding phase argument is
\begin{equation}
    \psi_e(\phi)=\epsilon_e^{\top}\phi .
\end{equation}
For the even-order interactions considered here, \(\epsilon_e\) is chosen to be balanced, namely
\begin{equation}
    \epsilon_e^{\top}\mathbf{1}=0.
\end{equation}
Equivalently, the phase argument contains equal numbers of positive and negative phase variables. For example, the fourth-order spin product \(\sigma_i \sigma_j \sigma_k \sigma_l\) is represented on the binary phase manifold by
\begin{equation}
    \sigma_i \sigma_j \sigma_k \sigma_l
    =
    \cos(\phi_i-\phi_j+\phi_k-\phi_l),
    \qquad \phi_i\in\{0,\pi\}.
\end{equation}
More generally, with the readout \(s_i=\cos\phi_i\), the higher-order cosine term exactly reproduces the corresponding spin product on the binary phase manifold, while away from this manifold it defines a continuous relaxation of the discrete interaction. For a fixed interaction order \(q\), we consider the higher-order phase energy
\begin{equation}
    E_{\mathrm{HO}}^{(q)}(\boldsymbol{\phi},K_s)
    =
    K\sum_{e\in\mathcal{E}_q}
    W_e^{(q)}\cos(\boldsymbol{\epsilon}_e^{\top}\boldsymbol{\phi})
    -
    \frac{K_s}{2}\sum_i\cos(2\phi_i),
\end{equation}
where \(\boldsymbol{\epsilon}_e\in\{0,\pm1\}^N\) specifies which phases enter the interaction and with what signs, so that \(\boldsymbol{\epsilon}_e^{\top}\boldsymbol{\phi}\) is the corresponding phase argument.
The associated gradient flow is
\begin{equation}
    \dot{\boldsymbol{\phi}}
    =
    -\nabla_{\boldsymbol{\phi}}E_{\mathrm{HO}}^{(q)}(\boldsymbol{\phi},K_s),
\end{equation}
such that
\begin{equation}
    \frac{dE_{\mathrm{HO}}^{(q)}}{dt}
    =
    -\left\|\nabla_{\boldsymbol{\phi}}E_{\mathrm{HO}}^{(q)}\right\|^2
    \le 0 .
\end{equation}
The second-harmonic term energetically favors the binary phase manifold
\(\phi_i\in\{0,\pi\}\), on which \(E_{\mathrm{HO}}^{(q)}\) reduces to the corresponding \(q^{\mathrm{th}}\) order Ising interaction up to an overall scale and an additive constant. The balanced condition also ensures that the trivial state
\(\boldsymbol{\phi}_{\rm triv}^{\ast}=(\pi/2)\mathbf{1}\)
is an equilibrium point of the dynamics since \(\epsilon_e^\top\phi_0=0\) for every hyperedge. To analyze the stability thresholds, define the interaction Jacobian matrix
\begin{equation}
    J_{\mathrm{HO}}^{(q)}(\boldsymbol{\phi})
    =
    \sum_{e\in\mathcal{E}_q}
    W_e^{(q)}
    \cos(\boldsymbol{\epsilon}_e^{\top}\boldsymbol{\phi})
    \boldsymbol{\epsilon}_e\boldsymbol{\epsilon}_e^{\top}.
\end{equation}
With this convention, the full Jacobian of the gradient flow is
\begin{equation}
    J_f^{(q)}(\boldsymbol{\phi},K_s)
    =
    KJ_{\mathrm{HO}}^{(q)}(\boldsymbol{\phi})
    -
    2K_s\operatorname{diag}(\cos 2\phi_i).
\end{equation}
At the trivial state \(\boldsymbol{\phi}_{\rm triv}^{\ast}=(\pi/2)\mathbf{1}\), the interaction Jacobian is
\begin{equation}
    J_{\mathrm{HO}}^{(q),\pi/2}
    =
    \sum_{e\in\mathcal{E}_q}
    W_e^{(q)}
    \boldsymbol{\epsilon}_e\boldsymbol{\epsilon}_e^{\top},
\end{equation}
while at a ground-state Ising phase configuration
\(\boldsymbol{\Phi}^{\ast}\in\{0,\pi\}^N\),
\begin{equation}
    J_{\mathrm{HO}}^{(q),\Phi^\ast}
    =
    \sum_{e\in\mathcal{E}_q}
    W_e^{(q)}
    \cos(\boldsymbol{\epsilon}_e^{\top}\boldsymbol{\Phi}^{\ast})
    \boldsymbol{\epsilon}_e\boldsymbol{\epsilon}_e^{\top}.
\end{equation}
The corresponding stability thresholds for this \(q\)-body interaction are therefore
\begin{equation}  
    \eta_{\mathrm{min}}^{(q)}
    =
    K_s^{(q),\pi/2}
    =
    -\frac{K}{2}
    \lambda_{\max}\!\left(J_{\mathrm{HO}}^{(q),\pi/2}\right),
    \qquad
    \eta_{\mathrm{max}}^{(q)}
    =
    \min_{\boldsymbol{\Phi}^{\ast}\in\mathcal{G}_{\phi}}
    \frac{K}{2}
    \lambda_{\max}\!\left(J_{\mathrm{HO}}^{(q),\Phi^\ast}\right),
\end{equation}
where \(\mathcal{G}_{\phi}\) denotes the set of ground-state binary phase configurations. The corresponding higher-order parameter gap is
\begin{equation}
    \Delta_{\mathrm{HO}}^{(q)}
    =
    \eta_{\mathrm{max}}^{(q)}
    -
    \eta_{\mathrm{min}}^{(q)} .
\end{equation}
To show that \(\Delta_{\mathrm{HO}}^{(q)}\ge0\), consider any ground-state configuration
\(\boldsymbol{\Phi}^{\ast}\in\mathcal{G}_{\phi}\) and define the corresponding configuration-wise quantity
\begin{equation}
    \Delta_{\mathrm{HO,bal}}^{(q)}(\boldsymbol{\Phi}^{\ast})
    =
    \frac{K}{2}
    \left[
    \lambda_{\max}\!\left(J_{\mathrm{HO}}^{(q),\Phi^\ast}\right)
    +
    \lambda_{\max}\!\left(J_{\mathrm{HO}}^{(q),\pi/2}\right)
    \right].
\end{equation}
It is sufficient to show that
\(\Delta_{\mathrm{HO}}^{(q)}(\boldsymbol{\Phi}^{\ast})\ge0\) for every
\(\boldsymbol{\Phi}^{\ast}\in\mathcal{G}_{\phi}\), since
\[
    \Delta_{\mathrm{HO}}^{(q)}
    =
    \min_{\boldsymbol{\Phi}^{\ast}\in\mathcal{G}_{\phi}}
    \Delta_{\mathrm{HO}}^{(q)}(\boldsymbol{\Phi}^{\ast}) .
\]
We define
\begin{equation}
    \Delta J_{\mathrm{HO}}^{(q)}
    =
    J_{\mathrm{HO}}^{(q),\Phi^\ast}
    +
    J_{\mathrm{HO}}^{(q),\pi/2}.
\end{equation}
Using the definitions of the interaction Jacobians,
\begin{equation}
    \Delta J_{\mathrm{HO}}^{(q)}
    =
    \sum_{e\in\mathcal{E}_q}
    W_e^{(q)}
    \left[
        1+\cos\!\left(\boldsymbol{\epsilon}_e^{\top}\boldsymbol{\Phi}^{\ast}\right)
    \right]
    \boldsymbol{\epsilon}_e\boldsymbol{\epsilon}_e^{\top}.
\end{equation}
Since each implemented hyperedge is balanced,
\(\boldsymbol{\epsilon}_e^{\top}\mathbf{1}=0\), we have
\begin{equation}
    \Delta J_{\mathrm{HO}}^{(q)}\mathbf{1}
    =
    \sum_{e\in\mathcal{E}_q}
    W_e^{(q)}
    \left[
        1+\cos\!\left(\boldsymbol{\epsilon}_e^{\top}\boldsymbol{\Phi}^{\ast}\right)
    \right]
    \boldsymbol{\epsilon}_e
    \left(\boldsymbol{\epsilon}_e^{\top}\mathbf{1}\right)
    =
    0 .
\end{equation}
Hence \(\Delta J_{\mathrm{HO}}^{(q)}\) has a zero eigenvalue. Since
\(\Delta J_{\mathrm{HO}}^{(q)}\) is symmetric,
\begin{equation}
    \lambda_{\max}\!\left(\Delta J_{\mathrm{HO}}^{(q)}\right)\ge 0 .
\end{equation}
By Weyl's inequality,
\begin{equation}
    \lambda_{\max}\!\left(J_{\mathrm{HO}}^{(q),\Phi^\ast}\right)
    +
    \lambda_{\max}\!\left(J_{\mathrm{HO}}^{(q),\pi/2}\right)
    \ge
    \lambda_{\max}\!\left(\Delta J_{\mathrm{HO}}^{(q)}\right)
    \ge 0 .
\end{equation}
Multiplying by \(K/2>0\), we obtain
\begin{equation}
    \Delta_{\mathrm{HO}}^{(q)}(\boldsymbol{\Phi}^{\ast})
    =
    \frac{K}{2}
    \left[
    \lambda_{\max}\!\left(J_{\mathrm{HO}}^{(q),\Phi^\ast}\right)
    +
    \lambda_{\max}\!\left(J_{\mathrm{HO}}^{(q),\pi/2}\right)
    \right]
    \ge 0 .
\end{equation}
Since this holds for every ground-state configuration
\(\boldsymbol{\Phi}^{\ast}\in\mathcal{G}_{\phi}\), taking the minimum over
ground-state configurations gives
\(\Delta_{\mathrm{HO}}^{(q)}\ge0.\)
Importantly, the above analysis accommodates both positive and negative hyperedge weights.\\

\textbf{Odd-Order Interactions.}
Native odd-order interactions can be incorporated by introducing a pinned reference spin
\(\sigma_0\). Consider an odd-order hyperedge \(e\) with \(q\in 2\mathbb{N}+1\). We define the augmented spin product by including the reference spin, 
\begin{equation}
    \sigma_0\prod_{i\in e}\sigma_i .
\end{equation}
Since \(q\) is odd, this is an even-order product over the augmented spin set \(\{\sigma_0\}\cup\{\sigma_i:i\in e\}\). If the reference spin is pinned to \(\sigma_0=+1\), the augmented product reproduces the original odd-order interaction; if \(\sigma_0=-1\), it reproduces the same interaction with the opposite sign. For example, a fifth-order product can be represented as
\begin{equation}
    \sigma_i \sigma_j \sigma_k \sigma_l \sigma_m
    \rightarrow
    \sigma_0 \sigma_i \sigma_j \sigma_k \sigma_l \sigma_m
    =
    \cos(\phi_0-\phi_i+\phi_j-\phi_k+\phi_l-\phi_m).
\end{equation}
The signs in the augmented phase argument can be chosen so that the argument is balanced, i.e., it contains equal numbers of positive and negative phase variables. Therefore, the same Jacobian-based zero-mode argument used for even-order balanced interactions applies in the augmented system. Consequently, the non-negative parameter-gap result extends to odd-order interactions after reference-spin embedding.\\

\textbf{B. Additive Phase Interactions.}
We also consider the higher-order analogue of the DIM-like phase-sum interaction. For example, on the binary phase manifold \(\phi_i\in\{0,\pi\}\),
\begin{equation}
    \sigma_i \sigma_j \sigma_k \sigma_l
    =
    \cos(\phi_i+\phi_j+\phi_k+\phi_l),
    \qquad \sigma_i=\cos\phi_i .
\end{equation}
This additive representation can be considered as the natural higher-order extension of the pairwise DIM interaction \(\cos(\phi_i+\phi_j)\). To reduce this representation to the balanced framework above, we introduce a reference phase \(\phi_0\). For a \(q\)-body hyperedge \(e\), we replace
\begin{equation}
    \cos\!\left(\sum_{i\in e}\phi_i\right)
    \quad\longrightarrow\quad
    \cos\!\left(\sum_{i\in e}\phi_i-q\phi_0\right).
\end{equation}
Let \(b_e^{(q)}\in\{0,1\}^{N}\) denote the incidence vector of the native hyperedge and define its augmented version
\[
    \bar b_e^{(q)}=(0,b_e^{(q)})\in\{0,1\}^{N+1}.
\]
With \(\tilde\phi=(\phi_0,\phi_1,\ldots,\phi_N)^\top\) and \(e_0=(1,0,\ldots,0)^\top\), define
\begin{equation}
    \tilde b_e^{(q)}
    =
    \bar b_e^{(q)}-q e_0 .
\end{equation}
Then
\begin{equation}
    \left(\tilde b_e^{(q)}\right)^\top\tilde\phi
    =
    \sum_{i\in e}\phi_i-q\phi_0,
    \qquad
    \left(\tilde b_e^{(q)}\right)^\top\mathbf{1}
    =
    q-q=0 .
\end{equation}
Thus, the reference-phase form is a balanced phase interaction in the augmented phase space. On the binary phase manifold,
\begin{equation}
    \cos\!\left(\sum_{i\in e}\Phi_i-q\Phi_0\right)
    =
    \sigma_0^q\prod_{i\in e}\sigma_i .
\end{equation}
Therefore, pinning \(\sigma_0=+1\) recovers the original additive spin product for both even and odd \(q\). The reference phase thus both balances the additive phase argument and accommodates arbitrary interaction order. Since the reference-phase embedding maps the additive interaction into the balanced class, the Jacobian zero-mode and Weyl-inequality argument established above applies directly, with \(\epsilon_e\) replaced by \(\tilde b_e^{(q)}\). Consequently, the reference-phase parameter gap for additive higher-order phase-sum interactions satisfies
\(
    \Delta_{\mathrm{HO},+}^{(q)}\ge 0 .
\)\\

\end{document}